\documentclass[a4paper,reqno]{amsart}

    \usepackage[english]{babel} 

    \usepackage{graphicx}
    \usepackage[]{color}  
    \usepackage[dvipsnames]{xcolor} 
    \usepackage{csquotes} 
    \usepackage{appendix} 
    \usepackage[a4paper]{geometry} 
    \usepackage{cancel} 
    \usepackage{comment}

    \usepackage{amsmath,amssymb,amsfonts,amsthm} 
    \usepackage{mathtools} 
    \usepackage{tikz-cd}  
    \usepackage{tikz}
        \usetikzlibrary{arrows}
    \usepackage[all,cmtip]{xy} 
    \usepackage[bbgreekl]{mathbbol} 
    \usepackage{array} 
    \usepackage{stmaryrd}
    \usepackage{array} 
    \usepackage{float} 
    \usepackage{hyperref} 
    
\theoremstyle{definition} 
    \newtheorem{definition}{Definition}
    \newtheorem{assumption}[definition]{Assumption}

\theoremstyle{plain} 
    \newtheorem{theorem}[definition]{Theorem}
    \newtheorem{proposition}[definition]{Proposition}
    \newtheorem{lemma}[definition]{Lemma}
    \newtheorem{corollary}[definition]{Corollary}
    
    \newtheorem{calculation}[definition]{Calculation}

\theoremstyle{remark} 
    \newtheorem{remark}[definition]{Remark}

\usepackage[bibstyle=alphabetic,citestyle=alphabetic,useprefix,giveninits=true, sorting=ynt, sortcites, minbibnames=99,maxbibnames=99,backend=biber]{biblatex}  
\renewbibmacro{in:}{} 
\bibliography{bibliography}  
\DeclareSourcemap{
 \maps[datatype=bibtex]{
   \map[overwrite]{ 
     \step[fieldsource=doi, final]
     \step[fieldset=url, null]
     \step[fieldset=eprint, null]
     }
    \map[overwrite]{ 
     \step[fieldsource=eprint, final]
     \step[fieldset=pages, null]
     \step[fieldset=eid, null]
     \step[fieldset=journal, null]
   }  
 }
}

    \newcommand{\wt}[1]{\widetilde{#1}}
     \newcommand{\tc}{c}
     \newcommand{\tom}{\omega}
     \newcommand{\te}{e}
     \newcommand{\txi}[1]{\zeta^{#1}}
    
     \newcommand{\comp}[1]{\langle #1 \rangle}

     \newcommand{\filt}[1]{(#1)}

    \newcommand{\intl}{\int\limits}
    \newcommand{\trintl}[1]{\mathrm{Tr}\intl_{#1}}

    \newcommand{\phip}{\varphi}
    \newcommand{\phin}{\varphi_n}

\usepackage{accents}

\newcommand{\BFV}{\mathsf{BFV}}
\newcommand{\F}{\mathbb{F}}
\newcommand{\C}{\mathbb{C}}
\newcommand{\FF}{\mathbb{FF}}
\newcommand{\Co}{C_\circ}
\newcommand{\CCo}{\C_\circ}
\newcommand{\Cres}{C_{\mathrm{res}}}

    \definecolor{light-gray}{gray}{0.95}
    \usepackage[color=light-gray]{todonotes}

\title{Double BFV Quantisation of 3d Gravity}
\author{Giovanni Canepa}
\address{INFN, Sezione di Firenze, via Sansone 1, 50019 Sesto Fiorentino (FI), Italy}
\email{giovanni.canepa.math@gmail.com}
\author{Michele Schiavina}
\address{Department of Mathematics, University of Pavia, via Ferrata 5, 27100, Pavia.}
\address{INFN Sezione di Pavia, via Bassi 6, 27100 Pavia, Italy}
\email{michele.schiavina@unipv.it}

\begin{document}
\thanks{G.C.\ acknowledges partial support of SNF Grants No P500PT 203085 and No P5R5PT 222221. M.S.\ would like to thank INDAM, Instituto Nazionale: di Alta Matematica, and its Mathematical Physics research group GNFM.
This material is partially based upon work supported by the Swedish Research Council under grant no. 2021-06594 while G.C. was in residence at Institut Mittag-Leffler in Djursholm, Sweden, during the spring semester of 2025.
}

\begin{abstract}
    We extend the cohomological setting developed by Batalin, Fradkin and Vilkovisky (BFV), which produces a resolution of coisotropic reduction in terms of hamiltonian dg manifolds, to the case of nested coisotropic embeddings $C\hookrightarrow C_\circ \hookrightarrow F$ inside a symplectic manifold $F$.
    To this, we naturally assign $\underline{C}$ and $\underline{C_\circ}$, as well as the respective BFV dg manifolds.
    
    We show that the data of a nested coisotropic embedding defines a natural graded coisotropic embedding inside the BFV dg manifold assigned to $\underline{C}$, whose reduction can further be resolved using the BFV prescription.
    We call this construction \emph{double BFV resolution}, and we use it to prove that ``resolution commutes with reduction'' for a large class of nested coisotropic embeddings. 
    
    We then deduce a quantisation of $\underline{C}$, from the (graded) geometric quantisation of the double BFV Hamiltonian dg manifold (when it exists), following the quantum BFV prescription. 
    
    As an application, we provide a well defined candidate space of (physical) quantum states of three-dimensional Einstein--Hilbert theory, which is thought of as a partial reduction of the Palatini--Cartan model for gravity.
\end{abstract}

\maketitle

\tableofcontents

\section{Introduction}

Quantisation of a mechanical system is a procedure that associates to a symplectic manifold endowed with a Hamiltonian flow, some space of states (typically Hilbert) with a distinguished (possibly unbounded) Hermitian operator over it \cite{dirac1981principles}. The history of quantisation techniques is long, and rich with more or less successful attempts, including some crucial no-go theorems that obstructed the original axiomatic approach proposed by Dirac \cite{Groenewold,VanHowe} (see also \cite{GotayObstructions}).

To overcome these obstructions, in geometric quantisation, one looks at a suitably chosen subalgebra of the algebra of functions over a sympletic manifold (thought of as the phase space of a mechanical system) and builds a quantisation map (a Lie algebra morphism) onto some suitable set of operators over a (possibly) Hilbert space, after choosing a polarization \cite{Kos73,souriau1970,Woodhouse,Gotayfunctorial}. 

When the mechanical system is constrained, i.e.\ the dynamics is tangent to a coisotropic submanifold $C\subset F$, one has the option of quantising either $F$ or the coisotropic reduction $\underline{C}$. In the physics literature, the quantisation of $F$ is often called the ``kinematical'' space of states, while the quantisation of $\underline{C}$ has the qualifier ``physical'', for physical configurations are only those which lie in the constraint set. The problem of comparing these two, also phrasable as the implementation of ``quantum constraints'', is known in symplectic geometry as the ``quantisation commutes with reduction paradigm'', originally conjectured in general and proven for Hamiltonian actions on K\"ahler manifolds under certain assumptions in \cite{GuilleminSternberg}, and later generalised by \cite{MeinrenkenRR,VergneMultiplicities,MeinrenkenSpinDirac,MeinrenkenSjamaar}. (A review of these progress following Guillemin and Sternberg's paper was written by Vergne \cite{VergneReview}.) 

Another approach - which is more suitable to applications in field theory - follows the idea that instead of pursuing the quantisation of the reduced space $\underline{C}$, which may be singular and unwieldy,\footnote{A reasonable intermediate stage is when the reduction is a smooth stratified manifold.} one can instead look at a coisotropic resolution of the quotient, in terms of the Koszul--Tate--Chevalley--Eilenberg complex or, crediting the pioneers of this method, the Batalin--Fradkin--Vilkovisky complex \cite{BV1,BV2,BV3}. (See \cite{Stasheff1997} for an explanation of the method that is more amenable to mathematicians.) A more recent viewpoint on this approach is presented in \cite{OhPark} and \cite{Schaetz:2008}, where it is shown that the BFV complex (or its underlying strong homotopy Lie algebra), controls deformations of coisotropic embeddings.

The BFV method is not only a way to resolve reductions, but it also proposes an alternative way to construct a space of physical quantum states: Instead of looking at the geometric quantisation of $\underline{C}$, one quantises a Hamiltonian dg manifold to a ``quantum" dg vector space\footnote{Observe that, having components in nonzero degree, $\mathbb{V}$ is not going to be a Hilbert space, since there must be null-norm vectors. One can sometimes give it a Krein space structure.} $(\mathbb{V},\mathbb{\Omega})$, and the physical space of states is its cohomology in degree zero $\mathbb{V}_{\mathrm{phys}}\doteq H^0(\mathbb{V},\mathbb{\Omega})$. The relation between BFV quantisation and geometric quantisation is explained in \cite{FradkinLinetsky}.

When one is looking at a nested embedding, i.e.\ a diagram $C\hookrightarrow \Co \hookrightarrow F$ where $C$ and $\Co$ are both coisotropic in $F$, one can perform coisotropic reduction in two steps, as $C$ descends to a coisotropic submanifold $\Cres$ in the (partial) reduction $\underline{\Co}$. In the case of group actions, where $C$ and $\Co$ are, for instance, the zero level sets of equivariant momentum maps for the action of (respectively) a group $G$ and a subgroup $G_\circ$ this is called Hamiltonian reduction by stages. Note, however, that in our general scenario we do not need to require that {the vanishing ideal $\mathcal{I}_{\Co}$ of} $\Co$ be a Poisson ideal of the vanishing ideal\footnote{This is {encoded by $\{\phi,\psi\}\propto \phi$ for any two generators $\phi\in \mathcal{I}_{\Co}$ and $\psi\in\mathcal{I}_C$. An example of this more restrictive scenario arises in Hamiltonian reduction by stages, where $C$ is the zero level set of a $G$-momentum map and $\Co$ is the zero level set of the momentum map associated to a normal subgroup $G_\circ\subset G$.}} $\mathcal{I}_C$ of $C$. This picture is especially relevant when one is concerned with the problem of quantisation and reduction of the residual embedding $\Cres\hookrightarrow\underline{\Co}$, assumed smooth. {(Meaning that one has access to a space $X$ and realises it as a partial reduction in a nested coisotropic embedding, $X\simeq \underline{\Co}$ in order to quantise it.)}

In particular, our main application---the quantisation of the reduced phase space of three dimensional general relativity---is precisely of this kind. Indeed, the Hamiltonian approach to three dimensional gravity in the Einstein--Hilbert (EH) formulation produces a coisotropic embedding in a symplectic manifold $\mathcal{C}_{EH}\hookrightarrow \mathcal{F}_{EH}$, known in the literature as the canonical constraint set of gravity \cite{Katz,deWitt}, which can be shown to be (symplectomorphic to) 
the (partial) coisotropic reduction of the constraint set of the Hamiltonian data of gravity in the Palatini--Cartan (PC) formulation \cite{CS2019,CCS2020,CaSc2019}. More precisely, there exists a nested coisotropic embedding 
\begin{equation}\label{e:PCnestedEmbedding}
    \mathcal{C}_{PC}\hookrightarrow \mathcal{C}_{\mathrm{Lor}}\hookrightarrow \mathcal{F}_{PC}
\end{equation} {and a space $\mathcal{C}_{PC,\mathrm{res}}\hookrightarrow \underline{\mathcal{C}_{\mathrm{Lor}}}$ such that $\mathcal{C}_{PC,\mathrm{res}}\simeq \mathcal{C}_{EH}$ and \cite{CS2019}
\begin{equation}\label{e:SummaryEHPCred}
\mathcal{C}_{EH}\simeq \mathcal{C}_{PC,\mathrm{res}}\hookrightarrow \underline{\mathcal{C}_{\mathrm{Lor}}} \simeq \mathcal{F}_{EH} \quad \leadsto \quad \underline{\mathcal{C}_{EH}}\simeq \underline{\mathcal{C}_{PC,\mathrm{res}}}.  
\end{equation}}
{It is important to note that the nested inclusion appearing in general relativity that we study here is \emph{not} a standard example of Hamiltonian reduction by stages. In fact, while the submanifolds $\mathcal{C}_{PC}$ and $\mathcal{C}_{\mathrm{Lor}}$ are given by the zero level set of equivariant momentum maps for the action of Lie groups $\mathcal{G}\supset \mathcal{G}_\circ$ (given by the semidirect product of diffeomorphisms with Lorentz transformations), crucially $\mathcal{G}_\circ$ (Lorentz transformations) is not normal.}

Hence, to find a quantisation of the coisotropic reduction $\underline{\mathcal{C}_{EH}}$---which would yield the desired space of physical quantum states---one can follow various alternative procedures. One can try to quantise it as a (possibly singular) symplectic manifold. This is the most direct path, but also the most difficult, owing to the nonlocal nature of the (currently known) ways to describe the constraint set,\footnote{See for instance \cite{RielloSchiavina1,RielloSchiavina2} for the example of Yang--Mills theory.} as well as its singular behaviour. One can alternatively try to quantise $\mathcal{F}_{EH}$ and find a quantum version of the constraints. This approach has a long history and rich literature in quantum gravity. See for instance \cite{deWitt,Thiemann_2007,DelcampDittrichRiello,BonzomDupuisGirelli}. Generally speaking, finding solutions to the quantum constraints is not easy.

In this paper we provide a third alternative, that uses the full extent of the relation between Einstein--Hilbert and Palatini--Cartan theory, as well as the strong equivalence between the dg models for the reduced phase space of PC theory and another topological model, called $BF$ theory \cite{CaSc2019}. 

Namely, as mentioned, the phase space of Einstein--Hilbert theory $\mathcal{F}_{EH}$ with its set of canonical constraints $\mathcal{C}_{EH}$ is obtained as the partial coisotropic reduction of the phase space of PC theory. 
Using the double BFV construction for the nested coisotropic embedding (Eq.\ \eqref{e:PCnestedEmbedding}), we implement the partial coisotropic reduction within the BFV framework: {there is a dg coisotropic submanifold $\mathbb{C}_{\mathrm{Lor}}\hookrightarrow \F_{PC}$ inside the BFV resolution $\F(\mathcal{F}_{PC},\mathcal{C}_{PC})$ (a symplectic dg manifold) of the degree-zero coisotropic embedding $\mathcal{C}_{PC}\hookrightarrow \mathcal{F}_{PC}$. We then exploit the fact that the reduction of the dg coisotropic submanifold $\mathcal{C}_{\mathrm{Lor}}$ can itself be resolved by BFV methods, to show explicitly that one can obtain the BFV resolution $\F(\mathcal{F}_{EH},\mathcal{C}_{EH})$ of the Reduced Phase Space $\underline{\mathcal{C}_{EH}}$ of EH theory as a dg-coisotropic reduction $\underline{\mathbb{C}_{\mathrm{Lor}}}$ inside the BFV resolution $\F(\mathcal{F}_{PC},\mathcal{C}_{PC})$ of the reduced phase space $\mathcal{C}_{PC}$ of PC theory. Namely, there is a morphism of Hamiltonian dg manifolds:
\[
\underline{\sigma}_{\BFV}\colon \underline{\mathbb{C}_{\mathrm{Lor}}} \to \F(\mathcal{F}_{EH},\mathcal{C}_{EH}).
\]
}

Since, furthermore, it was shown in \cite{CaSc2019} that the BFV resolution of the reduced phase space of PC theory is dg symplectomorphic to the BFV resolution of the reduced phase space of $BF$ theory, for which an explicit quantisation was obtained by Cattaneo, Mnev and Reshetikhin in \cite{CMR2}, we will use such quantisation of $BF$ theory as a starting point to build a quantisation of 3d PC theory. This, combined with our double BFV construction, will output a well-defined quantization of 3d Einstein Hilbert theory, in the sense of a space of (physical) quantum states.

We present two such quantisations, essentially determined by a choice of polarisation in the classical (i.e.\ degree zero) phase space, which can be thought of as the choice of position vs momentum representation, where position and momentum are represented by connection one forms and nondegenerate triads. While it is customary to think of the space of quantum states for GR as associated to quantising metrics, it turns out that choosing the opposite polarisation results in a simpler quantisation procedure and output. This shows that one can indeed consider the space of states of GR as functions over the space of connections, with the flatness constraint enforced via the quantum BFV operator. This appears to be possible even restricting to the nondegenerate sector, i.e.\ when triads are required to be nondegenerate. See Section \ref{sec:nondegeneracy}. 

For completeness we also showcase the output of our construction for the ``opposite'' polarisation, whose output is much less legible. This should be compared with the standard Wheeler-deWitt quantisation of Einstein--Hilbert gravity. While a direct comparison is not immediately possible, we argue that the symplectomorphisms we establish at the classical level between the relevant theories (namely EH theory and the partially reduced PC theory, see Section \ref{s:relation_Classical_theories} and, in particular, Theorem \ref{t:BF_coisotropic_EH}) are a strong indication that the three quantisations should be all equivalent in some appropriate sense. We leave this to further work.

This article is structured as follows. In Section \ref{s:coiso_by_BFV} we introduce the classical BFV construction for the resolution of coisotropic submanifolds and in Section \ref{s:quantumBFVstandard} we outline its quantisation. In Section \ref{s:doubleBFV} we describe the nested coisotropic embeddings, deduce some properties and show how to describe them with a double application of the BFV formalism. We also prove in Theorem \ref{thm:doubleBFVreduced} that this resolution commutes with reduction. Then, in Section \ref{s:quantisation_BFVonBFV}, we quantise the double BFV construction. Note that the results of these sections are proven in the finite dimensional setting.  In section \ref{s:example_linearconstraints}, we then consider a simple example with linear constraints as a warm up for the gravity case.

In Section \ref{s:BFVtheories_gravity} we present the BFV formulations of 3d gravity in the Einstein--Hilbert formalism (Section \ref{s:BFVtheory_EH}) and Palatini--Cartan formalism (Section \ref{s:BFVtheory_PC}), and the BFV formulation of $BF$ theory  (Section \ref{s:BFVtheory_BF}). Then, in Section \ref{s:relation_Classical_theories} we explicitly show how to obtain the BFV formulation of Einstein--Hilbert theory as a partial coisotropic reduction of the one of $BF$ theory. This (partial) reduction is then described within the double BFV framework in Section \ref{s:BFVonBFV_Classical}, adapted to the case at hand. Finally, in Section \ref{s:quantumBF_doubleBFV}, the quantisation scheme for double BFV theories we outlined in Section \ref{s:quantisation_BFVonBFV} is applied to this case and the main result is summarized in Theorem \ref{thm:quantisationEH_doubleBFV}.

\section*{Acknowledgements}
We thank Alberto S.\ Cattaneo, Jonas Schnitzer and Konstantin Wernli for useful comments, suggestions and discussions relevant to the making of this paper. 

\section{BFV resolution of coisotropic submanifolds} 
A key construction used in this work is the cohomological resolution of the reduction of coisotropic submanifolds. We recap here the construction and refer to  \cite[Section 5.1]{CCS2020} for more insight. (See also \cite{Stasheff1997,OhPark,Schaetz:2008,SchaetzTH}.)

\subsection{Classical BFV formalism}\label{s:coiso_by_BFV} 

Let $(F,\omega)$ be a (possibly graded) $0$-symplectic manifold. Furthermore let $\{\phi_i\}_{i\in I} \in C^{\infty}(F)$ be a collection of independent Hamiltonian functions.\footnote{These are often called \emph{constraints}.} Let us now assume that the subset\footnote{In the applications we will be interested in, it is generally too much to require that $C$ be a smooth embedded submanifold. The algebraic reduction is however still available in singular scenarios, and the cohomological resolution is compatible with an algebraic description of coisotropic sets.} $C$, defined by  
\[
C=\{x\in F:\phi_i(x)=0\ \forall i\in I\}\equiv\mathrm{Zero}\{\phi_i\},
\]
is coisotropic, meaning that\footnote{The Poisson brackets are defined on Hamiltonian function via the symplectic form: $\{f,g\}\doteq \iota_{X_f}\iota_{X_g}\omega$. With this prescription this definition is equivalent to the standard definition of coisotropic submanifold.} $\{\phi_i,\phi_j\}=f_{ij}^k \phi_k$ (where the $\{f_{ij}^k\}$'s are functions on $F$).

By coisotropic reduction we mean the space of leaves $\underline{C}$ of the characteristic foliation of $C\subset (F,\omega)$, generated by the Hamiltonian vector fields of the $\{\phi_i\}$'s (which are involutive because of the coisotropicity condition). At the level of functions, this means that 
\begin{align*}
   C^\infty(\underline{C})\simeq \left(C^\infty(F)/\mathcal{I}_C\right)^{\mathcal{I}_C} 
\end{align*}
 is a Poisson algebra, where $\mathcal{I}_C$ is the vanishing ideal of $C$, which is a Poisson subalgebra by assumption. In other words, (the smooth locus of) $\underline{C}$ is equipped with a symplectic form $\underline{\omega}$ defined such that its pullback along $\pi_C\colon C \to \underline{C}$ is $\omega|_C=\pi_C^*\underline{\omega}$.  

The BFV construction---after Batalin, Fradkin and Vilkovisky---provides a cohomological description of $\underline{C}$, i.e.\ it outputs a cochain complex $\mathfrak{BFV}^\bullet$ whose cohomology in degree zero is (a replacement of) $C^\infty(\underline{C})\simeq H^0(\mathfrak{BFV}^\bullet)$.

The construction goes through by enlarging $(F,\omega)$ to a graded $0$-symplectic manifold $\F(M,C)$ endowed with a cohomological vector field $Q$ such that $\mathfrak{BFV}^\bullet\doteq C^\infty(\F(M,C))$ is a complex with $Q$ as differential.

The degree zero cohomology of this complex then describes the space of functions on $\underline{C}$, with $C\subset F$, since $Q$ is constructed to be a combination of the Koszul--Tate complex associated to $\mathcal{I}_C\subset C^\infty(F)$, and the Chevalley--Eilenberg complex associated to the canonical Lie algebroid structure arising from $C\subset F$ being coisotropic (see e.g. \cite{CattaneoFelderRelative}). 

Explicitly, when $C=\mathrm{Zero}\{\phi_i\in C^\infty(F)\}$ one considers a Lagrange multiplier $c_i$ of degree $\mathrm{deg}(c_i)=1-\mathrm{deg}(\phi_i)$ for each constraint. If we call $W$ the space of Lagrange multipliers\footnote{In finite dimensions this can be just taken to be some $\mathbb{R}^k$. In field theory this is some vector space $W$. We keep the notation for further reference.}, we then define the new space of BFV fields to be\footnote{In our applications we shall consider spaces of sections of fibre bundles, so the cotangent bundle is truly defined as the space of densitised sections of the dual (tangent) bundle. Here the discussion considers a finite-dimensional toy model.} 
$\F(F,C)\doteq F\times T^* W[1]$ where the variables on the fiber of $T^* W[1]$ are denoted by $c_i^\dag$, and are of degree $\mathrm{deg}(c_i^\dag)= -\mathrm{deg}(c_i)$. The canonical symplectic form associated to this space is then\footnote{In the field theory scenario, this is a local, weak-symplectic form. The differential $d$ is then replaced with the vertical differential $\delta$, which can be thought of de Rham on a space of sections of a vector bundle.} 
\begin{align*}
    \wt{\omega} = \omega + d c^\dagger_i\, d c^i.
\end{align*}

\begin{remark}\label{rmk:generalBFV}
    Note that, whenever the Koszul complex is not acyclic, one can extend the procedure by introducing higher (negative) degree variables whose differential are nontrivial Koszul cocycles, and thus obtain a new, acyclic, complex. See \cite{Stasheff1997} for a complete description of the procedure. We will assume here that the procedure stops, but in principle we could ask $W$ to be a (nontrivial) graded vector space itself, and still apply the formulas provided here, with a little extra care.   
\end{remark}

By construction, $\F(F,C)$ comes equipped with a cohomological vector field $Q$, which is a geometrised version of the Koszul--Tate--Chevalley--Eilenberg differential. Another way of seeing it emerge is as follows: Let us look for a function $S \in C^{\infty}(\F)$ as 
\begin{align}\label{e:Rfunction}
    S =  c^i\phi_i + \frac12 f_{ij}^k c^\dagger_k c^ic^j + r
\end{align}
where $r$ is a function of higher homogeneous degree in the variables $c^\dag_i$ such that $\{S,S\}=0$ (the BFV master equation). Let now $Q$ be the Hamiltonian vector field of $S$. Then it can be shown that the cohomology of $Q$ in degree zero is isomorphic to 
$$C^\infty(\underline{C})\simeq(C^\infty(F)/\mathcal{I}_C)^{\mathcal{I}_C}$$ 
as a Poisson algebra, where $\mathcal{I}_C=\mathrm{Span}_{C^\infty(F)}\{\phi_i\}$ is the vanishing ideal of $C$ and the invariance is w.r.t. the Hamiltonian vector fields  of $\phi_i$ \cite{BV1981,Batalin:1983,Stasheff1997}. In other words, the Koszul--Tate--Chevalley--Eilenberg differential is the derivation associated to the cohomological vector field $Q$, which is Hamiltonian for the master functional $S$.

\begin{remark}[Notation]
    In what follows we will need to distinguish the BFV resolution of several different submanifolds, and as such we use a particular notation that keeps track of the various elements of the construction, as detailed by the following Definition.
\end{remark}

\begin{definition}\label{d:BFVtheory_of_C}
    The BFV data associated to the coisotropic submanifold $C\subset F$ is
    \[
    \BFV(F,C) = (\F(F,C), \omega_C, Q_C, S_C).
    \]
    This defines the BFV complex
    \[
    Q_C\circlearrowright\mathfrak{BFV}^\bullet(F,C) = C^\infty(\F(F,C)).
    \]
\end{definition}

\begin{remark}[Graded coisotropic]\label{rmk:dgcoiso}
Notice that we left the option open for $(F,\omega)$ to be graded symplectic. In particular, we will interested in cases where $(F,\omega,Q)$ is itself a dg manifold and $C\hookrightarrow F$ is coisotropic. One could forget the compatibility with $Q$ and simply require that $C\hookrightarrow F$ be coisotropic in the sense of graded manifold, i.e.\ its vanishing locus $\mathcal{I}_C$ is an associative ideal which is also a subalgebra in the graded Poisson algebra of functions $C^\infty(M)$. However, in the examples we will consider, a stronger notion will emerge naturally. Namely we ask $\mathcal{I}_C$ to be also a differential ideal w.r.t.\ $Q$, or in other words, $Q$ is tangent to $C$ as a vector field.
\end{remark}

\subsection{Quantum BFV formalism}\label{s:quantumBFVstandard}
Quantisation in the BFV formalism involves performing some quantisation of the graded symplectic manifold $(\F(F,C),\omega_C)$. In the original approach \cite{BatalinFradkin}, one constructs an algebra of quantum observables starting from all functions of the BFV space of fields, and then implements gauge fixing by further extending the BFV space of fields (this is analogous to the addition of fields to implement gauge fixing fermions within BV quantisation). However, one can equivalently approach the problem via geometric quantisation, following \cite{CMR2}. (See also \cite{FradkinLinetsky} for a review of the link between BFV quantisation and geometric quantisation.)

One chooses a polarisation---i.e., a foliation by Lagrangian submanifolds---and defines a space of polarised sections of a (possibly trivial) line bundle, which leads to a graded vector space $V_C$ (polarised sections of the line bundle), as well as a ``quantisation map'', i.e.\ a Lie algebra morphism 
\[
\mathfrak{q}\colon (\mathcal{O}, \{\cdot,\cdot\}) \to (\mathrm{End}(V_C), [\cdot,\cdot])
\]
where $\mathcal{O}\subset C^\infty(\F(F,C))$ is some appropriate Poisson subalgebra of ``quantisable'' functions. (Observe that, even assuming a symmetric pairing is given on $V_C$, the graded nature of the symplectic manifold will make it so that odd sections will have zero norm, so this vector space is not going to be ``Hilbert''.)

Assuming that one can quantise\footnote{In other words if $S_C\in \mathcal{O}$, which will be the case in our examples.} the BFV action $S_C$ to an operator $\Omega_C$, bracketing with it yields a differential $[\Omega_C,\cdot]$ in $\mathrm{End}(V_C)$,  making $(V_C,\Omega_C)$ into a cochain complex,\footnote{Observe indeed that $[\Omega_C,\Omega_C] = [\mathfrak{q}(S_C)),\mathfrak{q}(S_C))] = \mathfrak{q}(\{S_C,S_C\}) = 0$.} The quantum BFV complex. The cohomology in degree zero of this complex is then taken to be a quantisation of $\underline{C}$. 

Note that this touches upon the question ``does quantisation commute with reduction?'', which would relate the quantisation of $(F,\omega)$ to that of $(\underline{C},\underline{\omega})$, but looks at it from a slightly different point of view. We wish to relate a `classical resolution' to a `quantum resolution': i.e.\ a cohomological framework such that the true quantum (Hilbert) space of states --- which would be the quantisation of the space $\underline{C}$ --- is recovered as the cohomology in degree zero of the quantum BFV complex. In terms of the usual slogan, the question here changes to: ``does quantisation commute with resolution?''. We answer this question positively in a number of cases with Theorem \ref{thm:doubleBFVreduced}.

\subsection{Double BFV}\label{s:doubleBFV}
In this article we consider two combined applications of the BFV formalism. Apart from the ``usual'' resolution of coisotropic reduction in ordinary, i.e.\ non-graded, symplectic manifolds, we will also consider the scenario in which a \emph{graded} $0$-symplectic manifold---possibly representing the BFV space for some \emph{other} coisotropic reduction---is itself endowed with a (graded) coisotropic submanifold, and we shall approach \emph{that} reduction via BFV resolution as well. (Note that the explicit graded-coisotropic reduction, in the particular example we consider, is performed in Section \ref{s:relation_Classical_theories}.)
    
Indeed, one can consider a graded symplectic manifold obtained as the BFV resolution of an ordinary coisotropic reduction, and look at a coisotropic submanifold therein. In our particular example, this will correspond to a \emph{partial reduction} of the original, degree $0$, coisotropic submanifold. In other words, will be performing a BFV resolution over another BFV resolution. Let us be more explicit.

Consider the following diagram of inclusion of submanifolds inside a symplectic manifold $(F,\omega_F)$:
\[
\xymatrix{
C_\psi \ar[r]^{\iota_\psi} & F \\
C \ar[r] \ar[u] & C_\phi \ar[u]^{\iota_\phi}
}
\]
where we assume that
\begin{align*}
    &\begin{cases}
    C_\phi \doteq \mathrm{Zero}\{\phi_i, i=1,\dots, \ell_1\},\\
    C_\psi \doteq \mathrm{Zero}\{\psi_j, j=\ell_1+1,\dots, \ell_2\},
\end{cases}\\
&C\equiv C_{\phi}\cap C_{\psi}= \mathrm{Zero}\{\phi_i, \psi_j, i=1\dots \ell_1, j=\ell_1+1,\dots, \ell_2\} {\simeq \mathrm{Zero}\{\iota^*_\phi\psi_j\}\subset C_\phi}.
\end{align*}

We ask that all three submanifolds $C_\phi, C_\psi$ and $C$ be coisotropic and, for the purposes of this paper we assume the (symplectic) reduction of the coisotropic submanifold $C_\phi$,  denoted  by
\[
\pi_\phi\colon C_\phi \to \underline{C_\phi}, \qquad \iota_\phi\colon C_\phi\hookrightarrow F,
\]
to be smooth. {We further ask that the functions $\psi_j$ be compatible with the characteristic foliation of $C_\phi$ in the following sense: we ask that their restriction to $C_\phi$ define an (associative) ideal $\mathcal{I}_{\iota^*_\phi\psi}$ that is invariant w.r.t. the action of the Hamiltonian vector fields of the functions $\phi$: i.e.\ $X_{\phi_i}(\mathcal{I}_{\iota^*_\phi\psi})=\mathcal{I}_{\iota^*_\phi\psi}$} for all $\phi_i$. {We can then assume that} that there exist $\underline{\psi}{}_j\in C^\infty(\underline{C_\phi})$ such that the ideal generated by the $\underline{\psi}$ is a Lie subalgebra of $C^\infty(\underline{C_\phi})$ and it is isomorphic to the associative ideal\footnote{This is not going to be a subalgebra simply because $C^\infty(C_\phi)$ is not a Poisson algebra, unless $C_\phi$ is a Poisson submanifold, which is a much stronger requirement.} in $C^\infty(C_\phi)$ generated by the restrictions $\iota^*_\phi\psi$.  
{Note that the $\iota^*_\phi\psi_j$ are not required to be invariant (i.e.\ basic) themselves, so it is not going to be true in general that $\pi_\phi^*\underline{\psi}_j=\iota_\phi^*\psi_j$, which corresponds to the special case of ``ideal'' embeddings (Definition \ref{def:nestedcoisoembedding}).} Finally, the $\phi$'s and the $\psi$'s should be functionally independent, meaning that they cannot be written as a combination of the other functions in the set.

At the level of generators of the vanishing ideals, this structure is encoded by\footnote{{We stress here that the standard case of coisotropic reduction by stages is encoded by the condition $g=0$. It is an example of Hamiltonian reduction by stages if, additionally, the structure functions are constant. Our main example will showcase constant structure functions but with $m=0$ and $g\not=0$.}}
\begin{equation}\label{e:Coisostructure}
\{\phi_i,\phi_j\} = f_{ij}^k\phi_k, \qquad \{\psi_i,\psi_j\} = h_{ij}^k\psi_k,\qquad\{\phi_i,\psi_j\} = m_{ij}^k\phi_k + g_{ij}^k\psi_k
\end{equation}
for $f_{ij}^k,h_{ij}^k,g_{ij}^k,m_{ij}^k\in C^{\infty}(F)$. 
Note that in the above formulas we are listing using only one index $i\in\{1,\dots, \ell_1+\ell_2\}$, so that
\[
\phi_i \not= 0 \iff 1\leq i\leq \ell_1, \qquad \psi_i \not= 0 \iff \ell_1+1\leq i\leq \ell_1+k,
\]
and
\begin{align*}
f_{ij}^k\not = 0 & \iff 1\leq i,j,k\leq \ell_1& h_{ij}^k\not = 0 &\iff \ell_1+1\leq i,j,k\leq \ell_1+\ell_2\\
m_{ij}^k\not = 0 & \iff \begin{cases}
    1\leq i,k\leq \ell_1, \\ 
    \ell_1+1\leq k\leq \ell_1+\ell_2
\end{cases}
    & g_{ij}^k\not = 0 &\iff \begin{cases}
        1\leq i\leq \ell_1, \\ 
        \ell_1+1\leq j,k\leq \ell_1+\ell_2
    \end{cases}
\end{align*} 
This clearly also means one can write $\Phi_i=\phi_i$ for $1\leq i\leq \ell_1$ and $\Phi_i=\psi_i$ for $\ell_1+1\leq i\leq \ell_1+\ell_2$, and $C=\mathrm{Zero}\{\Phi_i\}$ is coisotropic since $\{\Phi_i,\Phi_j\} = F_{ij}^k\Phi_k$.

\begin{lemma} \label{l:Relations_structure_constants}
We have the following identities 
    \begin{subequations}\begin{gather}\label{e:pseudo_Jacobi_relation}
    \mathrm{Cyc}_{p,i,j}\left(\{\phi_p,f_{ij}^k\} + f_{ij}^mf_{p m}^k\right) = 0, \\
    \mathrm{Cyc}_{p,i,j}\left(\{\psi_p,h_{ij}^k\} + h_{ij}^mh_{p m}^k\right) = 0, \\
    \mathrm{Cyc}_{j,k} \left(\{\psi_j,m_{ki}^p\} + m_{ki}^l m _{jl}^p  - \frac{1}{2}h_{jk}^l m_{il}^p\right)=0,\\
    \mathrm{Cyc}_{j,k} \left(\{\psi_j,g_{ki}^l\} +m_{ki}^l g _{jl}^p + g_{ki}^l h _{jl}^p - \frac{1}{2}h_{jk}^l g_{il}^p  - \frac{1}{2} \{\psi_i,h_{jk}^p\}\right)=0,\\
    \mathrm{Cyc}_{ij}\left(\{\phi_i,g_{jk}^p\} + g_{jk}^q g_{iq}^p - \frac12 f_{ij}^q g_{kq}^p\right) =0,\label{e:pseudo_Jacobi_relation2}\\
    \mathrm{Cyc}_{p,i}\left( \{\phi_p, m_{ij}^k\} + m_{ij}^m f_{p m}^k + g_{ij}^m m_{p m}^k - \frac{1}{2} f_{ip}^l m_{jl}^k - \frac{1}{2} \{\psi_j, f_{ip}^k \} \right)=0,
    \end{gather}
    \end{subequations}
    
\end{lemma}

\begin{proof}
    
    All these relations follow from the structural equations \eqref{e:Coisostructure} by taking iterated Poisson brackets, computing their cyclic permutations and using Jacobi identity. The identities are then obtained by recalling that the $\phi_i,\psi_j$ are functionally independent.
    
    For example let us prove \eqref{e:pseudo_Jacobi_relation}. Using Jacobi identity we have
    \begin{align*}
        &\mathrm{Cyc}_{p,i,j} \{\phi_p, \{\phi_i, \phi_j\}\}
        =\mathrm{Cyc}_{p,i,j} \{  \phi_p, f_{ij}^k \phi_k \} \\
        &= \mathrm{Cyc}_{p,i,j} f_{ij}^m f_{pm}^k \phi_k + \mathrm{Cyc}_{p,i,j} \{  \phi_p, f_{ij}^k \} \phi_k
    \end{align*}
    from which we deduce the required identity. All the remaining equations are obtained similarly by considering all possible combinations of $\phi_i,\psi_j$.
\end{proof}

\begin{proposition}\label{prop:reductiondecomposition}Let $\mathcal{I}_C$, $\mathcal{I}_\phi$, $\mathcal{I}_{\underline{\psi}}$ denote the vanishing ideals associated to the sets $\{\phi_i,\psi_j\}$, $\{\phi_i\}$ and $\{\underline{\psi_j}\}$ respectively. We have
    \[
    C^\infty(\underline{C})\equiv \left(C^\infty(F)/\mathcal{I}_C\right)^{\mathcal{I}_C} \simeq \left(C^\infty(\underline{C_\phi})/I_{\underline{\psi}}\right)^{I_{\underline{\psi}}} \equiv \left(\left(C^\infty(F)/\mathcal{I}_\phi\right)^{\mathcal{I}_\phi}/I_{\underline{\psi}}\right)^{I_{\underline{\psi}}}
    \]
\end{proposition}

\begin{proof}
    Denote by $[f]_\phi\in C^\infty(\underline{C_\phi})$ the class of an $\mathcal{I}_\phi$ invariant function in $C^\infty(C_\phi)$. 
    A class of functions in $C^\infty(\underline{C_\phi})/I_{\underline{\psi}}$ is $I_{\underline{\psi}}$ invariant iff for every representative $[f]_\phi + G(\underline{\psi})$ we have $X_{\underline{\psi}}([f]_\phi + G(\underline{\psi})) = 0$. However, we can test this condition in {$F$} by choosing {an invariant representative $f + F(\phi)\in C^\infty(F)$ in the class $[f]_\phi$} and {a basic function representing $\underline{\psi}$, i.e.\ of the form $\iota^*_\phi\psi = \pi^*\underline{\psi}$}, namely
    \begin{align*}
        0 = X_{\underline{\psi}}([f]_\phi + G(\underline{\psi})) &\iff X_{\iota^*_\phi\psi}(f + F(\phi) + G(\iota^*_\phi\psi)) = 0 \\
        &\iff X_{\iota^*_\phi\psi}(f) = X_{\iota^*_\phi\psi}(F(\phi) + G(\iota^*_\phi\psi)) = F'(\phi) + G'(\iota^*_\phi\psi).
    \end{align*}
    which means that $X_{\iota^*_\phi\psi}(f) = 0 \mod \phi \mod \iota^*_\phi\psi$, that is to say $f$ defines an element in $\left(C^\infty(M)/\mathcal{I}_C\right)^{\mathcal{I}_C}\simeq C^\infty(\underline{C})$, since $f \mod \phi$ is also $X_\phi$ invariant by assumption.

    On the other hand such a function $[f]\in C^\infty(\underline{C})$ is represented by an $X_\phi$- and $X_\psi$-invariant function $f + F(\phi) + G(\psi) = f + F(\phi) + G(\iota^*_\phi\psi) + H(\phi)$, which defines a function $[f]_\phi + G(\underline{\psi}) \in C^\infty(\underline{C_\phi})$ which represents the class $[[f]_\phi]_{\underline{\psi}}\in C^\infty(\underline{C_\phi})/\mathcal{I}_{\underline{\psi}}$.
\end{proof}

{
\begin{definition}[Nested Coisotropic Embedding]\label{def:nestedcoisoembedding}
    A nested coisotropic embedding is a sequence of presymplectic manifolds
    \[
    C\hookrightarrow \Co\hookrightarrow F,
    \]
    with $\Co \equiv C_\phi$ and $C\equiv C_\phi\cap C_\psi$, for some set of generators $\{\phi_i,\psi_i\}$ that satisfy the structure equations \eqref{e:Coisostructure}. We also denote $\Cres \equiv C_{\underline{\psi}}$, so that $\underline{\Cres}\simeq\underline{C}$ (Proposition \ref{prop:reductiondecomposition}).

    We will say that a nested coisotropic embedding is 
    \begin{enumerate}
        \item \emph{nice} iff the structure functions $g_{ij}^k$ are constant;
        \item \emph{ideal} iff the structure functions $g_{ij}^k=0$, i.e.\ if $C$ is generated by functions on $\Co$ that are basic w.r.t.\ the coisotropic reduction $\pi_\circ\colon \Co \to \underline{\Co}$.
        \item \emph{of Lie type} iff all structure functions are constant;
    \end{enumerate}
\end{definition}

\begin{remark}
Note that we have two degenerate scenarios, when $C\equiv \Co\hookrightarrow F$ and when $C\hookrightarrow \Co\equiv F$. In the first case the set of (residual) constraints $\{\psi_j\}=\emptyset = \{\underline{\psi}\}_j$ which implies that $C_{\mathrm{res}}\equiv \underline{\Co}$, which is symplectic and thus its reduction is trivially the space itself $\underline{C_{\mathrm{res}}}\equiv \underline{\Co}\simeq \underline{C}$. In the second case, instead, the set of constraints $\{\phi_i\}=\emptyset$ so that $\psi_j \equiv \underline{\psi}_j$ and $C\equiv C_{\mathrm{res}}$, whence clearly we have again $\underline{C}=\underline{C_{\mathrm{res}}}$.
\end{remark}

We have the following obvious relation between these notions.
\begin{proposition}
    Coisotropic embeddings that are ideal, or of Lie type, are nice.
\end{proposition}

When a nested coisotropic embedding is ideal, the vanishing ideal generated by the $\phi$ functions is a Poisson ideal of the vanishing ideal generated by both sets of functions $\phi$ and $\psi$, which in turn is a Poisson subalgebra of all functions. In other words, the $\psi$ functions restricted to $C_\phi$ are invariant under the action of any $\phi$ (not only the ideal is invariant but all generators are).

Our main example of General Relativity---the application that motivated this construction---features an ideal nested coisotropic embedding, as it satisfies the additional requirement $g_{ij}^k=0$.  However, it is worthwhile to explore the general case.

}

\begin{remark}[Reduction by stages]
Coisotropic embeddings of Lie type emerge, for example, when $C$ is the zero-level set of the momentum map for a Hamiltonian action $G\circlearrowright F$, while $\Co$ is the zero-level set of the momentum map associated to a (normal) subgroup. In other words, an example is given when performing symplectic reduction ``by stages'' \cite{marsden2007stages}. While in this section we will \emph{not} assume that the functions $f_{ij}^k$ and $h_{ij}^k$ be structure constants of a Lie algebra, in our main example this will be the case. However, it is worthwhile to allow for a more general scenario.
\end{remark}

We summarise this construction with the following diagram
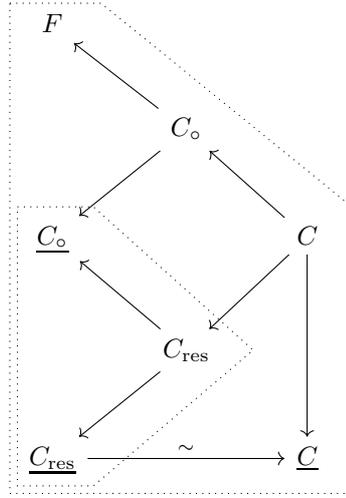
\[
\tikz[
overlay]{
    \draw[dotted] (0, 0.5) -- (1,0.5) -- (3.1, -1.4)-- (1,-3.2)--(0,-3.2)--(0,0.5);
    \draw[dotted] (-0.1, 3.2) -- (1.1,3.2) -- (4.4, 0.5)-- (4.4,-3.3)--(-0.1,-3.3)--(-0.1, 3.2);
}
\begin{tikzcd}[row sep=2.5em]
    F & &\\
 &\Co \arrow[ul] \arrow[ld] &\\
\underline{\Co}& & \arrow[ul]\arrow[dl]C \arrow[dd]\\
& \arrow[ul]\arrow[dl]\Cres &\\
 \underline{\Cres} \arrow[rr, "\sim"] & & \underline{C}
\end{tikzcd}
\]

We can then construct two BFV spaces, one associated with the reduction of $C\hookrightarrow F$ and one associated with $\Cres\hookrightarrow \underline{\Co}$. They are given, respectively, by
\begin{subequations}\begin{align}
\BFV(F,C)&\doteq \left(\F(F,C),\ \omega_{C},\ Q_C, S_C\right), & \F(F,C) &\doteq F \times T^*W[1], \label{e:BFVBigRes}\\ 
\BFV(\underline{\Co},\Cres)& \doteq\left(\F(\underline{\Co},\Cres),\ \omega_{\mathrm{res}},\  Q_{\mathrm{res}}, S_{\mathrm{res}}\right), & \F(\underline{\Co},\Cres) &\doteq \underline{\Co} \times T^*W_{\mathrm{res}}[-1]. \label{e:BFVSmallRes}
\end{align}\end{subequations}

Note that we denote by $T^*W[1]$, $T^*W_\circ[1]$ and $T^*W_{\mathrm{res}}[1]$ the BFV components respectively associated to the resolution of the quotients of zero loci $C$, $\Co$ and $\Cres$ w.r.t.\ their characteristic foliations, so that
\[
H^0(\mathfrak{BFV}^\bullet(F,C))\simeq H^0(\mathfrak{BFV}^\bullet(\underline{\Co},\Cres))  \simeq C^\infty(\underline{C}).
\]

Resolutions are diagrammatically represented by a wavy arrow, relating to the reduction diagram to their left:
\begin{equation}\label{e:BFVdiagrams}
\xymatrix{
F & \ar@{_{(}->}[l]C \ar[d]\ar@{~>}[r]^-{BFV}& \F(F,C) & \underline{\Co} & \ar@{_{(}->}[l]\Cres \ar[d]\ar@{~>}[r]^-{BFV} & \F(\underline{\Co},\Cres)\\
& \underline{C} &  & & \underline{\Cres} &
}
\end{equation}

Using generators $\{\phi,\psi\}$ we have\footnote{For simplicity, here we are assuming that there are no reducibilities/isotropies in the associated algebroids, and that the functions $\phi_i,\psi_j$ are functionally independent. This allows us to pick $W_\phi\simeq \mathbb{R}^{\ell_1}$ and similarly for $\psi$. More generally, the BFV structure might include degree $\pm2$ variables (and higher) to account for the nonvanishing of the higher cohomology groups of the Koszul complex. See Remark \ref{rmk:generalBFV}.} $T^*W[1] = T^*W_\phi[1]\times T^*W_\psi[1] = T^*\mathbb{R}^{\ell_1}[1] \times T^*\mathbb{R}^{\ell_2}[1]$. We denote coordinates on $T^*W_\phi[1]$ and $T^*W_\psi[1]$ by $(\chi^\dag,\chi)$ and $(\lambda^\dag,\lambda)$ respectively, while we use a Darboux chart $(q,p)$ for $F$.

\begin{assumption}\label{ass:BFV_action}
     With the notation above, the BFV data associated to the coisotropic submanifold $C\subset F$ has the following explicit form:
    \begin{gather*}
        \BFV(F,C)\doteq\left(\F(F,C)\doteq F \times T^*\mathbb{R}^{\ell_1}[1] \times T^*\mathbb{R}^{\ell_2}[1], \omega_C, Q_C, S_C\right)\\
        S_C = \int_\Sigma \chi^i {\phi_i} + \lambda^i \psi_i + \frac{1}{2} f_{ij}^k    \chi^i\chi^j  \chi^\dag_k + \frac{1}{2}h_{ij}^k \lambda^i\lambda^j\lambda^\dag_k + g_{ij}^k \chi^i\lambda^j\lambda^\dag_k +m_{ij}^k\chi^i\lambda^j \chi^\dag_k . 
    \end{gather*}
\end{assumption}

\begin{remark}
    Sufficient conditions for the BFV data to have the form described in Assumption \ref{ass:BFV_action} are the following:
    \begin{enumerate}
        \item The structure functions $f,h,g,m$ are constant on $F$.
        \item The structure functions satisfy for $\theta_1=f$, $\theta_2= h$, $\theta_3= g$ and $\theta_4=m$
            \begin{align} \label{e:parentheses_between_structure_functions}
                \{\theta_i, \theta_j\}_F=0.
            \end{align}
    \end{enumerate}
    We leave the proof of this statement to the reader.
\end{remark}

\begin{theorem}\label{thm:BFVcoisotropic}
    Let $C\hookrightarrow\Co \hookrightarrow F$ be a nested coisotropic embedding and denote the associated BFV (differential) graged manifold by 
    \[
    \F(F,C)\doteq F \times T^*W_\phi[1]\times T^*W_\psi[1]\simeq F \times T^*\mathbb{R}^{\ell_1}[1] \times T^*\mathbb{R}^{\ell_2}[1]
    \]
    Consider the following functions on $\F(F,C)$ 
    \[
        \mathsf{M}^\mu\doteq\langle \mu,\boldsymbol{\Phi}^\dag\rangle, \quad \mathsf{L}^\rho\doteq Q_C\mathsf{M}^\rho = \langle \rho, Q_C\boldsymbol{\Phi}^\dag\rangle, \quad \mathsf{J}^\rho \doteq \langle \rho,\boldsymbol{\Phi}\rangle ,
    \]
    where,\footnote{It is useful to think of $\boldsymbol{\Phi}$  and $\boldsymbol{\Phi}^\dag$ as taking values in the duals $(\mathbb{R}^{\ell_1})^*$ and $(\mathbb{R}^{\ell_1}[1])^*$ respectively.} denoting variables in $T^*W_\phi[1]\simeq T^*\mathbb{R}^{\ell_1}[1]$ by $(\chi,\chi^\dag)$ we write
    \[
    \begin{cases}
        \boldsymbol{\Phi}^\dag\colon \F(F,C) \to \mathbb{R}^{\ell_1}[-1], & \boldsymbol{\Phi}^\dag\doteq \chi^\dag,\\
        \boldsymbol{\Phi}\colon \F(F,C) \to \mathbb{R}^{\ell_1}, & \boldsymbol{\Phi}\doteq Q_C\chi^\dag\vert_{\chi^\dag=0} = Q_C\chi^\dag \mod \boldsymbol{\Phi}^\dag
    \end{cases}
    \]
    and setting $\rho\in \mathbb{R}^{\ell_1}[1]$ and $\mu\in \mathbb{R}^{\ell_1}[2]$, both $\mathsf{M}^\mu$ and $\mathsf{L}^\rho$ are degree 1. Denote by $\mathbb{I}_\circ$ the vanishing ideal generated by $\mathsf{J}^\rho,\mathsf{M}^\mu$. Then, $\mathbb{I}_\circ$ is a Poisson subalgebra of $C^\infty(\F(F,C))$, with
    \begin{gather*}
    \mathsf{J}^\rho=\mathsf{L}^\rho - \mathsf{M}^{\llbracket\rho,\chi\rrbracket + m(\rho,\lambda)}, \\
    \{\mathsf{M}^\mu,\mathsf{M}^\mu\} =0, \quad \{\mathsf{L}^\rho,\mathsf{L}^\rho\} = \mathsf{L}^{\llbracket\rho,\rho\rrbracket},\quad \{\mathsf{L}^\rho,\mathsf{M}^\mu\} = \mathsf{M}^{\llbracket\rho,\mu\rrbracket}, \\
    \{\mathsf{J}^\rho,\mathsf{J}^\rho\} = \mathsf{J}^{\llbracket{\rho,\rho\rrbracket}}, \quad \{\mathsf{J}^\rho,\mathsf{M}^\mu\}=0
    \end{gather*}

    \begin{align*}
        \llbracket\rho,\chi\rrbracket \doteq \rho^i\chi^jf_{ij}^k e_k & & m(\rho,\lambda)\doteq \rho^i\lambda^jm_{ij}^k e_k
    \end{align*}
    given a basis $\{e_k\}$ of $\mathbb{R}^{\ell_1}$, and ${X}_{\mathsf{J}^\rho}$ the Hamiltonian vector field of $\mathsf{J}^\rho$. Hence, denoting by $\CCo$ the subset defined by the vanishing ideal $\mathbb{I}_\circ$. Then $\CCo\hookrightarrow(\F(F,C),\omega_C)$ is a coisotropic submanifold, if smooth, and moreover $\mathbb{I}_\circ$ is a differential ideal\footnote{This shows in particular that $\CCo$ is coisotropic in a dg manifold, according to the stronger notion discussed in Remark \ref{rmk:dgcoiso}.} w.r.t. $Q_C$.
\end{theorem}

\begin{proof}
    We begin by observing that, using the notation introduced above
    \[
    Q_C\chi^\dag_i = \phi_i + f_{ij}^k  \chi^j  \chi^\dag_k + g_{ij}^k \lambda^j\lambda^\dag_k +m_{ij}^k\lambda^j \chi^\dag_k = (\underbrace{\phi_i  + g_{ij}^k \lambda^j\lambda^\dag_k}_{Q_C\chi^\dag_i\vert_{\chi^\dag=0}}) + (f_{ij}^k  \chi^j  \chi^\dag_k +m_{ij}^k\lambda^j \chi^\dag_k)
    \]
    hence for any $\rho\in \mathbb{R}^{\ell_1}[1]$, we have
    \[
    \mathsf{L}^\rho = \rho^i(Q_C\chi^\dag_i\vert_{\chi^\dag=0}) + \rho^i(f_{ij}^k  \chi^j  \chi^\dag_k +m_{ij}^k\lambda^j \chi^\dag_k)  = \mathsf{J}^\rho + \mathsf{M}^{\llbracket\rho,\chi\rrbracket + m(\rho,\lambda)}
    \]
    where $\llbracket\rho,\chi\rrbracket = f_{ij}^k \rho^i \chi^j e_k$, $m(\rho,\lambda)=m_{ij}^k\rho^i\lambda_{\psi^j} e_k$ and, denoting by $g(\rho,\lambda)\equiv g_{ij}^k\rho^i\lambda^j e_k$, we write
    \begin{align}\label{e:Expression_J_general}
        \mathsf{J}^\rho = \rho^i\phi_i + \rho^ig_{ij}^k\lambda^j\lambda^\dag_k = \langle\rho,\phi\rangle + \langle g(\rho,\lambda),\lambda^\dag\rangle.
    \end{align}

    Moreover, we observe that the Hamiltonian vector field ${X}_{\mathsf{M}^\mu}$ of $\mathsf{M}^\mu$ is the odd vector field
    \[
    \iota_{{X}_{\mathsf{M}^\mu}}\omega_C = d\mathsf{M}^\mu, \qquad {X}_{\mathsf{M}^\mu} = \mu\frac{\partial}{\partial \chi}, \qquad |X_{\mathsf{M}^\mu}| = 1,
    \]
    since the symplectic form on $\F(F,C)$ is 
    \[
    \omega_{C} = \omega + \sum_i d\chi^\dag_i \wedge d\chi^i + \sum_j d\lambda^\dag_j \wedge d\lambda^j, 
    \]
    with $\omega$ denoting the symplectic form on $F$. It is immediate to conclude that 
    \[
    \{\mathsf{M}^\mu,\mathsf{M}^\nu\}=0
    \]
    for any $\mu,\nu\in\mathbb{R}^{\ell_1}[k]$, for arbitrary $k$. (Note that this implies also that $\{\mathsf{M}^\mu,\mathsf{M}^\mu\}=0$ which is nontrivial for $\mu$ even, since in that case $\mathsf{M}^\mu$ is an odd function.) Similarly, we also have that for any function $\mathbf{G}\colon \F(F,C) \to (\mathbb{R}^{\ell_1}[k])^*\otimes \mathbb{R}^{\ell_1}[k]$ and for any $\nu\in \mathbb{R}^{\ell_1}[k]$ we have
    \[
    \{\mathsf{M}^\mu, \mathsf{M}^{\langle \mathbf{G},\nu\rangle}\} = \mathsf{M}^{\langle{X}_{\mathsf{M}^\mu} \mathbf{G},\nu\rangle}
    \]
    due to linearity. For example, for $\langle \mathbf{G},\rho\rangle=\llbracket\rho,\chi\rrbracket + m(\rho,\lambda)$ we easily compute 
    \[
    \{\mathsf{M}^\mu,\mathsf{M}^{\llbracket\rho,\chi\rrbracket + m(\rho,\lambda)}\} = \mathsf{M}^{{X}_{\mathsf{M}^\mu} (\llbracket\rho,\chi\rrbracket + m(\rho,\lambda))} = \mathsf{M}^{\llbracket\rho,\mu\rrbracket}
    \]
    
    By direct inspection we can also check that, choosing $\rho\in\mathbb{R}^{\ell_1}[1]$,
    \[
    \{\mathsf{M}^\mu,\mathsf{J}^\rho\} = {X}_{\mathsf{M}^\mu}(\mathsf{J}^\rho) = 0
    \]
    since $\mathsf{J}^\rho$ does not depend on $\chi$, from which we also conclude that
    \[
    \{\mathsf{M}^\mu,\mathsf{L}^\rho\} = \{\mathsf{M}^\mu,\mathsf{J}^\rho\} + \{\mathsf{M}^\mu,\mathsf{M}^{\llbracket\rho,\chi\rrbracket + m(\rho,\lambda)}\} = \{\mathsf{M}^\mu,\mathsf{M}^{\llbracket\rho,\chi\rrbracket + m(\rho,\lambda)}\} = \mathsf{M}^{\llbracket\rho,\mu\rrbracket},
    \]
    and using $Q^2=0$ we compute:
    \[
    \{\mathsf{L}^\rho,\mathsf{L}^\rho\} = Q\{\mathsf{M}^\rho,\mathsf{L}^\rho\} = Q\mathsf{M}^{\llbracket\rho,\rho\rrbracket} = \mathsf{L}^{\llbracket\rho,\rho\rrbracket}
    \]
    On the other hand, using \eqref{e:Expression_J_general} we have
    \begin{align*}
        \{\mathsf{J}^\rho,\mathsf{J}^\rho\}&= \{\langle\rho,\phi\rangle,\langle\rho,\phi\rangle\}_F + 2 \{\langle\rho,\phi\rangle, \langle g(\rho,\lambda),\lambda^\dag\rangle\}_F \\
        &\phantom{=}+ \{ \langle g(\rho,\lambda),\lambda^\dag\rangle, \langle g(\rho,\lambda),\lambda^\dag\rangle\}_F  + 2\langle g(\rho,g(\rho,\lambda)),\lambda^\dag\rangle\\
        &= \{\phi_i, \phi_j\}_F \rho^i \rho^j + 2 \{\phi_i,  g_{jm}^k\}\rho^i \rho^j \lambda^m \lambda^\dag_k  \\
        &\phantom{=}+ \{g_{in}^l, g_{jm}^k\}\rho^i \lambda^n \lambda^\dag_l \rho^j \lambda^m \lambda^\dag_k + 2g_{ij}^k \rho^i g_{lm}^j\rho^m \lambda^l \lambda^\dag_k.
    \end{align*}
    Now, using \eqref{e:pseudo_Jacobi_relation2}, \eqref{e:parentheses_between_structure_functions} and the definition of $f$, we get
    \begin{align*}
        \{\mathsf{J}^\rho,\mathsf{J}^\rho\}&= \rho^i \rho^j f_{ij}^k \phi_k + g^k_{lj} f^{j}_{im}\rho^i\rho^m \lambda^l \lambda^\dag_k\\
        &= \mathsf{J}^{\llbracket \rho, \rho \rrbracket}.
    \end{align*}
    We are left to show that $\mathbb{I}_\circ$ is stable under $Q_C$. We use the result above, $\mathsf{J}^\rho = \mathsf{L}^\rho - \mathsf{M}^{\llbracket\rho,\chi\rrbracket + m(\rho,\lambda)}$, twice, to show that $Q_C\mathbb{I}_\circ \subset \mathbb{I}_\circ$. Indeed, since $Q_C\mathsf{L}^\rho =  Q_C (Q_C\mathsf{M}^\rho)=0$
    \begin{multline}
    Q_C\mathsf{J}^\rho = Q_C \mathsf{L}^\rho - Q_C\left(\mathsf{M}^{\llbracket\rho,\chi\rrbracket + m(\rho,\lambda)}\right) =- \mathsf{M}^{Q_C\left(\llbracket\rho,\chi\rrbracket + m(\rho,\lambda)\right)} \pm \mathsf{L}^{\llbracket\rho,\chi\rrbracket + m(\rho,\lambda)}\\
    =- \mathsf{M}^{Q_C\left(\llbracket\rho,\chi\rrbracket + m(\rho,\lambda)\right)}\pm \mathsf{J}^{\llbracket\rho,\chi\rrbracket + m(\rho,\lambda)} \mp \mathsf{M}^{\llbracket\llbracket\rho,\chi\rrbracket + m(\rho,\lambda),\chi\rrbracket + m(\llbracket\rho,\chi\rrbracket + m(\rho,\lambda),\lambda)}\in \mathbb{I}_\circ.
    \end{multline}
    Similarly we compute
    \[
    Q_C\mathsf{M}^\mu = \mathsf{L}^\mu = \mathsf{J}^\mu - \mathsf{M}^{\llbracket\llbracket\mu,\chi\rrbracket + m(\mu,\lambda)}\in \mathbb{I}_\circ.
    \]
    
 \end{proof}

\begin{theorem}\label{thm:doubleBFV}
    Let $C\hookrightarrow\Co \hookrightarrow F$ be a nested coisotropic embedding, and assume that the coisotropic locus $\CCo$ is smooth, as defined in Theorem \ref{thm:BFVcoisotropic}. Assuming furthermore that its reduction by $\underline{\CCo}\doteq \CCo/\CCo^{\omega_C}$ is also smooth, we have that
    \[
    \mathsf{Body}(\underline{\CCo})\simeq \underline{\Co}. 
    \]
    If additionally the nested coisotropic embedding is nice, then
    \[\underline{\CCo} \simeq \F(\underline{\Co}, \Cres) = \underline{\Co} \times T^*W_{\mathrm{res}}[-1].
    \]
\end{theorem}
\begin{proof}
    It is immediate to note that the characteristic foliation of the constraint $\mathbf{\Phi}^\dag\equiv \chi^\dag=0$, generated by the vector field ${X}_{\mathsf{M}^\mu} = \langle\mu,\frac{\partial}{\partial \chi}\rangle$, reduces to a point (cross whatever is left from the reduction of the second constraint), since $\chi^\dag=0$ selects the zero section in the cotangent bundle of $W=W_\phi$, which is Lagrangian in it.

    The Hamiltonian vector field of $\mathsf{J}^\rho$ instead reads
    \[
    {X}_{\mathsf{J}^\rho} = \rho_i \{\phi_i,\cdot \}_F + \rho^i\lambda^j\lambda^\dag_k\{g_{ij}^k, \cdot\}_F + \rho^ig_{ij}^k\{\lambda^j\lambda^\dag_k, \cdot\}_{T^*W_{\psi}[-1]}
    \]
    This induces the system of ODE's:
    \begin{equation}\label{e:gradedODE}
    \begin{cases}
        \dot{\lambda}_{k} = \rho^ig_{ij}^k\lambda^j \\
        \dot{\lambda}^\dag_{j} = \rho^ig_{ij}^k\lambda^\dag_k\\
        \dot{p}_i = \rho^m\left(\frac{\partial \phi_m}{\partial{q}_i} + \lambda^j\lambda^\dag_k\frac{\partial g_{mj}^k}{\partial q_i}\right)\\
        \dot{q}^i = -\rho^m\left(\frac{\partial \phi_m}{\partial{p}_i} + \lambda^j\lambda^\dag_k\frac{\partial g_{mj}^k}{\partial p_i}\right)
    \end{cases}
    \end{equation}
    Suppose that a solution exists of the form $(p(t),q(t), \lambda(t), \lambda^\dag(t))$, where $(p(t),q(t))$ are general functions of all variables, that we can expand in power series of $\lambda^\dag$ (observe that this power series has a finite number of terms for degree reasons):
    \begin{equation}\label{e:gradedflowODE}
    \begin{cases}
        p(t) = p_{(0)} + \lambda_k^\dag p^k_{(1)} + \lambda_k^\dag\lambda_m^\dag p^{km}_{(2)} + \dots\\
        q(t) = q_{(0)} + \lambda_k^\dag q^k_{(1)} + \lambda_k^\dag\lambda_m^\dag q^{km}_{(2)} + \dots
    \end{cases}
    \end{equation}
    where $(p_{(i)},q_{(i)})$ denote degree $i$ functions of all variables except $\lambda^\dag$.
    Then, this is a solution of the ODE's \eqref{e:gradedODE} only if $(\dot{p},\dot{q})$ are of the form above. Observe that the only terms homogeneous of degree $0$ in $\lambda^\dag$ emerge from $(\dot{p}_{(0)},\dot{q}_{(0)})$, owing to the fact that $\dot{\lambda}^\dag\propto \lambda^\dag$, meaning that, expanding the equation in powers of $\lambda^\dag$, at lowest order one gets 
    \[
    \begin{cases}
        \dot{p}_{(0),i} = \rho^m\frac{\partial \phi_m}{\partial{q}_i}\\
        \dot{q}^i_{(0)} = -\rho^m\frac{\partial \phi_m}{\partial{p}_i}
    \end{cases}
    \]
    The ODE's \eqref{e:gradedflowODE} are then solved by
    \[
    (p(t),q(t)) = \left(p_{(0)}(t) + P, q_{(0)}(t) + Q\right)
    \]
    where $(P,Q)$ are functions of all the variables, at least linear in both $(\lambda,\lambda^\dag)$, and homogeneous of (total) degree zero. In other words, in degree zero there is a flow $(p_{(0)}(t), q_{(0)}(t))$, which is then "lifted" to $\F(F,C)$. This means that the above equations \eqref{e:gradedflowODE} define a morphism of the sheaf defining the graded symplectic manifold $\F(F,C)$, which covers a morphism of the sheaf of functions on the base $F$. Hence, we conclude that the body of $\underline{\CCo}$ coincides with $\underline{\Co}$. (See \cite{Vysoky} for a recent exposition of the theory of graded manifolds.)
    
    The nested coisotropic embedding is nice, by definition, when the $g_{ij}^k$ are constant. In this case, the system of ODE's simplifies to:

    \begin{equation}\label{e:simplifiedgradedODE}
    \begin{cases}
        \dot{\lambda}_{k} = \rho^ig_{ij}^k\lambda^j \\
        \dot{\lambda}^\dag_{j} = \rho^ig_{ij}^k\lambda^\dag_k\\
        \dot{p}_i = \rho^m \frac{\partial \phi_m}{\partial{q}_i} \\
        \dot{q}^i = -\rho^m\frac{\partial \phi_m}{\partial{p}_i} 
    \end{cases}
    \end{equation}

    The  equations for the variables in degree zero, which define the foliation in $F$, now only involve degree $0$ variables, while the ODE's for the $(\lambda,\lambda^\dag)$ variables represent the adjoint and coadjoint action of a lie algebra with structure constants $g_{ij}^k$.
    Explicitly one can parametrize the reduced space by flowing along $(p(t),q(t))$ and fix a representative in terms of a slice. This uniquely fixes $\rho$ for a given choice of representative of $\underline{\Co}$ inside $F$. Then we let the variables $(\lambda,\lambda^\dag)$ flow using the chosen $\rho$ to another element in $T^*W_{\mathrm{res}}$.
\end{proof}

\begin{remark}
    An instance of this reduction in the case of constant coefficients can be found in Appendix \ref{a:coisotropic_BFtoEH} for $BF$ theory.
\end{remark}

Given a coisotropic submanifold $\C$ of a graded symplectic manifold $\F'$ we can apply the BFV machinery, only this time in the graded setting, and build
\[
\BFV(\F',\C)=(\F(\F',\C), \omega_\C,Q_\C, S_\C).
\]

In the special case where $\F'$ is obtained from the BFV resolution of a coisotropic embedding $C\hookrightarrow F$ that factors through $C\hookrightarrow \Co\hookrightarrow F$, Theorem \ref{thm:doubleBFV} ensures that there exists a graded coisotropic submanifold $\CCo\subset \F(F,C)$, whose resolution can be resolved via BFV. Hence, the general procedure instructs us to build:
\begin{definition}[Double BFV data] \label{d:doubleBFVdata}
The \emph{double BFV data} associated to the nested coisotropic embedding $C\hookrightarrow \Co \hookrightarrow F$ is:
    \[
    \mathbb{BFV}(F,\Co,C)\doteq \BFV(\BFV(F,C),\CCo) =\left(\F(\F(F,C),\CCo),\omega_{\CCo},Q_{\CCo},S_{\CCo}\right) \doteq \left(\FF,\bbomega,\mathbb{Q},\mathbb{S}\right),
    \]
where we denoted $\FF$ as a shorthand for $\FF(F,\Co,C)\doteq\F(\F(F,C),\CCo)$ and $(\bbomega,\mathbb{Q},\mathbb{S})$ as shorthand for $(\omega_{\CCo},Q_{\CCo},S_{\CCo})$, while $\CCo$ is as in Theorem \ref{thm:doubleBFV}.
\end{definition}

Introducing coordinates $(\mu^\dag,\mu, \rho^\dag,\rho)$ in $\FF=\F(F,C)\times T^*\mathbb{R}^{\ell_1+\ell_2}$ so that (cf. Theorem \ref{thm:BFVcoisotropic}) 
\[
\mathbb{Q}\mu^\dag \doteq Q_C\chi^\dag \vert_{\chi^\dag=0} = \boldsymbol{\Phi}, \quad \mathbb{Q}\rho^\dag \doteq \chi^\dag = \boldsymbol{\Phi}^\dag,
\]
together with $\mathbb{S}\in C^\infty(\FF)$, defined as in equation \eqref{e:Rfunction} such that $\{\mathbb{S},\mathbb{S}\}=0$, providing a resolution of $\underline{\CCo}$. Specifically we have:
\begin{align*}
\mathbb{S} & \doteq  \mathsf{J}^\rho + \mathsf{M}^\mu + \frac12 \langle \llbracket \rho,\rho\rrbracket, \rho^\dag \rangle \\
&= \langle\rho,\boldsymbol{\Phi}\rangle + \langle \mu, \boldsymbol{\Phi}^\dag\rangle + \frac12 \langle \llbracket \rho,\rho\rrbracket, \rho^\dag \rangle \\
&= \langle \rho, Q_C\chi^\dag \vert_{\chi^\dag=0}\rangle + \langle\mu,\chi^\dag\rangle + \frac12 \langle \llbracket \rho,\rho\rrbracket, \rho^\dag \rangle.
\end{align*}

Recall that the original graded symplectic manifold is endowed with a function $S_C \in C^{\infty}(\F(F,C))$ satisfying the Classical Master Equation $\{S_C,S_C\}=0$ (with respect to $\omega_C$), this is the Hamiltonian function of the cohomological vector field $Q_C$, and it corresponds to the master functional for the BFV resolution of $\underline{C}$. We then look for an extension $\check{S}_C$ of $S_C$, i.e.\ a function in  $C^\infty(\FF(F,\Co,C))$ whose restriction to the zero section 
\[
\mathbb{0}\colon \F(F,C)\hookrightarrow \FF(F,\Co,C)=\F(F,C)\times T^*\mathbb{R}^{\ell_1+\ell_2}
\]
coincides with $S_C$, and such that 
\begin{equation}\label{e:projectableextension}
\{\mathbb{S},\check{S}_C\} = \mathbb{Q}(\check{S}_C) = 0.
\end{equation}

Note that, with respect to the internal grading of the double BFV space $\FF$, generated by $\mu,\rho$ and their dual variables, the function $\check{S}_C$ has degree zero. Then, we have that 
\[
\{\check{S}_C+ \mathbb{S},\check{S}_C+\mathbb{S}\} = \{\check{S}_C,\check{S}_C\},
\]
meaning that the sum $\check{S}_C+ \mathbb{S}$ satisfies the classical master equation if and only if the extension $\check{S}_C$ does. Suppose that the extension $\check{S}_C$ satisfies the classical master equation only up to $\mathbb{Q}$-exact terms, namely
\begin{equation}\label{e:projectableextensionCME}
\frac12\{\check{S}_C,\check{S}_C\} = \mathbb{Q}(f),
\end{equation}
for some function $f$, then, Equation \eqref{e:projectableextension} allows us to conclude that 
\[
\underline{S_{\CCo}}\equiv [\check{S}_C]_{\mathbb{Q}}\in H^0(C^\infty(\FF),\mathbb{Q}) \simeq C^\infty(\underline{\CCo})
\] 
and thus it is a function on $(\underline{\CCo},\underline{\omega_{\CCo}})$ that---in virtue of \eqref{e:projectableextensionCME}---satisfies the classical master equation: $\{\underline{S_{\CCo}},\underline{S_{\CCo}}\}=0$, where the Poisson bracket on cohomology is defined as
\[
\{\underline{S_{\CCo}},\underline{S_{\CCo}}\} \doteq [\{\check{S}_C + \mathbb{Q}(g), \check{S}_C + \mathbb{Q}(h)\}]_{\mathbb{Q}} = [\{\check{S}_C,\check{S}_C\} + \mathbb{Q}(\{\check{S}_C,g\} + \{\check{S}_C,h\})]_{\mathbb{Q}} = 0,
\]
meaning that $\underline{Q_{\CCo}}\doteq \{\underline{S_{\CCo}},\cdot\}$ is cohomological. 

As we see now, in some particular cases, one has a stronger result, namely that $\check{S}_C$ strictly satisfies the CME.

\begin{theorem}[Resolution commutes with reduction] \label{thm:doubleBFVreduced}
    Let $\mathbb{BFV}(F,\Co,C)\equiv \left(\FF,\bbomega,\mathbb{Q},\mathbb{S}\right)$ be the double BFV data associated to the nested coisotropic embedding $C\hookrightarrow \Co \hookrightarrow F$, and let $\BFV(F,C)\equiv\left(\F(F,C),\omega_C,Q_C,S_C\right)$ be the BFV data associated to the coisotropic embedding $C\hookrightarrow F$. If the nested coisotropic embedding is of Lie type, then there exists an extension $\check{S}_C\in C^\infty(\FF(F,\Co,C))$ s.t. $\check{S}_C\vert_\mathbb{0}=S_C$, given by
    \begin{align}
        \label{e:generalSC}
    \check{S}_C = S_C &+  \langle \mu,  \rho^\dag \rangle+\langle \llbracket \chi, \mu \rrbracket, \mu^{\dag} \rangle +\langle \llbracket \chi, \rho \rrbracket, \rho^{\dag} \rangle + \langle m(\lambda, \rho), \rho^\dag \rangle \\
    &+ \langle m(\lambda, \mu), \mu^\dag \rangle +\langle m(\chi, g(\rho,\lambda)), \mu^\dag \rangle ,\nonumber
    \end{align}
    which is such that
    \begin{enumerate}
        \item it is a $\mathbb{Q}$-cocycle: $\mathbb{Q}(\check{S}_C) = 0$,
        \item it satisfies the classical master equation $\{\check{S}_C,\check{S}_C\}= 0$.
    \end{enumerate}
    Furthermore, there is an isomorphism of Hamiltonian dg-manifolds:\footnote{We have shortened $\F_{\mathrm{res}}\equiv \F(\underline{\Co},\Cres)$.}
    \[
    \underline{\mathbb{BFV}(F,\Co,C)}\doteq\left(\underline{\CCo},\underline{\omega_{\CCo}}, \underline{Q_{\CCo}}, \underline{S_{\CCo}}\right)
    \simeq \left(\F_{\mathrm{res}},\ \omega_{\mathrm{res}},\  Q_{\mathrm{res}}, S_{\mathrm{res}}\right) \doteq \BFV(\underline{\Co},\Cres),
    \]
    where $\underline{\mathbb{BFV}(F,\Co,C)}$ is the reduction of the coisotropic embedding $\CCo \hookrightarrow \F(F,C)$ in the BFV resolution of the coisotropic embedding $C\hookrightarrow F$, and $\BFV(\underline{\Co},\Cres)$ is the BFV resolution of the reduction of the coisotropic embedding $\Cres\hookrightarrow \underline{\Co}$ (cf. Equation \ref{e:BFVSmallRes}).
\end{theorem}

\begin{proof}
    We begin by writing down $\mathbb{Q}$:
    \begin{align*}
    \mathbb{Q}(q)&=\left\langle\rho, \frac{\partial\phi}{\partial p}\right\rangle, &  \mathbb{Q}(\lambda)&=g(\rho,\lambda), & \mathbb{Q}(\chi)&=\mu,\\  
    \mathbb{Q}(p)&=\left\langle\rho, \frac{\partial\phi}{\partial q}\right\rangle, & \mathbb{Q}(\lambda^\dag)&=g^*(\rho,\lambda^\dag), &\mathbb{Q}(\chi^\dag)&=0,\\ \mathbb{Q}(\mu) &= 0, & \mathbb{Q}(\rho) &= \frac12 \llbracket\rho,\rho\rrbracket, \\ \mathbb{Q}(\mu^\dag) &= \chi^\dag, & \mathbb{Q}(\rho^\dag) &= \phi + g^*(\lambda,\lambda^\dag) + f^*(\rho,\rho^\dag).
    \end{align*}
    Then, if the coisotropic embedding is of Lie type we have constant structure functions, and using the identities among them provided by Lemma \ref{l:Relations_structure_constants}, it is a straightforward calculation to show that $\check{S}_C$ as given in \eqref{e:generalSC} is a $\mathbb{Q}$-cocycle:
    \[
    \mathbb{Q}(\check{S}_C) = 0.
    \]
    Similarly one shows that $\frac12\{\check{S}_C,\check{S}_C\}= 0$.

    The fact that $\underline{\CCo}$ and $\F_{\mathrm{res}}=\F(\underline{\Co},\Cres)$ are diffeomorphic as graded symplectic manifolds is the content of Theorem \ref{thm:doubleBFV}. Then we just need to show that they the respective Hamiltonian dg-data is equivalent. The residual BFV resolution has master action
    \[
    S_{\mathrm{res}} = \langle\underline{\lambda},\underline{\psi}\rangle + \frac12 \langle h(\underline{\lambda},\underline{\lambda}),\underline{\lambda}^\dag\rangle,
    \]
    where $(\underline{\lambda},\underline{\lambda}^\dag)$ denote coordinates in $T^*W_{\mathrm{res}}$ and the $\underline{\psi}$ are the residual generators of $\Cres$.
    
    Then, considering the following diagram:
    \[
    \xymatrix{
        & \FF = \F(F,C) \times T^*\mathbb{R}^{\ell_1+\ell_2}\\
        & \CCo \ar[u]_{\bbiota_\circ} \ar[d]^{\bbpi_\circ}\\
    \F(\underline{\Co},\Cres) & \ar[l]^-\sim \underline{\CCo}
    }
    \]
    one can lift $S_{\mathrm{res}}$ from $\F(\underline{\Co},\Cres)$ to a function $\hat{S}_{\mathrm{res}}$ on $\CCo$, and compare it to the restriction to $\CCo$ of functions in $\FF$.
    
    We use this construction to show that there exists $Y\in C^\infty(\FF)$ such that
    \[
    \bbiota^*_\circ(\check{S}_C + \mathbb{Q}(Y)) = \bbpi^*_\circ S_{\mathrm{res}}.
    \]    
    (Recall that $[\check{S}_C]_{\mathbb{Q}}\in C^\infty(\underline{\CCo})$ by construction of the resolution of $\underline{\CCo}$.)
    Indeed, picking 
    \[
    Y=\langle\chi,\rho^\dag\rangle + \frac12 \langle \llbracket\chi,\chi\rrbracket,\mu^\dag\rangle + \langle m(\chi,\lambda),\mu^\dag\rangle
    \]
    one can easily show that (assuming again constant structure functions and using Lemma \ref{l:Relations_structure_constants})
    \[
    \check{S}_C + \mathbb{Q}(Y) = \langle{\lambda},{\psi}\rangle + \frac12 \langle h({\lambda},{\lambda}),{\lambda}^\dag\rangle + \langle m(\lambda,\rho),\rho^\dag\rangle +\langle m(\chi, g(\rho,\lambda)), \mu^\dag \rangle\doteq \wt{S}_{\mathrm{res}}.
    \]
    Then one has that
    \[
    \bbiota^*_\circ \wt{S}_{\mathrm{res}} =  \langle{\lambda},{\psi}\rangle + \frac12 \langle h({\lambda},{\lambda}),{\lambda}^\dag\rangle
    \]
    which is manifestly the pullback along $\bbpi_\circ$ of $S_{\mathrm{res}}$, since $\underline{\lambda}$ and $\underline{\lambda}^{\dag}$ are identified with the equivalence classes $[\underline{\lambda}]$ and $[\lambda^{\dag}]$ defined by the ODE's \eqref{e:simplifiedgradedODE}.
\end{proof}

\begin{remark}[Generalisation]\label{rmk:generalisation}
    {We believe that a generalisation of the previous theorem should be available for all nested coisotropic embeddings, not just those of Lie type. Conjecturally, though, one will have a slightly weaker condition on the master equation for $\check{S}_C$, namely that it is satisfied up to $\mathbb{Q}$-exact terms, as discussed above.}
    
    Observe that on $C^\infty(\FF)$ there exist both $\mathbb{Q}= \{\mathbb{S},\cdot\}$ as well as a lift $\check{Q}_C$ of the cohomological vector field $Q_C$ on $\F$. This corresponds to a choice of $\check{S}_C$ via $\check{Q}_C=\{\check{S}_C,\cdot\}$. Moreover, there is the total vector field that one gets by adding the two functions together $Q_{tot}\doteq\{\check{S}_C+\mathbb{S},\cdot\}$.
    
    For nested coisotropic embeddings of Lie type one has that both $\check{Q}_C$ and $Q_{tot}$ are cohomological. In a more general scenario, however, neither will necessarily be, if the lift of $S_C$ is chosen as above. If we wanted to find a total differential on $\FF$, we would need to perform some sort of homological perturbation, which is to be expected, since our answer for $\check{S}_C$ uses the assumption of Lie-type coisotropic embeddings.
    
    In our scenario, but conjecturally in greater generality, one has that 
    \[
    H^0(\FF,\mathbb{Q}_\mathrm{tot})\simeq H^0(H^0(\FF,\mathbb{Q}),Q_{\mathbb{\CCo}})\simeq H^0(\F_{\mathrm{res}},Q_{\mathrm{res}})\simeq H^0(\F,Q_C)\simeq C^\infty(\underline{C}).
    \]
\end{remark}

The double resolution procedure is summarised by the following diagram, where wavy arrows denote BFV resolution of the large triangle diagram to their left (cf. Equation \ref{e:BFVdiagrams}):
\begin{equation}
    \xymatrix{
     F & \ar@{_{(}->}[l] \Co\ar[d] & \ar@{_{(}->}[l]\ar[d]C \ar@{~>}[r]& \F(F,C) & \ar@{_{(}->}[l]\CCo \ar[d] \ar@{~>}[r]& \F(\F(F,C),\CCo)\equiv\FF(F,\Co,C) \\
     & \underline{\Co} & \ar@{_{(}->}[l]\ar[d] \Cres \ar@{~>}[r] &\F(\underline{\Co},\Cres) \ar[r]^-{\sim}& \underline{\CCo} &\\
     && \underline{\Cres}\simeq \underline{C} & &
    \ar@{.}^{BFV_{tot}}(-5,5);(42,5)
    \ar@{.}(-5,5);(-5,-33)
    \ar@{.}(-5,-33);(42,-33)
    \ar@{.}(42,5);(42,-33)
    \ar@{.}^{BFV_{\underline{\psi}}}(5,-10);(41,-10)
    \ar@{.}(5,-10);(5,-32)
    \ar@{.}(5,-32);(41,-32)
    \ar@{.}(41,-10);(41,-32)
    \ar@{.}^{BFV_{double}}(48,5);(85,5)
    \ar@{.}(48,5);(48,-20)
    \ar@{.}(48,-20);(85,-20)
    \ar@{.}(85,5);(85,-20)
    }
\end{equation}

\subsection{Quantisation in double BFV}\label{s:quantisation_BFVonBFV}

We want now to quantise a BFV theory such as the one described in Section \ref{s:doubleBFV}, i.e.\ we are given some BFV data $\BFV(F,C)=(\F(F,C),\omega_C,Q_C,S_C)$ with $\{S_C,S_C\}=0$, and the double BFV data $\mathbb{BFV}(F,\Co,C)=(\FF(F,\Co,C),\bbomega,\mathbb{Q},\mathbb{S})$ associated with the resolution of the nested coisotropic embeddings $C\hookrightarrow \Co \hookrightarrow F$, together with an extension $\check{S}_C$ of $S_C$ such that $\frac12\{\check{S}_C+ \mathbb{S},\check{S}_C+\mathbb{S}\}=\frac12\{\check{S}_C,\check{S}_C\}=0$ (or, more generally $\frac12\{\check{S}_C+ \mathbb{S},\check{S}_C+\mathbb{S}\}=\frac12\{\check{S}_C,\check{S}_C\}=\mathbb{Q}(f)$ for some $f$, see Remark \ref{rmk:generalisation}). Let us employ the shorthand $\FF(F,\Co,C)\equiv \FF$ and $\F(F,C)\equiv \F$.

Suppose that we have a quantisation of the double BFV data $\mathbb{BFV}(F,\Co,C)$, i.e.\ a graded vector space $\mathbb{V}$, as well as a Lie algebra morphism 
\begin{equation}\label{e:quantummorphism}
\mathfrak{q}\colon \mathcal{A}\subset C^\infty(\FF)\to \mathrm{End}(\mathbb{V})
\end{equation}
from a suitable subalgebra $\mathcal{A}$ onto a suitable space of operators over $\mathbb{V}$, endowed with a coboundary operator $[\mathbb{\Omega},\cdot]$, i.e.\ $[\mathbb{\Omega},\mathbb{\Omega}]=\mathbb{\Omega}{}^2=0$. Often this can be obtained as the geometric quantisation of $\FF$, and $\mathbb{\Omega}$ is the quantisation of the function $\mathbb{S}$, so that the cohomology in degree zero of the operator $[\mathbb{\Omega},\cdot]$ is a quantisation of $C^\infty(\underline{\CCo})$, which is the cohomology in degree zero of $\mathbb{Q}\equiv\{\mathbb{S},\cdot\}$.

Note, however, that one obtains the graded vector space $\mathbb{V}_\circ\doteq H^0(\mathbb{V},\mathbb{\Omega})$, thought of as a quantisation of $\underline{\CCo}\simeq \F(\underline{\Co},\Cres)$, but with no differential. However, the BFV quantisation of the BFV data $\BFV(\underline{\Co},\Cres)=\left(\F(\underline{\Co},\Cres),\omega_{\mathrm{res}},Q_{\mathrm{res}},S_{\mathrm{res}}\right)$ should come equipped with a coboundary operator on $\mathrm{End}(\mathbb{V}_\circ)$. Normally, this is expected to arise as the (geometric) quantisation of the function $S_{\mathrm{res}}\in C^\infty(\F(\underline{\Co},\Cres))$, but via double BFV we can obtain one through an operator $\check{\Omega}_C\in \mathrm{End}(\mathbb{V})$ obtained as $\check{\Omega}_C\doteq \mathfrak{q}(\check{S}_C)$ such that 
\begin{subequations}\label{e:conditions_Omega_descends}
    \begin{align}\label{e:conditions_Omega_descends1}
    [\check{\Omega}_C, \mathbb{\Omega}]=0 \\
    \check{\Omega}_C^2\equiv[\check{\Omega}_C,\check{\Omega}_C] = A \mathbb{\Omega} + \mathbb{\Omega}B    \label{e:conditions_Omega_descends2}
\end{align}
\end{subequations}
for some operators $A,B$. 

Indeed, the first condition guarantees that $\check{\Omega}_C$ descends to $[\check{\Omega}_C]_{\mathbb{\Omega}}$ acting on  $\mathbb{V}_\circ\equiv H^0(\mathbb{V},\mathbb{\Omega})$, and the second condition implies that its $\mathbb{\Omega}$-cohomology class squares to zero: $[\check{\Omega}_C]_{\mathbb{\Omega}}^2=0$. Then, a quantisation of the BFV data resolving the reduction $ \underline{\CCo}$ is given by the couple $\left(\mathbb{V}_\circ, [\check{\Omega}_C]_{\mathbb{\Omega}}\right)$ where $\check{\Omega}_C$ is a quantisation of $\check{S}_C$, hence, $[\check{\Omega}_C]_{\mathbb{\Omega}}$ can be viewed as a quantisation of $S_{\mathrm{res}}$. (Note that this was proven explicitly only for nested coisotropic embeddings of the Lie type. However, our main example is of this type.)

\begin{theorem}[Quantisation of double BFV data]\label{thm:doubleBFVquant}
    Let $C\hookrightarrow \Co \hookrightarrow F$ be a nested coisotropic embedding {of Lie type}, and $\FF(F,\Co,C)$ its double BFV resolution. If the Lie algebra morphism $\mathfrak{q}$ of Equation \eqref{e:quantummorphism} exists such that $\mathbb{\Omega}^2=0$, and if $\check{\Omega}_C$ exists so that Equations \eqref{e:conditions_Omega_descends} are satisfied, then a quantisation of $\underline{C}$ is given as
    \[
    H^0\left(H^0\left(\mathbb{V},\mathbb{\Omega}\right),[\check{\Omega}_C]_{\mathbb{\Omega}}\right).
    \]
\end{theorem}

The following lemma gives an useful tool in order to prove the equations \eqref{e:conditions_Omega_descends1} and \eqref{e:conditions_Omega_descends2} for some classes of operator.

\begin{lemma}\label{l:operator_and_commutator}
    Let $\check{\Omega}'_C = \check{\Omega}_C + [\mathbb{\Omega},Z]$. Then $\check{\Omega}'_C$ satisfies \eqref{e:conditions_Omega_descends1} and \eqref{e:conditions_Omega_descends2} if and only if $\check{\Omega}$ does, with 
    \[
    A' = A+ 2[Z, \check{\Omega}_C]+[Z,[\mathbb{\Omega},Z]], \quad B' = B+ 2[Z, \check{\Omega}_C]+[Z,[\mathbb{\Omega},Z]].
    \]
\end{lemma}
\begin{proof}
    We drop the $C$ subscript. Let us suppose that $\check{\Omega}$ satisfies \eqref{e:conditions_Omega_descends1} and \eqref{e:conditions_Omega_descends2}. Then we have:
    \begin{align*}
        [\check{\Omega}', \mathbb{\Omega}]= [\check{\Omega}, \mathbb{\Omega}]+ [[\mathbb{\Omega},Z], \mathbb{\Omega}] = \frac{1}{2}[[\mathbb{\Omega},\mathbb{\Omega} ],Z] =0
    \end{align*}
    where we used the (graded) Jacobi identity. For the second equation we have
    \begin{align*}
        (\check{\Omega}')^2 &\equiv [\check{\Omega}',\check{\Omega}']= [\check{\Omega},\check{\Omega}] + 2[\check{\Omega}, [\mathbb{\Omega},Z]] + [[\mathbb{\Omega},Z],[\mathbb{\Omega},Z]]\\
        &= A\mathbb{\Omega} + \mathbb{\Omega}B + 2[\mathbb{\Omega},[Z, \check{\Omega}]] + 2[Z, [\check{\Omega},\mathbb{\Omega}]]+ [\mathbb{\Omega},[Z,[\mathbb{\Omega},Z]]]+[Z,[[\mathbb{\Omega},Z],\mathbb{\Omega}]] \\
        &= A\mathbb{\Omega} + \mathbb{\Omega}B + \mathbb{\Omega} (2[Z, \check{\Omega}]+[Z,[\mathbb{\Omega},Z]])+ (2[Z, \check{\Omega}]+[Z,[\mathbb{\Omega},Z]])\mathbb{\Omega} + \frac{1}{2}[Z,[[\mathbb{\Omega},\mathbb{\Omega}],Z]]\\
        &= \mathbb{\Omega} (B + 2[Z, \check{\Omega}]+[Z,[\mathbb{\Omega},Z]])+ (A + 2[Z, \check{\Omega}]+[Z,[\mathbb{\Omega},Z]])\mathbb{\Omega}\\
        &= \mathbb{\Omega}B' + A' \mathbb{\Omega}.
    \end{align*}
\end{proof}

\begin{remark}
    Lemma \ref{l:operator_and_commutator} helps to understand condition \eqref{e:conditions_Omega_descends2}. Indeed, not all elements in $[\check{\Omega}]$ square to zero, i.e.\ such property is not invariant under change of representative. Indeed, if we start with an operator $\check\Omega_0$ that squares to zero, we can add an element in the image of $\mathbb{\Omega}$ and still remain in the same cohomology class, while the new $\check{\Omega}' \doteq \check{\Omega}_0 + [\mathbb{\Omega},Z]$ will not square to zero, unless $2[Z,\check{\Omega}_0] + [Z,[\mathbb{\Omega},Z]] = 0$.
    Furthermore, since we would like $\check{\Omega}$ to be a quantisation of a function defined in the $\mathbb{Q}$-cohomology,\footnote{Recall that $\mathbb{Q}=Q_{\CCo}$.} it makes sense that this should be defined only up to the quantisation of an element in the image of $\mathbb{Q}$, which correspond to the commutator of $\mathbb{\Omega}$ and some operator $Z$ after quantisation. Then, Lemma \ref{l:operator_and_commutator} tells us that the quantisation is only required to square to zero in the $\mathbb{\Omega}$-cohomology, which is exactly the content of Equation \eqref{e:conditions_Omega_descends2}.
\end{remark}

\subsection{Example: linear constraints} \label{s:example_linearconstraints}
As a simple example we can consider a finite dimensional case with linear constraints, where one can explicitly show that quantisation commutes with reduction, which in turn commutes with resolution.

Let $V$ be a finite dimensional $\mathbb{R}$-vector space and $ F= V^{\times 6}$, with $\omega= \sum_{i=1}^3 da_i db_i$ where $a_i, b_i$ are coordinates on $V^{\times 6}$. We define linear functions $\phi= \chi b_2$ and $\psi = \lambda b_3$, which clearly form a Poisson subalgebra and an associative ideal, so that their vanishing locus is coisotropic. 
We can then resolve $C= C_{\phi} \cap C_{\psi}$ using the BFV procedure as in Assumption \ref{ass:BFV_action} and get, on $\F(C,F)=F \times T^*\mathbb{R}^2$:
\begin{align*}
   \omega_C&= \omega + d \chi d\chi^\dag + d\lambda d\lambda^\dag \\
   S_c& = \chi b_2 + \lambda b_3
\end{align*}
and $Q_c\in \mathfrak{X}(\F)^1$ has the only non zero components 
\begin{align*}
    Q_C \chi^\dag &= b_2 & Q_C \lambda^\dag &= b_3  \\
    Q_C a_2 &= \chi & Q_C a_3 &= \lambda.
\end{align*} 
Following Theorem \ref{thm:BFVcoisotropic} we define on $\mathbb{F}(C,F)$
\begin{align*}
    M^{\mu} &= \mu \chi^\dag  & L^{\rho} = J^{\rho} = \rho b_2.
\end{align*}
All the brackets of these functions with respect to $\omega_C$ are again 0, and hence they define a coisotropic submanifold $\mathbb{C}_\circ$. It is then simple to check that $\mathsf{Body}(\underline{\CCo})\simeq \underline{\Co}$ and that $\underline{\Co}= V^{\times 4}$ with 
\begin{align*}
    \omega_{\underline{\Co}} = da_1 db_1 + da_3 db_3.
\end{align*}
In $\underline{\Co}$, $\Cres$ is defined by $\psi$ and its BFV resolution is given by
$\omega_{\mathrm{res}}= \omega_{\underline{\Co}} + d\lambda d \lambda^\dag$, 
$S_{\mathrm{res}}= \lambda b_3$ and $Q_{\mathrm{res}}$ has only nonzero components $Q_{\mathrm{res}}a_3 = \lambda$ and $Q_{\mathrm{res}}\lambda^\dag = b_3$.

The double BFV data associated to the nested coisotropic embedding $C\hookrightarrow \Co \hookrightarrow F$ is given by $\left(\FF,\bbomega,\mathbb{Q},\mathbb{S}\right)$ where
\begin{align*}
    \bbomega & = \omega_C + d\mu d\mu^\dag + d \rho d\rho^\dag\\
    \mathbb{S} & = \mu \chi^\dag + \rho b_2
\end{align*}
and $\mathbb{Q}$ has the only non zero components 
\begin{align*}
    \mathbb{Q} \chi &= \mu & \mathbb{Q} \rho^\dag &= b_2  \\
    \mathbb{Q} a_2 &= \rho & \mathbb{Q} \mu^\dag &= \chi^\dag.
\end{align*} 
Furthermore $\check{S}_C = S_C + \mu \rho^\dag$ satisfies $\mathbb{Q}(\check{S}_C)=0$ and $\{\check{S}_C,\check{S}_C\}=0$.

It is now easy to check that $\check{S}_C + \mathbb{Q}Y = \widetilde{S}_{\mathrm{res}}$ with $Y= \chi\rho^\dag$ and that
    \[
    H^0(\FF,\mathbb{Q}_\mathrm{tot})\simeq H^0(H^0(\FF,\mathbb{Q}),Q_{\mathbb{\CCo}})\simeq H^0(\F_{\mathrm{res}},Q_{\mathrm{res}})\simeq H^0(\F,Q_C)\simeq C^\infty(\underline{C}).
    \]
    where in this case $\underline{C}= V^{\times 2}$ with symplectic form $\underline{\omega}= da_1 db_1.$

    We can now quantize this simple theory. There are three possible quantizations, one quantizing directly $\underline{C}$, the second quantizing the BFV theory $(\mathbb{F}, \omega_C, S_C, Q_C)$ and the third using the double BFV approach.

    In the first case, we have to quantize $\underline{C} \simeq V^{\times 2}$ with symplectic form $\omega= da_1 db_1$. Choosing the Lagrangian foliation generated by $\frac{d}{db_1}$, the quantization of $\underline{C}$ will be $\mathbb{V}_{\underline{C}} = C^{\infty}(\mathbb{B})$ where $\mathbb{B} \simeq V$ generated by $a_1$, and the quantization map will be
    \begin{align*}
        f(a_1, b_1) \mapsto f\left(a_1, i\hbar \frac{d}{da_1}\right).
    \end{align*}

    In the second case we quantize $(\mathbb{F}, \omega_C, S_C, Q_C)$, and the quantization of $\underline{C}$ is recovered by $H^0(\mathbb{V}_C, \Omega_C)$ where now, choosing a compatible polarization,  $\mathbb{V}_C$ is generated by $a_i, \chi, \lambda$, the quantization map is given by 
    \begin{align*}
        f(a_i, b_i, \chi, \chi^\dag, \lambda, \lambda^\dag) \mapsto f\left(a_i, i\hbar \frac{d}{da_i}, \chi, i\hbar \frac{d}{d\chi}, \lambda, i\hbar \frac{d}{d\lambda} \right)
    \end{align*}
    and 
    \begin{align*}
        \Omega_C = i\hbar \chi  \frac{d}{da_2} + i\hbar \lambda  \frac{d}{da_3}.
    \end{align*}
    Hence degree zero $\Omega_C$-closed functions of $a_i,\lambda$ and $\chi$ must be independent of $a_2,a_3$ via $\chi \frac{\partial f}{\partial a_2}\equiv0\equiv\lambda \frac{\partial f}{\partial a_3}$, and are also independent of $\lambda$ and $\chi$ as there are no degree $-1$ variables. Hence, they descend to functions in $C^\infty(\mathbb{B})$, i.e. $H^0(\mathbb{V}_C, \Omega_C)\simeq C^\infty(\mathbb{B})$.

    Lastly, let us use the double BFV approach. Now $\mathbb{V}$ is givebn by functions in the variables $a_i, \chi, \lambda, \mu, \rho$, and the quantization map is given by 
    \begin{align*}
        f(a_i, b_i, \chi, \chi^\dag, \lambda, \lambda^\dag, \mu, \mu^\dag, \rho, \rho^\dag) \mapsto f\left(a_i, i\hbar \frac{d}{da_i}, \chi, i\hbar \frac{d}{d\chi}, \lambda, i\hbar \frac{d}{d\lambda}, \mu, i\hbar \frac{d}{d\mu},\rho, i\hbar \frac{d}{d\rho}\right).
    \end{align*}
    The operator $\mathbb{\Omega}$ is 
    \begin{align*}
        \mathbb{\Omega} = i\hbar \mu  \frac{d}{d\chi} + i\hbar \rho  \frac{d}{da_2}
    \end{align*}
    and 
    \begin{align*}
        \check{\Omega}_C = i\hbar \chi  \frac{d}{da_2} + i\hbar \lambda  \frac{d}{da_3} + i\hbar \mu  \frac{d}{d \rho}.
    \end{align*}
    Specifically, one can easily show that
    \[
    \mathbb{\Omega}_{\mathrm{tot}} = \mathbb{\Omega} + \check{\Omega}_C = i\hbar \mu  \frac{d}{d\chi} + i\hbar \rho  \frac{d}{da_2} + i\hbar \chi  \frac{d}{da_2} + i\hbar \lambda  \frac{d}{da_3} + i\hbar \mu  \frac{d}{d \rho}
    \]
    squares to zero, and its cohomology in degree zero, using the same argument used above returns functions of $a_1$ alone (observe that $(\rho + \chi)\frac{\partial f}{\partial a_2}=0$ implies $\frac{\partial f}{\partial a_2}=0$ because $\rho$ and $\chi$ are homogeneous of different total degree).
    
    We can also note that $H^0(\mathbb{V},\mathbb{\Omega})$ is isomorphic to functions of $a_1,a_3,\lambda$ only, with residual differential (recall that the degree-zero cohomology in $\mathbb{\Omega}$ is w.r.t. the grading defined by $\mu,\rho$)
    \[
    [\check{\Omega}_C]_{\mathbb{\Omega}} =  i\hbar \lambda  \frac{d}{da_3},
    \] 
    whence we also get $H^0\left(H^0\left(\mathbb{V},\mathbb{\Omega}\right),[\check{\Omega}_C]_{\mathbb{\Omega}}\right)\simeq C^\infty(\mathbb{B}) \simeq H^0(\mathbb{V},\mathbb{\Omega}_{\mathrm{tot}})$.

\section{BFV data for gravity theories} \label{s:BFVtheories_gravity}
In this section we follow the guideline presented in the previous sections for finite dimensional systems and apply it to the case of field theory, where spaces are infinite dimensional nuclear Fr\'echet manifold. We cannot apply the results directly because of this, and this will require us to demonstrate the various steps of the construction for the specific case directly. Our (motivating) example is three dimensional General Relativity, viewed from several different angles. (Note that we will use the notation $\mathcal{F}$ to denote spaces of fields, as opposed to $F$, as is customary.)

The quantisation of the theory (in the sense of a vector space of physical states) is well defined using the construction we developed in this paper. Note that to check the validity of Theorem \ref{thm:doubleBFVquant} in this case would be tantamount to quantising the reduced phase space of GR, which is well known to be a very challenging task. Hence, we propose an alternative route, which is mathematically viable and motivated by a theorem that holds in finite dimensions.

\subsection{BFV data for Einstein--Hilbert (EH) theory} \label{s:BFVtheory_EH}
In this article we will consider three field theories: Einstein--Hilbert theory, which is a theory of metrics under the action of diffeomorphisms, Palatini--Cartan theory, which is a theory of coframes and $SO(2,1)$ connections under the action of the semidirect product of the Lorentz group and diffeomorphisms, and $BF$ theory, which is a theory of Lie algebra-valued differential forms, invariant under the shift by covariantly-exact forms. The goal of this section is to introduce them at a classical and BV-BFV level and to establish the relations between them.

A full description of the BV-BFV ``enhancement'' of the classical Einstein--Hilbert theory has been done for a general dimension $d>2$ in \cite{CS2016b}. We recall here the most important results (for another take on the BFV data associated with EH theory see also \cite{BSW}).

We denote by $\BFV_{EH}=(\F_{EH}, \varpi_{EH}, Q_{EH}, S_{EH})$ the BFV data associated to Einstein--Hilbert theory, where the space of BFV fields $\F_{EH}$ is the graded $0$-symplectic manifold\footnote{In the remainder of the paper, cotangent fibres will be given by sections of the ``densitised dual bundle''.}
\begin{equation}\label{e:boundaryfields_EH}
\F_{EH} = T^*\underbrace{\left(S_{nd}^2(T\Sigma) \times \mathfrak{X}[1](\Sigma) \times C^\infty[1](\Sigma)\right)}_{\{{g}, \txi{},\txi{n}\}},
\end{equation}
with $S^2_{nd}(\Sigma)$ denoting the space of non-degenerate symmetric tensor fields of rank two (also known as co-metrics), endowed with the canonical exact symplectic structure:
\begin{equation}\label{e:symplecticstructure_EH}
	\varpi_{EH} = \delta \alpha_{EH}
	=
	\delta \int_{\Sigma} \left(\langle{\Pi}, \delta{g} \rangle + \langle\phip,\delta\txi{}\rangle + \langle\phin,\delta\txi{n}\rangle\right),
\end{equation}
where 
\begin{align}\label{e:fiberfields_EH}
    \{\Pi,\phip,\phin\}\in \left(S^2(T^*\Sigma)\oplus \Omega^1[-1](\Sigma) \oplus C^\infty[-1](\Sigma)\right)\otimes\mathrm{Dens}(\Sigma)
\end{align}
denote variables in the cotangent fiber, respectively conjugate to $\{{g},\txi{},\txi{n}\}$. We stress that here we are considering cometrics ${g}$, so that ${g}^{-1}\in \mathcal{R}(\Sigma)$ is a Riemannian metric.\footnote{This picture can be enhanced to the case where ${g}^{-1}$ is Lorentzian. The two cases correspond to ${g}=g_{M}\vert_\Sigma$ being the restriction to a spacelike or timelike hypersurface in some ambient Lorentzian manifold $(M,g)$. See \cite{CS2016b}.} 

The BFV action $S_{EH}$ is  a functional of degree 1 on $\F_{EH}$, given by the local expression
\begin{align}\label{ADMBoundaction}
S_{EH}=&\int_{\Sigma} \left\{ H_n\txi{n} + \langle{\Pi}, L_{\txi{}}{g}\rangle  + {\phin}L_{\txi{}}\txi{n} - {g}(\phip,d\txi{n})\txi{n} + \left\langle\phip,\frac12[\txi{},\txi{}]\right\rangle\right\}
\end{align}
where we denoted by $\mathrm{Tr}_{g}[{\Pi}^2]={g}^{\mu\nu}{g}^{\rho\sigma}{\Pi}_{\nu\rho}{\Pi}_{\mu\sigma}$ and $\mathrm{Tr}_{g}{\Pi}={g}^{\mu\nu}{\Pi}_{\mu\nu}$ and
\begin{equation}\label{Hamiltonianconstraint}
H_n({g},{\Pi}) = \left(\frac{1}{{\mathrm{vol}_g}}\left(\mathrm{Tr}_{g}[{\Pi}^2] - \frac{1}{d-1}\mathrm{Tr}_{g}[{\Pi}]^2\right) + {\mathrm{vol}_g}\left(R^{\Sigma} -2\Lambda\right)\right)
\end{equation}
where $R^{\Sigma}$  is the trace of the Ricci tensor with respect to the metric ${g}^{-1}$, the cosmological constant is $\Lambda\in \mathbb{R}$ and $\mathrm{vol}_g=\sqrt{\mathsf{g}}\,\mathrm{vol}$, where $\sqrt{\mathsf{g}}$ is the square root of the determinant of the metric $g^{-1}$. Finally, $Q_{EH}$ is the Hamiltonian vector field of $S_{EH}$ with respect to $\varpi_{EH}$.

In some of the computations in the next sections it will be useful to have an expression of the boundary action and symplectic form with indices w.r.t. a local chart spelled-out\footnote{Note that the superscript ${}^n$ is not a space-time index, as there are no transverse directions. It is reminiscent of the bulk-to-boundary induction procedure, but it has here only a symbolic meaning.} (with a slight abuse of notation to denote with the same symbol both a density valued tensor and its chart components against the fixed euclidean volume form $\mathrm{vol}$, see \cite[Remark 3.7]{BSW} for a detailed discussion)
\begin{align}
    S_{EH}=&\int_{\Sigma}
    \left\lbrace
        \frac{1}{\sqrt{\mathsf{g}}} \left({\Pi}^{\mu\nu}{\Pi}_{\mu\nu} - \frac{1}{d-1}{\Pi}^2 \right) 
        + \sqrt{\mathsf{g}}\left(R^\Sigma -2\Lambda\right) 
        + \partial_{\mu}\left(\txi{\mu}{\phin}\right) 
        - {g}^{\mu\nu}{\varphi}_{\nu}\partial_{\mu}\txi{n}
    \right\rbrace \txi{n} \mathrm{vol} \nonumber\\ 
    & + \int_{\Sigma}
    \left\lbrace
        - 2\partial_{\rho}\left({g}^{\rho \sigma}{\Pi}_{\sigma \mu}\right) 
        - (\partial_{\mu}{g}^{\rho \sigma}){\Pi}_{\rho \sigma} 
        + \partial_{\rho} \left(\txi{\rho}{\varphi}_{\mu}\right) 
    \right\rbrace \txi{\mu}\mathrm{vol}.\label{e:ADMBoundaction}
\end{align}
The corresponding symplectic form is:
\begin{align*}
    \varpi_{EH}=\int_{\Sigma}\left(\langle\delta {g},\delta \Pi\rangle + \langle \delta\zeta,\delta \varphi\rangle\right)=\int_{\Sigma} \left(\delta {g}^{\mu\nu}\delta{\Pi}_{\mu\nu} + \delta\txi{\rho}\delta{\varphi}_{\rho}+ \delta\txi{n}\delta{\varphi}_{n}\right)\mathrm{vol}.
\end{align*}

\begin{remark}
    The BFV data described here arises from the application of the general procedure outlined in the previous sections to\footnote{We denote spaces of fields with calligraphic letters to emphasise that they are infinite dimensional.} $\mathcal{F}_{EH}=T^*S_{nd}^2(T\Sigma)$ to resolve the reduction of $\mathcal{C}_{EH}\hookrightarrow\mathcal{F}_{EH}$, defined as the vanishing locus of $H_n$ and $H_{\mu}\doteq\Pi_{\rho\sigma}\partial_{\mu} {g}^{\rho\sigma} - \partial_\rho({g}^{\rho\sigma}\Pi_{\sigma\mu})$. 
    \[
    \BFV_{EH}\doteq \BFV(\mathcal{F}_{EH}, \mathcal{C}_{EH}) = (\F_{EH}, \varpi_{EH}, Q_{EH}, S_{EH})
    \]

    (Note that $\mathcal{C}_{EH}$ is recovered from the BFV data presented above by setting $H_n\txi{n} = 0$ for all $\txi{n}$ and $\langle\Pi,L_{\txi{}} {g}\rangle = 0$ for all $\txi{}$ up to a boundary term. Then, the space of BFV fields is $\F_{EH} \doteq \F(\mathcal{F}_{EH}, \mathcal{C}_{EH})$, and the rest of the data is constructed accordingly. We use $\varpi$ as a symbol for symplectic forms hereinafter, in order to avoid clashes with the symbol used for connection forms later on.
\end{remark}

\subsection{BFV data for \texorpdfstring{$BF$}{BF} theory} \label{s:BFVtheory_BF}
Consider an orientable $3$-dimensional manifold $M$ endowed with a trivial vector bundle $\mathcal{V}\to M$ whose fibres $V$ are equipped with Minkowski metric $\eta$.
The classical fields of $BF$ theory are $B \in \Omega^1_{nd}(M, \mathcal{V})$, where $\Omega^1_{nd}(M, \mathcal{V})$ denotes the space of bundle isomorphisms $TM\to \mathcal{V}$, also called ``triads'' or coframes, and a principal $SO(2,1)$ connection $A \in \mathcal{A}_P$. Note that we have an identification of $\mathfrak{so}(2,1)\simeq \wedge^2V\simeq V^*$. Hence, we can think of $A\vert_U$ as a (local) Lie-algebra valued one-form $A\vert_U\in\Omega^1(U,\wedge^2\mathcal{V})$, for $U\subset M$. For simplicity, we will assume principal bundles to be trivial, so that this local expression can be globalised.

In this article we are interested in the BFV theory associated to $BF$ theory. This is obtained by looking at a (compact) codimension 1 submanifold $\Sigma\hookrightarrow M$ (without boundary), by looking at the restriction of the above mentioned fields to $\Sigma$, which defines the space (of degree 0 fields on the codimension 1 stratum) $\mathcal{F}_{BF}\doteq\Omega^1_{nd}(\Sigma,\mathcal{V}) \times \mathcal{A}_{P_\Sigma}$, where $P_\Sigma$ denotes the induced principal bundle on $\Sigma$, together with a submanifold $\mathcal{C}_{BF}\hookrightarrow \mathcal{F}_{BF}$ given by flat $\wedge^2\mathcal{V}$-valued connections on $\Sigma$ and covariantly constant restricted $B$ fields.
The BFV procedure then leads to the following definition. (We refer to \cite{CaSc2019} for a complete description of nondegenerate $BF$ theory in three dimensions.)
\begin{definition}\label{d:BFV_of_BFtheory}
The BFV data for $BF$ theory is given by the quadruple 
$$\BFV_{BF}\doteq\BFV(\mathcal{F}_{BF},\mathcal{C}_{BF})=\left(\F_{BF},\varpi_{BF},Q_{BF},S_{BF}\right),$$ 
where the BFV space of fields can be written as  
$$\F_{BF} \doteq \F(\mathcal{F}_{BF},\mathcal{C}_{BF}) = \Omega^1_{nd}(\Sigma, \mathcal{V}) \oplus \mathcal{A}_{P_\Sigma} \oplus T^* \left(\Omega^0[1]( \Sigma, \wedge^2\mathcal{V})\oplus \Omega^0[1]( \Sigma, \wedge^1\mathcal{V})\right),
$$
and we denote the fields by 
\begin{align}\label{e:boundaryfields_BF}
 A &\in \Omega^1(\Sigma,\wedge^2\mathcal{V}), & \chi &\in \Omega^0[1]( \Sigma, \wedge^2\mathcal{V}), &  B^\dag &\in \Omega^2[-1]( \Sigma, \wedge^2\mathcal{V})\\
 B &\in \Omega^1_{nd}(\Sigma, \mathcal{V}), & \tau &\in \Omega^0[1]( \Sigma, \wedge^1\mathcal{V}), & A^\dag &\in \Omega^2[-1]( \Sigma, \wedge^1\mathcal{V}). \nonumber
\end{align}
The corresponding BFV symplectic form and BFV action respectively read\footnote{Observe, however, that in general this is not a cotangent bundle, and only in the case of an induced trivial principal bundle on $\Sigma$ we can look at the structure as defining a ``cotangent bundle'' over triads.}
\begin{align}\label{e:BF_sympl_form}
    \varpi_{BF} &= \trintl{\Sigma} \delta B \delta A + \delta \chi \delta A^\dag + \delta \tau \delta B^\dag\\
    S_{BF} &= \trintl{\Sigma} \chi d_A B + \tau F_A +   \tau[\chi, B^\dag] + \frac{1}{2}[\chi, \chi] A^\dag,\label{e:BF_action}
\end{align}
where the trace signs denotes that we are canonically\footnote{This does not require a choice, as explained in \cite[Footnote 5]{CCS2020}.} identifying top forms on $V$ with $\mathbb{R}$, and
\begin{align}\label{e:coisotropic_submanifold_BF}
    \mathcal{C}_{BF} = \{(A,B)\in\mathcal{F}_{BF} \ | \ F_A=0,\ d_AB = 0\}.
\end{align}

Finally, $Q_{BF}$ is the Hamiltonian vector field of $S_{BF}$.
\end{definition}

\begin{remark}
    A quick computation shows that the components of $Q_{BF}$ are as follows:
    \begin{align*}
        Q_{BF}(B) &= [\chi, B] + d_A \tau & Q_{BF}(\tau) &= [\chi, \tau] & Q_{BF}({A^\dag})&= [\chi, A^\dag] + d_A B + [\tau, B^\dag]\\
        Q_{BF}(A) &= d_A \chi & Q_{BF}(\chi) &= \frac{1}{2}[\chi, \chi] & Q_{BF}({B^\dag}) &= [\chi, B^\dag] + F_A.
    \end{align*}
\end{remark}

\begin{remark}
    The requirement that $B$ be nondegenerate as a map $B\colon TM \to \mathcal{V}$, when restricted to $\Sigma$ is tantamount to asking that the image of $B$ be a two-dimensional subspace on the fibers of $\mathcal{V}|_{\Sigma}$.
\end{remark}

\begin{remark}[$BF$ nested coisotropic data]\label{rmk:Gauss}
    Observe that the constraint set for $BF$ theory is composed of flat connections and covariantly constant $B$-fields. We can thus look at the nested embedding
    \[
    \mathcal{C}_{BF}\hookrightarrow \mathcal{C}_{\mathrm{Gauss}} \hookrightarrow \mathcal{F}_{BF}
    \]
    where
    \begin{align}\label{e:gaussian_coiso_BF}
        \mathcal{C}_{\mathrm{Gauss}} \doteq \{(A,B)\in\mathcal{F}_{BF}\ |\ d_AB = 0\}.
    \end{align}

    This is sometimes denoted the Gauss' constraint, and it generates (as a momentum map) $SO(2,1)$ Gauge transformations.
\end{remark}

\subsection{BFV data for Palatini--Cartan (PC) theory} \label{s:BFVtheory_PC}
Another gauge-theoretic model that describes general relativity goes under the name of Palatini--Cartan theory\footnote{There is an unfortunate disagreement in the literature concerning nomenclature of theories. See \cite[Nomenclature]{CS2019}.} The BV-BFV description of three-dimensional Palatini--Cartan theory has been explored in \cite{CaSc2019}. 

Note that, differently from the Einstein--Hilbert case, the BFV theory for PC theory differs in $d=3$ and $d>3$, and the construction is significantly more involved (see \cite{CS2017,CCS2020}) for the latter. In this article we will restrict to $d=3$, where the theory is strongly equivalent to the \emph{nondegenerate} sector of $BF$ theory (\cite{CaSc2019} building on \cite{CSS2017}) and recall here the relevant results.

\begin{proposition}[\cite{CaSc2019} Proposition 21]\label{prop:BFVdata}
The BFV data for PC theory in dimension $d=3$, denoted $\BFV_{PC}\doteq \BFV(F_{PC},\mathcal{C}_{PC})=(\F_{PC},\omega_{PC},Q_{PC},S_{PC})$ are: 
\begin{itemize}
\item The BFV space of fields, given by the total space of the bundle
\begin{equation}
\F_{PC} \longrightarrow \Omega_{nd}^1(\Sigma, \mathcal{V}),
\end{equation}
with local trivialisation on an open $\mathcal{U} \subset \Omega_{nd}^1(\Sigma, \mathcal{V})$
\begin{equation}\label{LoctrivF1}
\F_{PC}\simeq \mathcal{U} \times \Omega^1( \Sigma, \wedge^2\mathcal{V})\oplus T^* \left(\Omega^0[1]( \Sigma, \wedge^2\mathcal{V})\oplus \mathfrak{X}[1](\Sigma) \oplus C^\infty[1](\Sigma)\right),
\end{equation}
while $F_{PC}\simeq\mathcal{U}\times \Omega^1( \Sigma, \wedge^2\mathcal{V})$, and fields denoted by\footnote{Here we deliberately use a different notation for the generators $(\xi,\xi^n)$ of PC theory and $(\txi{},\txi{n})$ of EH theory, to distinguish them at a glance.}
\begin{gather*}
\te\in \mathcal{U}, \quad \tom\in\Omega^1(\Sigma,\wedge^2\mathcal{V}), \quad \text{in degree zero},\\ 
\tc\in\Omega^0[1](\Sigma,\wedge^2\mathcal{V}), \quad \xi\in\mathfrak{X}[1](\Sigma),\quad \xi^n\in C^\infty[1](\Sigma), \quad \text{in degree one},\\ 
\tom^\dag\in\Omega^2[-1](\Sigma,\mathcal{V}), \quad \te^\dag\in\Omega^{2}[-1](\Sigma,\wedge^2 \mathcal{V}),\quad  \text{in degree minus one}
\end{gather*}
together with a fixed vector field $\epsilon_n \in \Gamma(\mathcal{V})$, completing the image of elements $\te\in\mathcal{U}$ to a basis of  $\mathcal{V}$, in terms of which we write
\[
\mathcal{C}_{PC}=\{(e,\omega)\in F_{PC}\ |\ d_\omega e =0, \ e_a F_\omega = \epsilon_n F_\omega = 0\} \hookrightarrow F_{PC};
\]
\item The codimension-$1$ one-form, symplectic form and action functional
\begin{subequations}\label{GR_BFV_data}\begin{align}
\alpha_{PC}
    &= \trintl{\Sigma} - \te \delta \tom+ \tom^\dag \delta \tc -\te^{\dag} \epsilon_n \delta \xi^n- (\iota_{\delta \xi} \te) \te^{\dag}+ \iota_{\xi}\tom^\dag \delta \tom,\\\label{BoundarytwoformGR}
\varpi_{PC}
    &= \trintl{\Sigma} - \delta\te \delta \tom+  \delta\tom^\dag \delta \tc - \delta\te^{\dag} \epsilon_n \delta \xi^n+ \iota_{\delta \xi} \delta({\te_\bullet\otimes} \te^{\dag})+ \delta(\iota_{\xi}\tom^\dag )\delta \tom,\\\label{BoundaryactionGR}
S_{PC}
    &= \trintl{\Sigma} - \iota_{\xi} \te F_{\tom} -\epsilon_n \xi^n F_{\tom} - \tc d_{\tom} \te + \frac12 [\tc,\tc]\tom^{\dag} +\frac12 \iota_{\xi} \iota_{\xi} F_{\tom}\tom^\dag+ \frac12 \iota_{[\xi,\xi]}\te\te^\dag \nonumber\\ 
        &\phantom{\trintl{\Sigma}}+\tc d_{\tom}( \iota_{\xi} \tom^\dag)+ L_{\xi}^{\tom} (\epsilon_n \xi^n) \te^\dag- [\tc, \epsilon_n \xi^n]\te^\dag,
\end{align}\end{subequations}
where $\iota_{\delta \xi} \delta({\te_\bullet\otimes} \te^{\dag})$ indicates that the contraction acts on the one-form\footnote{{In a local chart we can write $\iota_{\delta \xi} \delta({\te_\bullet\otimes} \te^{\dag}) = \delta\xi^a \delta(e_a e^\dag)$.}} $\te$, and the trace symbol has the same meaning as in Definition \ref{d:BFV_of_BFtheory};
\item The cohomological vector field $Q_{PC}$, such that $\iota_{Q_{PC}}\varpi_{PC}=\delta S_{PC}$. 
\end{itemize}
\end{proposition}

The relation between $BF$ theory and PC theory has been fully described in \cite[Theorem 31]{CaSc2019}. In particular:

\begin{proposition}[\cite{CaSc2019} Section 4.2]
    There is an isomorphisms of complexes 
    \[
        \psi\colon (C^\infty(\F_{PC}),Q_{PC})\stackrel{\sim}{\to} (C^\infty(\F_{BF}),Q_{BF}), \qquad \psi\circ Q_{PC} = Q_{BF}\circ \psi,
    \]
whose explicit expression on fields is:\footnote{With reference to \cite{CaSc2019}, this is $\psi^{\filt{1}}$. Note that we do not write everything explicitly for the sake of brevity of expressions. It is straightforward to expand, e.g., $\iota_{\xi}A^\dag$ using the expression for $\xi$.}
\begin{equation}\label{e:PCtoBF}
\psi: \begin{cases}
 B= \te - \iota_{\xi} \tom^\dag  &\\
 B^\dag= \te^\dag &\\
 A = \omega - \iota_{\xi} \te^\dag &\\
  A^\dag = \tom^\dag &\\
 \chi = -\tc +\frac{1}{2}\iota^2_{\xi} \te^\dag  &\\
 \tau = -\iota_{\xi} \te - \epsilon_n \xi^n + \frac{1}{2} \iota^2_{\xi} \tom^\dag &\\
\end{cases}
(\psi)^{ -1}: \begin{cases}
\te= B+ \iota_{\xi}A^\dag &\\
\xi^{a}= -\tau^{\comp{a}} - \tau^{\comp{a}}\tau^{\comp{b}} A_{ab}^{\dag\comp{a}}&\\
\xi^n= -\tau^{\comp{n}} - \tau^{\comp{a}}\tau^{\comp{b}} A_{ab}^{\dag\comp{n}}&\\
\tom= A+ \iota_{\xi}B^\dag &\\
\tc=-\chi+ \frac{1}{2}\iota_{\xi}\iota_{\xi}B^\dag &\\
\tom^\dag=A^\dag  &\\
\te^\dag=B
^\dag 
\end{cases}
\end{equation}
\end{proposition}

As a final background element we will need, we report a result of Cattaneo and Schiavina, relating the reduced phase space of PC and EH theories, via partial reduction \cite{CS2019}.

\begin{remark}
    The structure of canonical constraints for Palatini--Cartan theory follows the general picture presented in Section \ref{s:doubleBFV}, namely, we have two sets of constraints given by the ``Gauss constraint'' $\mathsf{L}^c=\trintl{\Sigma} cd_\omega e$, which generates Lorentz gauge transformations and the remaining ``diffeomorphism'' constraints $\mathsf{M}^{\wt{\xi}} = \trintl{\Sigma}\left(\iota_{\xi}\te + \xi^n \epsilon_n\right)F_{\tom}$, where $\wt{\xi}=\{\xi^n,\xi\}\in C^\infty(\Sigma)\times \mathfrak{X}(\Sigma)$. The flow of $\mathsf{L}^c$ preserves $\mathsf{M}^{\wt{\xi}}$:
    \begin{align*}
        \{\mathsf{L}^c, \mathsf{M}^{\wt{\xi}}\} = \trintl{\Sigma}\xi^n [c,\epsilon_n] F_\omega = \trintl{\Sigma}\xi^n \left( [c,\epsilon_n]^{(a)}e_a + [c,\epsilon_n]^{(n)}\epsilon_n\right) F_\omega = \mathsf{M}^{\widetilde{\Xi}}
    \end{align*}
    where we defined $\widetilde{\Xi}= [c,\epsilon_n]^{(a)} + [c,\epsilon_n]^{(n)}$ and we used the fact that we can decompose any section of $\mathcal{V}$ using the basis given by $\{e_a, \epsilon_n\}$.
\end{remark}

Note that $\{\mathsf{L}^c, \mathsf{M}^{\wt{\xi}}\}=\mathsf{M}^{\widetilde{\Xi}}$ is again of the $M$ kind, so that we can proceed in stages and reduce first by the Hamiltonian action associated to $\mathsf{L}^c$.
\begin{theorem}[\cite{CS2019}]\label{t:RPS_PCandEH}
    The reduced phase space for Palatini–-Cartan theory is isomorphic to that of Einstein–-Hilbert theory. In particular, let $\mathcal{C}_{EH} \subset\mathcal{F}_{EH}$ be the submanifold of canonical constraints for EH theory,\footnote{Formally, this is recovered as $\mathcal{C}_{EH}=\mathrm{Spec}(H^0(C^\infty(\F_{EH})))$.} denote by $\iota_{\mathrm{res}}\colon \mathcal{C}_{PC}\to \mathcal{C}_{\mathrm{Lor}}$ and by $\underline{\mathcal{C}_{\mathrm{Lor}}}$ the coisotropic reduction of the zero locus of the function 
    \[
    \mathsf{L}^c=\trintl{\Sigma} cd_\omega e, \qquad \mathcal{C}_{\mathrm{Lor}}\doteq \mathrm{Zero}(\mathsf{L}^c)\subset \mathcal{F}_{PC}\simeq \mathcal{U}\times \Omega^1(\Sigma,\wedge^2\mathcal{V}),
    \]
    with $\pi_{\mathrm{Lor}}\colon \mathcal{C}_{\mathrm{Lor}}\to\underline{\mathcal{C}_{\mathrm{Lor}}}$ the associated quotient map. Then, there is a symplectomorphism
    \[
    \sigma\colon \underline{\mathcal{C}_{\mathrm{Lor}}} \to \mathcal{F}_{EH}
    \]
    and, denoting $\mathcal{C}_{\mathrm{res}}\doteq \pi_{\mathrm{Lor}}(\iota_{\mathrm{res}}(\mathcal{C}_{PC}))\subset \underline{\mathcal{C}_{\mathrm{Lor}}}$, we have $\mathcal{C}_{EH} = \sigma(\mathcal{C}_{\mathrm{res}})$.
\end{theorem}

The relationship between the EH and PC theories of gravity is summarised by the following diagram:

\[
\tikz[
overlay]{
    \draw[dotted] (0, 0.5) --node[above] {\color{gray}\scriptsize{EH theory}} (1,0.5) -- (3.2, -1.4)-- (1,-3.2)--(0,-3.2)--(0,0.5);
    \draw[dotted] (3.7, 0.5) -- (4.7,0.5) -- (6.9, -1.4)-- (4.7,-3.2)--(3.7,-3.2)--(3.7,0.5);
    \draw[dotted] (3.6, 3.4) --node[above] {\color{gray}\scriptsize{PC theory}} (5,3.4) -- (8.5, 0.5)-- (8.5,-3.3)--(3.6,-3.3)--(3.6, 3.4);
}
\begin{tikzcd}[row sep=2.5em]
    && \mathcal{F}_{PC} & &\\
&&  &\mathcal{C}_{\mathrm{Lor}} \arrow[ul, hook] \arrow[ld, "\pi_{Lor}"] &\\
\mathcal{F}_{EH}  & & \arrow[ll, "\sigma", "\sim"'] \underline{\mathcal{C}_{\mathrm{Lor}}}& & \arrow[ul, "\iota_{\mathrm{res}}", hook]\arrow[dl]\mathcal{C}_{PC} \arrow[dd, "\pi_{PC}"]\\
& \arrow[ul, hook] \arrow[dl, "\pi_{EH}"] \mathcal{C}_{EH}  && \arrow[ll, "\sigma", "\sim"'] \arrow[ul, hook]\arrow[dl, "\pi_{\mathrm{res}}"]\mathcal{C}_{\mathrm{res}} &\\
\underline{\mathcal{C}_{EH}}  & & \arrow[ll, "\underline{\sigma}", "\sim"'] \underline{\mathcal{C}_{\mathrm{res}}} \arrow[rr, "\sim"] & & \underline{\mathcal{C}_{PC}}
\end{tikzcd}
\]

\section{\texorpdfstring{$BF$}{BF} to EH reduction: Classical picture}

Since nondegenerate $BF$ theory is strongly equivalent to PC theory, which reduces to EH theory, we are going to look at the direct reduction from $BF$ theory to EH theory in this section.
We will first look at the BFV coisotropic reduction from the BFV data associated to $BF$ theory, implementing the Loentz constraint within the BFV framework. Then we will move to the resolution of this coisotropic reduction to obtain the ``double BFV'' construction outlined in Section \ref{s:doubleBFV}.

\subsection{EH theory as a coisotropic reduction of \texorpdfstring{$BF$}{BF} theory}\label{s:relation_Classical_theories}

It is possible to find a coisotropic submanifold in the BFV data associated to $BF$ theory, whose reduction is (symplectomorphic to) the BFV data for Einstein--Hilbert theory. We will perform the calculations explicitly, presenting the reduction in a local chart (in the space of BFV fields).

Consider the submanifold 
\[
\C_{\mathrm{Gauss}}\subset \F_{BF}=\Omega^1_{nd}(\Sigma, \mathcal{V}) \oplus \mathcal{A}_P \oplus T^* \left(\Omega^0[1]( \Sigma, \wedge^2\mathcal{V})\oplus \Omega^0[1]( \Sigma, \wedge^1\mathcal{V})\right),
\]
defined as the zero locus of the following constraints:
\begin{subequations}
\begin{gather}
    \mathbf{\Phi}=(QA^\dag)\vert_{A^\dag=0} =  d_{A} B + [\tau, B^\dagger]=0 \label{e:ConstraintBF1}\\
    \mathbf{\Phi}^\dag= A^\dagger=0. \label{e:ConstraintBF2}
\end{gather}
\end{subequations}

\begin{remark}
    The constraint \eqref{e:ConstraintBF1} is the BFV analogue of the torsionless constraint $d_{A} B=0$ that is used to impose that $A$ is the Levi--Civita connection associated to $B$. Note in particular that here we no longer require the torsion to vanish but we require it to be proportional to higher ghost degree terms. The submanifold $\C_{\mathrm{Gauss}}$ is the analogue of $\CCo$ of the general theory for the nested coisotropic embedding $\mathcal{C}_{BF}\hookrightarrow \mathcal{C}_{\mathrm{Gauss}}\hookrightarrow \mathcal{F}_{BF}$.
\end{remark}

    The fact that that the constraints \eqref{e:ConstraintBF1} and \eqref{e:ConstraintBF2} are coisotropic is a consequence of Theorems \ref{thm:BFVcoisotropic} and \ref{thm:doubleBFV}. In order to see the mechanism at work explicitly we can promote them to integral functions using some Lagrange multipliers and compute the brackets between themselves with respect to $\varpi_{BF}$.
    Namely, \eqref{e:ConstraintBF1} is an equation in $\Omega^{2,1}(\Sigma)$ and hence we need a Lagrange multiplier $\rho \in \Omega^{0,2}[1](\Sigma)$, while \eqref{e:ConstraintBF2} is an equation in $\Omega^{2,1}[-1](\Sigma)$ and we need a Lagrange multiplier $\mu \in \Omega^{0,2}[2](\Sigma)$.\footnote{Note that we use odd Lagrange multipliers in order to be consistent with the BFV description presented in Section \ref{s:BFVonBFV_Classical}.} Then we have the following lemma:
    \begin{lemma} \label{l:brackets_constraints}
    Consider the functions 
    \begin{align*}
        \mathsf{J}^{\rho} = \trintl{\Sigma}\rho (d_A B + [\tau, B^{\dag}]) \quad\text{and}\quad  \mathsf{M}^{\mu} = \trintl{\Sigma}\mu A^{\dag},
    \end{align*} then we have:
    \begin{align*}
        \{\mathsf{J}^{\rho},\mathsf{J}^{\rho}\}&= \frac{1}{2}\mathsf{J}^{[\rho,\rho]} & 
        \{\mathsf{J}^{\rho},\mathsf{M}^{\mu}\}&= 0 & 
        \{\mathsf{M}^{\mu},\mathsf{M}^{\mu}\}&= 0.
    \end{align*}
\end{lemma}
\begin{proof}
    This is a direct consequence of Theorem \ref{thm:BFVcoisotropic}. It is anyway worthwhile to compute it explicitly, for further reference.
    The components of the Hamiltonian vector fields of $\mathsf{J}^{\rho}$ and $\mathsf{M}^{\mu}$ are
    \begin{align*}
        {X}_{\mathsf{J}}(B)&= [\rho, B] & {X}_{\mathsf{J}}(A)&= d_A \rho & {X}_{\mathsf{J}}(\tau)&= -[\rho, \tau]\\
        {X}_{\mathsf{J}}({B^\dag})&= -[\rho, B^{\dag}] & {X}_{\mathsf{J}}(\chi)&= 0& {X}_{\mathsf{J}}({A^\dag})&= 0\\
        {X}_{\mathsf{M}}(B)&= 0 & {X}_{\mathsf{M}}(A)&= 0 & {X}_{\mathsf{M}}(\tau)&= 0\\
        {X}_{\mathsf{M}}({B^\dag})&= 0 & {X}_{\mathsf{M}}(\chi)&= \mu & {X}_{\mathsf{M}}({A^\dag})&= 0.
    \end{align*}
    Hence we get:
    \begin{align*}
        \{\mathsf{J}^{\rho},\mathsf{J}^{\rho}\}= \iota_{{X}_{\mathsf{J}}}\delta \mathsf{J}^{\rho} &= \trintl{\Sigma} -\rho [d_A \rho, B] + \rho d_A([\rho, B]) + \rho [[\rho, \tau], B^\dag] + \rho [ \tau, [\rho,B^\dag] ]\\
        &= \trintl{\Sigma} \frac{1}{2}[\rho,\rho] (d_A B + [\tau, B^\dag])= \frac{1}{2}\mathsf{J}^{[\rho,\rho]}\\
        \{\mathsf{J}^{\rho},\mathsf{M}^{\mu}\} = \iota_{{X}_{\mathsf{J}}}\delta \mathsf{M}^{\mu} &= 0\\
        \{\mathsf{M}^{\mu},\mathsf{M}^{\mu}\} = \iota_{{X}_{\mathsf{M}}}\delta \mathsf{M}^{\mu}&=0.
    \end{align*}
\end{proof}

An immediate corollary is:
\begin{corollary}\label{c:constraints_coisotropic}
        The submanifold $\C_{\mathrm{Gauss}}\hookrightarrow\F_{BF}$ defined by \eqref{e:ConstraintBF1} and \eqref{e:ConstraintBF2} is coisotropic.
\end{corollary}

 Let us call $\underline{\C_{\mathrm{Gauss}}}$ the coisotropic reduction of $\F_{BF}$ with respect to \eqref{e:ConstraintBF1} and \eqref{e:ConstraintBF2} and let us denote by $\pi_{\C_{\mathrm{Gauss}}}$ the projection $\pi_{\C_{\mathrm{Gauss}}}: \C_{\mathrm{Gauss}}\rightarrow \underline{\C_{\mathrm{Gauss}}}$. We then have the following theorem:

\begin{theorem}\label{t:BF_coisotropic_EH}
    The symplectic space $\underline{\C_{\mathrm{Gauss}}}$ comes equipped with a solution of the master equation $\underline{S}_{BF}\in C^\infty(\underline{\C_{\mathrm{Gauss}}})$, i.e.\ such that $(\underline{S}_{BF},\underline{S}_{BF})=0$ and $\pi^*_{\C_{\mathrm{Gauss}}}\underline{S}_{BF}=\iota^*_{\C_{\mathrm{Gauss}}}S_{BF}$. Then, $\BFV_{\underline{\C_{\mathrm{Gauss}}}} \equiv\underline{\mathbb{BFV}(\mathcal{F}_{BF},\mathcal{C}_{\mathrm{Gauss}},\mathcal{C}_{BF})}= (\underline{\C_{\mathrm{Gauss}}},\varpi_{\underline{\C_{\mathrm{Gauss}}}},\underline{\mathcal{S}}_{BF})$
    forms a BFV triple, and it is strongly BFV-equivalent to $\BFV_{EH}$, via the symplectomorphism
    \[
    \sigma_{\BFV}\colon \underline{\C_{\mathrm{Gauss}}}\to \F_{EH}.
    \]
\end{theorem}

\begin{proof}
We divide the proof in multiple steps. In the first step, we solve the constraints that define $\C_{\mathrm{Gauss}}$ and perform the corresponding coisotropic reduction, obtaining the BFV theory $\BFV_{\underline{\C_{\mathrm{Gauss}}}}$. In the second step, we then show an explicit symplectomorphism between this and the EH BFV theories.

The first step has already been considered in \cite{CS2019,thi2007}. Here we consider a slightly different perspective, using a coisotropic reduction and extend it to the BFV setting by including the antifields.

As a first step we fix a basis $\{v_0,v_1,v_2\}$ of $\mathcal{V}_{\Sigma}$ such that in components $B_{\mu}^i$ (both $i=1,2$ and $\mu=1,2$) is an invertible diad. Note that this is possible since, in components, $B$ is represented by a $3\times 2$ matrix $B_\mu^{I}$, with $I=0,1,2$ and $i$ a sub index of $I$, with at least an invertible $2\times 2$ minor, by the non-degeneracy property of $B$. Then the constraint \eqref{e:ConstraintBF1} can be rewritten in components as the following two equations:
    \begin{subequations}
        \begin{align}
            \partial_{[\mu} B_{\nu]}^i + A_{[\mu}^{ij} B_{\nu]}^k \eta_{jk} + A_{[\mu}^{i0} B_{\nu]}^0 \eta_{00} +\tau^k B^{\dag {ij}}_{\mu\nu}\eta_{jk}+\tau^0 B^{\dag {i0}}_{\mu\nu}\eta_{00}=0 \label{e:ConstraintBF_LeviCivita}\\
            \partial_{[\mu} B_{\nu]}^0 + A_{[\mu}^{0i} B_{\nu]}^j \eta_{ij} +\tau^i B^{\dag {0j}}_{\mu\nu}\eta_{ij}=0. \label{e:ConstraintBF_antisym}
        \end{align} 
    \end{subequations}

    We first consider \eqref{e:ConstraintBF_LeviCivita}, solve and reduce by it. It can be solved for $A_{\mu}^{ij}$ as
    \begin{align*}
        A_{\mu}^{ij} = \Gamma_{\mu}^{ij} + 2 X_{\mu\nu}^l(B^{-1})^\nu_k \epsilon^{ij}\epsilon_{l}^{\;k}
    \end{align*}
    where $\Gamma_{\mu}^{ij}$ is the spin connection associated to the diad $B_\mu^i$ 
    \begin{align*}
        \Gamma_{\mu }^{ij}= \eta^{k[i}(B^{-1})^{\nu}_k\left(\partial_{\mu}B_{\nu}^{j]}-\partial_{\nu}B_{\mu}^{j]}+\eta^{j]l}(B^{-1})^{\sigma}_l B_{\mu }^{m} \partial_{\sigma}B_{\nu}^n\eta_{mn} \right)
    \end{align*}
    and
    \begin{align*}
        X_{\mu\nu}^i = A_{[\mu}^{i0} B_{\nu]}^0 \eta_{00} +\tau^k B^{\dag {ij}}_{\mu\nu}\eta_{jk}+\tau^0 B^{\dag {i0}}_{\mu\nu}\eta_{00}.
    \end{align*}
    
    The corresponding symplectic reduction is carried on in Appendix \ref{a:coisotropic_BFtoEH} and the result is collected in Remark \ref{r:intermediate_theory}. Dropping the double apices, after the reduction we get a BFV theory with symplectic form (the suffix ``int'' stands for intermediate)
        \begin{align*}
            \varpi_{\mathrm{int}} &= \int_{\Sigma} \left(2 \delta A_{\mu}^{0i} \delta B_{\nu}^{j}  +2 \delta \tau^i \delta B^{\dag 0j} + \delta \tau^0 \delta B^{\dag ij}  \right) \epsilon_{ij}dx^{\mu}dx^{\nu}
        \end{align*}
    and  action
        \begin{align*}
            S_{\mathrm{int}} &= \trintl{\Sigma} \tau F_{A} \\
        \end{align*}
         where we still have to impose the constraint \eqref{e:ConstraintBF_antisym}, which after the reduction of the submanifold defined by \eqref{e:ConstraintBF_LeviCivita} reads
        \begin{align*}
            \left(A_{[\mu}^{0i}B_{\nu]}^j  \eta_{ij} +\tau^i B^{\dag {0j}}_{\mu\nu}\eta_{ij}\right)\epsilon^{\mu\nu}=0. 
        \end{align*}
    The coisotropic reduction with respect to it is carried out in Appendix \ref{a:coisotropic_BFtoEH}. We can view this last constraint as fixing the antisymmetric part of $A_{\mu}^{0i} B_{\nu}^j \eta_{ij}$.
    Hence the only components of $A$ that are left free are the ones contained in the symmetric combination 
    \begin{align}\label{e:def_K_EH}
        K_{\mu\nu}:= B_{(\mu}^a A_{\nu)}^{0b} \eta_{ab},
    \end{align}
    which will play the role of the momentum in the EH formulation. 
    
    We will now show the explicit expression of the map between the spaces of fields of the intermediate BFV theory, described above, and of the BFV theory associated to Einstein--Hilbert gravity. Consider $\sigma_{\BFV}: \mathbb{F}_{\mathrm{int}} \to \mathbb{F}_{EH}$ defined (in a local chart) by
    \begin{equation}\label{e:MapBFtoEH}
        \sigma_{\BFV}\colon\begin{cases}
        g_{\mu\nu}= B_{\mu}^i B_{\nu}^j \eta_{ij}\\
        \Pi^{\mu\nu}= \sqrt{\mathsf{g}} (K_{\mu\nu}- g_{\mu\nu}g^{\rho\sigma}K_{\rho\sigma})-  \frac{1}{2} B_{\mu}^i B_{\nu}^j \left(\epsilon_{ik}\eta_{jl}+\epsilon_{jk}\eta_{il}\right)\tau^l  {B^{\dag}}^{0k}_{\rho\sigma}\epsilon^{\rho\sigma}\\
        \txi{n} =   \tau^0\\
        \txi{\mu} =  (B^{-1})^{\mu}_i\tau^i\\
        \phin = {B^{\dag}}^{ij}_{\rho\sigma}\epsilon^{\rho\sigma}\epsilon_{ij}\\
        \varphi_{\mu} = 2 {B^{\dag}}^{0i}_{\rho\sigma}\epsilon^{\rho\sigma}B_{\mu}^j\epsilon_{ij}.
        \end{cases}
    \end{equation}

    In order to conclude, one has to show the following three properties:
    \begin{enumerate}
        \item The map $\sigma_{\BFV}$ is invariant under the characteristic distribution of the residual constraint \eqref{e:ConstraintBF_antisym}, and therefore it descends to a map $\underline{\sigma_{\BFV}}\colon\underline{\C_{\mathrm{Gauss}}}\to \mathbb{F}_{EH}$,
        \item The map $\underline{\sigma_{\BFV}}$ is a symplectomorphism, i.e.
        \begin{align*}
            \underline{\sigma_{\BFV}}^{*}(\varpi_{EH}) = \varpi_{\mathrm{int}},
        \end{align*}
        \item the transformation preserves the actions, i.e.
        \begin{align*}
            \underline{\sigma_{\BFV}}^{*}(S_{EH}) = S_{\mathrm{int}}.
        \end{align*}
    \end{enumerate}
    We prove these three properties in Appendices \ref{a:independence}, \ref{a:symplectomorphism} and \ref{a:actions} respectively.
\end{proof}

Let us now draw the connection with Palatini--Cartan theory. Using the symplectomorphism between $\F_{PC}$ and $\F_{BF}$ presented in \eqref{e:PCtoBF}, it is possible to express the constraints defining the graded coisotropic submanifold $\C_{\mathrm{Gauss}}\subset{\F_{BF}}$ as defining a graded coisotropic submanifold within $\F_{PC}$ via:
\begin{align}
    & \mathbf{\Phi}=d_{\omega} e + [\lambda \epsilon_n, e^\dagger]=0 \label{e:ConstraintPC1}\\
    & \mathbf{\Phi}^\dag=\omega^\dagger=0. \label{e:ConstraintPC2}
\end{align}
We denote by $\C_{\mathrm{Lor}}$ the associated vanishing locus. Since the map in Equation \eqref{e:PCtoBF} is a symplectomorphism, we immediately also get the following ``corollary'' of Corollary \ref{c:constraints_coisotropic}

\begin{corollary}
    The submanifold of $\F_{PC}$ defined by \eqref{e:ConstraintPC1} and \eqref{e:ConstraintPC2} is coisotropic.
\end{corollary}

Similarly to what we did in the context of $BF$ theory, let us denote by $\pi_{\C_{\mathrm{Lor}}}\colon \C_{\mathrm{Lor}}\to \underline{\C_{\mathrm{Lor}}}$ the coisotropic reduction of $\C_{\mathrm{Lor}}$. Then, by precomposing with the map $\psi$ in Equation \eqref{e:PCtoBF} we obtain the following corollary to Theorem \ref{t:BF_coisotropic_EH}:
\begin{corollary}\label{c:PC_coisotropic_EH}
    The symplectic space $\underline{\C_{\mathrm{Lor}}}$ comes equipped with a solution of the master equation $\underline{S}_{PC}\in C^\infty(\underline{\C_{\mathrm{Lor}}})$, i.e.\ such that $(\underline{S}_{PC},\underline{S}_{PC})=0$ and $\pi^*_{\C_{\mathrm{Lor}}}\underline{S}_{PC}=\iota^*_{\C_{\mathrm{Lor}}}S_{PC}$. Then, 
    \[\BFV_{{\C}_{PC}}\equiv\underline{\mathbb{BFV}(\mathcal{F}_{PC},\mathcal{C}_{\mathrm{Lor}},\mathcal{C}_{PC})}=(\underline{\C_{\mathrm{Lor}}},\varpi_{\underline{\C_{\mathrm{Lor}}}},\underline{\mathcal{S}}_{PC})
    \]
    forms a BFV triple, and it is strongly BFV-equivalent to $\BFV_{EH}$.
\end{corollary}

    For future reference we only recall the explicit expression of the map $\underline{\sigma^{PC}_{\BFV}}\colon\underline{\C_{\mathrm{Lor}}}\to \mathbb{F}_{EH}$
    between the (partial) coisotropic reduction in the Palatini--Cartan space of fields and the space of fields for EH theory:
    \begin{equation}\label{e:MapPCtoEH}
       \sigma^{PC}_{\BFV} : \begin{cases}
        g_{\mu\nu}= e_{\mu}^a e_{\nu}^b \eta_{ab}\\
       \Pi^{\mu\nu}= \sqrt{\mathsf{g}} (K_{\mu\nu}- g_{\mu\nu}g^{\rho\sigma}K_{\rho\sigma})-  \frac{1}{2} e_{\mu}^i e_{\nu}^j \left(\epsilon_{ik}\eta_{jl}+\epsilon_{jk}\eta_{il}\right) e_{\theta}^l \xi^{\theta} {e^{\dag}}^{0k}_{\rho\sigma}\epsilon^{\rho\sigma}\\
        \txi{n} =  \epsilon_n^0 \xi^n\\
        \txi{\mu} = \xi^\mu +  \epsilon_n^c \xi^n(e^{-1})^{\mu}_c\\
        \phin = e^{\dag ab}\epsilon_{ab}\\
        \varphi_{\mu} =  2 e^{\dag 0a}e_{\mu}^b\epsilon_{ab}.
        \end{cases}
    \end{equation}

\subsection{Double BFV description}\label{s:BFVonBFV_Classical}
The coisotropic reduction described in Section \ref{s:relation_Classical_theories} can also be described cohomologically using the BFV formalism as presented in Section \ref{s:doubleBFV}. This will have the advantage to have a clearer quantisation scheme, as described in Section \ref{s:quantisation_BFVonBFV}.

Let us focus on the relation between $BF$ theory and EH theory. We want to describe the coisotropic submanifold $\C_{\mathrm{Gauss}}\subset\F_{BF}$ defined by \eqref{e:ConstraintBF1} and \eqref{e:ConstraintBF2}. Following the procedure outlined in Section \ref{s:coiso_by_BFV}, we first enlarge the space of fields to accommodate the Lagrange multipliers of the constraints and their antifields. As in the classical picture we introduce two Lagrange multipliers  $\rho \in \Omega^{0,2}[1](\Sigma)$ and $\mu \in \Omega^{0,2}[2](\Sigma)$, so that the new space of fields is the symplectic vector bundle
\[
\left(\FF_{BF} = \F_{BF} \times T^*\left(\Omega^{0,2}[1](\Sigma)\times \Omega^{0,2}[2](\Sigma)\right) \to \F_{BF}, \bbomega_{BF}\right).
\]

The brackets of the BFV constraints have already been computed in Lemma \ref{l:brackets_constraints}, hence, we build a BFV action $\mathbb{S}_{BF}\in C^\infty(\FF_{BF})$ as 
\begin{align} \label{e:BFV_action_R}
    \mathbb{S}_{BF} = \trintl{\Sigma} \rho (d_A B + [\tau, B^{\dag}])+ \mu A^{\dag} + \frac{1}{2}[\rho,\rho] \rho^\dag +r.
\end{align}
The following Lemma shows that $\mathbb{S}_{BF}$ satisfies the CME with $r=0$:
\begin{lemma}\label{l:CME_R} 
If $r=0$ we have
\begin{align*}
    \{\mathbb{S}_{BF},\mathbb{S}_{BF}\}_{\bbomega_{BF}}=0.
\end{align*}
\end{lemma}
\begin{proof}
    Let $r=0$. The Hamiltonian vector field of $\mathbb{S}_{BF}$ with respect to $\bbomega_{BF}$ is renamed $\mathbb{Q}_{BF}\equiv {X}_{\mathbb{S}_{BF}}$, and it reads:
    \begin{align*}
        \mathbb{Q}_{BF}(B)&= [\rho, B] & \mathbb{Q}_{BF}(A)&= d_A \rho \\
        \mathbb{Q}_{BF}(\tau)&= -[\rho, \tau] & \mathbb{Q}_{BF}({B^\dag})&= -[\rho, B^{\dag}] \\ \mathbb{Q}_{BF}(\chi)&= \mu & \mathbb{Q}_{BF}({A^\dag})&= 0 \\
        \mathbb{Q}_{BF}(\rho)&= \frac{1}{2}[\rho,\rho] & \mathbb{Q}_{BF}({\rho^\dag})&= [\rho, \rho^\dag] +  d_A B + [\tau, B^{\dag}] \\
        \mathbb{Q}_{BF}(\mu)&= 0  &
        \mathbb{Q}_{BF}({\mu^\dag})&= A^\dag.
    \end{align*}
    Since we already computed the brackets between the constraints, we get
        \begin{align*}
        \{\mathbb{S}_{BF},\mathbb{S}_{BF}\}_{\bbomega_{BF}} &= \{\mathsf{J}^{\rho},\mathsf{J}^{\rho}\}_{{\varpi_{BF}}} + \{\mathsf{J}^{\rho},\mathsf{M}^{\mu}\}_{{\varpi_{BF}}} + \{\mathsf{M}^{\mu},\mathsf{M}^{\mu}\}_{{\varpi_{BF}}} + \iota_{\mathbb{Q}_{BF}}\iota_{\mathbb{Q}_{BF}}\trintl{\Sigma} \delta \mu \delta \mu^\dag +  \delta \rho \delta \rho^\dag \\
        &=  \trintl{\Sigma} \frac{1}{2}[\rho,\rho] (d_A B + [\tau, B^\dag]) -\frac{1}{2}[\rho,\rho] ([\rho,\rho^\dag] + d_A B + [\tau, B^\dag])\\
        &=  -\trintl{\Sigma}\frac{1}{2}[[\rho,\rho],\rho]\rho^\dag=0,
    \end{align*}
    where in the last passage we used the graded Jacobi identity.
\end{proof}

Hence we have constructed a classical BFV theory that encodes the (graded) coisotropic submanifold defined by \eqref{e:ConstraintBF1} and \eqref{e:ConstraintBF2}. Next, we want to construct a function $\check{S}_{BF} \in  C^{\infty}(\FF_{BF})$, such that its restriction to the zero section $\F_{BF}\hookrightarrow \FF_{BF}$ coincides with $S_{BF}$ and such that $\{\check{S}_{BF} + \mathbb{S}_{BF},\check{S}_{BF} + \mathbb{S}_{BF}\}_{\bbomega_{BF}}=0$. 

\begin{lemma} 
The function $\check{S}_{BF} \in C^\infty(\FF_{BF})$ defined as
    \begin{align}
        \check{S}_{BF} := S_{BF} + \int_\Sigma [ \chi, \mu ]\mu^\dag + \mu \rho^\dag + [\chi, \rho] \rho^\dag. 
    \end{align}
satisfies
    \begin{align}\label{e:brackets_SR}
    \{\check{S}_{BF}+ \mathbb{S}_{BF},\check{S}_{BF}+\mathbb{S}_{BF}\}_{\bbomega_{BF}}=0
    \end{align}
and its restriction to the zero section is $\check{S}_{BF}\vert_{\F_{BF}} = S_{BF}$.
\end{lemma}
\begin{proof}
    It is immediate to check that the restriction to the zero section is $S_{BF}$.
    
    Since $\{\mathbb{S}_{BF},\mathbb{S}_{BF}\}_{\bbomega_{BF}}=0$, Equation \eqref{e:brackets_SR} is satisfied iff 
    \[
    \{\check{S}_{BF},\mathbb{S}_{BF}\}_{\bbomega_{BF}}=0 \quad \text{and}  \quad \{\check{S}_{BF},\check{S}_{BF}\}_{\bbomega_{BF}}=0.
    \] 
    For the former condition, using the expression of the Hamiltonian vector field of $\mathbb{S}_{BF}$ found in the proof of Lemma \ref{l:CME_R}, we get:
    \begin{align*}
        \{\check{S}_{BF},\mathbb{S}_{BF}\}_{\bbomega_{BF}}= \iota_{\mathbb{Q}_{BF}} \delta \check{S}_{BF} &= \trintl{\Sigma} [\rho, \tau] F_A + \tau [F_A, \rho] + [\rho, \tau][\chi, B^\dag]+ \tau [\mu, B^\dag]+\tau[\chi, [\rho, B^\dag]] \\
        & \phantom{\trintl{\Sigma}} + [\rho, B] d_A \chi + B[d_A \rho, \chi] + B d_A \mu + A^\dag[\chi, \mu] + [\mu, \rho] \rho^\dag\\
        & \phantom{\trintl{\Sigma}} + \frac{1}{2}[\chi, [\rho, \rho]]\rho^\dag + [\chi,\rho] [\rho, \rho^\dag] +  [\chi,\rho]d_A B + [\chi,\rho][\tau, B^{\dag}]\\
        & \phantom{\trintl{\Sigma}}+ \mu[\rho, \rho^\dag] +  \mu d_A B + \mu[\tau, B^{\dag}] + [\mu, \mu] \mu^\dag + [\mu, \chi]A^\dag\\
        &= \trintl{\Sigma} d_A ([\rho, B]\chi + \mu B) = 0
    \end{align*}
    where we used that $\partial\Sigma=\emptyset$, as well as $[\mu,\mu]=0$ for degree reasons and the Jacobi identity as follows:
    \begin{align*}
        [\rho, \tau][\chi, B^\dag]+\tau[\chi, [\rho, B^\dag]]+ [\chi,\rho][\tau, B^{\dag}] = \tau [\rho,[\chi, B^\dag]] +\tau[\chi, [\rho, B^\dag]] + \tau[B^{\dag},[\chi,\rho]]=0,\\
        \frac{1}{2}[\chi, [\rho, \rho]]\rho^\dag + [\chi,\rho] [\rho, \rho^\dag] = \frac{1}{2}[\chi, [\rho, \rho]]\rho^\dag + [[\chi,\rho] ,\rho] \rho^\dag=0.
    \end{align*}
    To show that $\{\check{S}_{BF},\check{S}_{BF}\}_{\bbomega_{BF}}=0$ we have to compute the Hamiltonian vector field $\check{Q}_{BF}$ of $\check{S}_{BF}$. By denoting $\check{Q}_{BF}= Q_{BF}+ \mathrm{q}$ we get
        \begin{align*}
        \mathrm{q}_B&= 0 & \mathrm{q}_{A^\dag}&= [\rho,\rho^\dag] + [\mu, \mu^\dag] &
        \mathrm{q}_\tau&= 0 & \mathrm{q}_{\rho^\dag}&= [\chi, \rho^\dag]& \mathrm{q}_\rho&= [\chi,\rho] +\mu \\ \mathrm{q}_A&= 0 &
        \mathrm{q}_\chi&= \mu &  \mathrm{q}_{B^\dag}&= 0 &
        \mathrm{q}_\mu&= [\chi,\mu]  &
        \mathrm{q}_{\mu^\dag}&= [\chi,\mu^\dag] + \rho^\dag.
    \end{align*}
    Using the definition of the Poisson brackets, we can express $\{\check{S}_{BF},\check{S}_{BF}\}_{\bbomega_{BF}}$ as follows:
    \begin{align*}
        \{\check{S}_{BF},\check{S}_{BF}\}_{\bbomega_{BF}} &= \frac{1}{2}\iota_{Q_{BF}}\iota_{Q_{BF}}\bbomega_{BF} + \iota_{Q_{BF}}\iota_{\mathrm{q}} \bbomega_{BF}+ \frac{1}{2}\iota_{\mathrm{q}}\iota_{\mathrm{q}}\bbomega_{BF}\\
        &= \{S_{BF},S_{BF}\}_{\varpi_{BF}} + \frac{1}{2}[\chi,\chi]([\rho,\rho^\dag] + [\mu, \mu^\dag]) + [\chi,\mu]([\chi,\mu^\dag] + \rho^\dag)\\
        &\phantom{=}+ ([\chi,\rho] +\mu) [\chi, \rho^\dag]\\
        &= \frac{1}{2}[[\chi,\chi]\rho]\rho^\dag]+ \frac{1}{2}[[\chi,\chi] ,\mu] \mu^\dag+ [[\chi,\mu],\chi]\mu^\dag+[[\chi,\rho],\chi] \rho^\dag=0
    \end{align*}
    where we used that $S_{BF}$ satisfies the CME for $\varpi_{BF}$ and the Jacobi identity.
\end{proof}

\begin{remark}\label{r:S_m}
    Since we are only interested in the equivalence class of $\check{S}_{BF}$ it is possible to add any term in the image of $\mathbb{Q}_{BF}$ to a representative $\check{S}_{BF}$. In particular it is a straightforward computation to show that 
    \begin{align}\label{e:S_rescheck}
        \check{S}_{BF} + \mathbb{Q}_{BF}\left( \chi\rho^\dag + \frac{1}{2}[\chi,\chi]\mu^\dag\right) = \trintl{\Sigma} \tau F_A =: \check{S}_{\mathrm{res}},
    \end{align}
and hence $[\check{S}_{BF}]_{\mathbb{Q}_{BF}}= [\check{S}_{\mathrm{res}}]_{\mathbb{Q}_{BF}} = S_{\mathrm{res}}\equiv \underline{S_{\C_{\mathrm{Gauss}}}}$ as a consequence of Theorem \ref{thm:doubleBFVreduced}. (C.f. with Appendix \ref{a:coisotropic_BFtoEH}, Equation \eqref{e:action_reduced}. Note that in this example the structure function $m=0$, so---with reference to the proof of Theorem \ref{thm:doubleBFVreduced}---we have $\wt{S}_{\mathrm{res}} =\check{S}_{\mathrm{res}}$.)
\end{remark}

\section{\texorpdfstring{$BF$}{BF} to EH reduction: Quantum picture} \label{s:quantumBF_doubleBFV}
The goal of this section is to approach BFV quantisation of three dimensional gravity, using the relation to BF theory and its BFV quantisation to its full extent. Our starting point is going to be an adaptation of the quantization of BF theory developed by Cattaneo, Mnev and Reshetikhin in \cite{CMR2}, and we are going to transfer this quantization to Einstein--Hilbert gravity, using the relation between the respective classical theories.

We follow the procedure outlined in Sections \ref{s:doubleBFV} and \ref{s:quantisation_BFVonBFV} applied to the cohomological description of the classical BFV relation between $BF$ theory and Einstein--Hilbert theory  given in Section \ref{s:BFVonBFV_Classical}. The first step in this program is finding a quantisation of the double BFV data $\mathbb{BFV}_{BF}=\left(\FF_{BF},\bbomega_{BF},\mathbb{Q}_{BF},\mathbb{S}_{BF}\right)$. In order to do so we use the geometric quantisation approach and choose a polarisation. Recall that the symplectic form has Darboux form 
\[
\bbomega_{BF}
    = \trintl{\Sigma} \delta B \delta A + \delta \chi \delta A^\dag + \delta \tau \delta B^\dag + \delta\mu^\dag \delta\mu + \delta \rho^\dag\delta\rho = \delta \bbalpha_{BF},
\]

In particular we choose here the Lagrangian foliation generated by the distribution 
\[
\mathbb{P}\colon \left\{\frac{\delta}{\delta \mu^\dag},\frac{\delta}{\delta \rho^\dag}, \frac{\delta}{\delta B^\dag}, \frac{\delta}{\delta A^\dag},
\frac{\delta}{\delta B}\right\}
\]
whose space of leaves is\footnote{This corresponds to the vertical polarisation in the cotangent bundle $\FF_{BF} \to \mathbb{B}_{BF}$, up to the choice of a reference connection $A_0$.} 
\[
\mathbb{B}_{BF} \doteq \FF_{BF}/\mathbb{P} \simeq  \underbrace{\mathcal{A}(\Sigma)}_A \oplus  \underbrace{\Omega^0[1]( \Sigma, \wedge^2\mathcal{V})}_{\chi}\oplus \underbrace{\Omega^0[1]( \Sigma, \wedge^1\mathcal{V})}_{\tau}\oplus\underbrace{\Omega^0[1]( \Sigma, \wedge^2\mathcal{V})}_{\rho}\oplus \underbrace{\Omega^0[2]( \Sigma, \wedge^2\mathcal{V})}_{\mu} 
\]

Choosing a trivial line bundle with connection given by the one form 
\[
\bbalpha_{BF} = \trintl{\Sigma} B\delta A + A^\dag\delta \chi +  B^\dag\delta \tau  + \mu^\dag \delta\mu + \rho^\dag\delta\rho, 
\]
and using this choice of polarisation, we construct the graded vector space simply as
\[
\mathbb{V}_{BF}\doteq C^\infty(\mathbb{B}_{BF}),
\]
and the quantisation map (for polarisation preserving, Hamiltonian functions) is the Lie algebra morphism (with $\nabla$ associated to the connection one form $\bbalpha_{BF}$)
\[
\mathfrak{q}\colon \mathcal{A}_{\mathbb{P}}\subset(C^\infty_{\mathbb{P}}(\FF_{BF}),\{\cdot,\cdot\}_{\bbomega_{BF}}) \to (\mathrm{End}(\mathbb{V}_{BF}), [\cdot,\cdot]), \quad f\mapsto -i\hbar \nabla_{X_f}  + f = -i\hbar X_f - \iota_{X_f}\bbalpha_{BF} f + f.
\]

\begin{remark}
    Note that this (pre-)quantization data can be inherited by the full densitised cotangent bundle $T^\vee\mathcal{A}(\Sigma) \simeq \Omega^1(\Sigma,\mathcal{V})$, by restricting to open subsets of the fibre given by nondegenerate $B$-fields.
\end{remark}

\begin{lemma}\label{l:Omega_R2}
The quantum double BFV operator $\mathbb{\Omega}_{BF}: \mathbb{V}_{BF} \rightarrow \mathbb{V}_{BF}$ associated to the geometric quantisation of (Lemma \ref{l:CME_R})
\[
\mathbb{S}_{BF} = \trintl{\Sigma} \rho (d_A B + [\tau, B^{\dag}])+ \mu A^{\dag} + \frac{1}{2}[\rho,\rho] \rho^\dag
\]
in the polarisation $\mathbb{P}$ reads:
\begin{equation}\label{e:OmegaBF_Apol}
    \mathbb{\Omega}_{BF} \doteq \mathfrak{q}(\mathbb{S}_{BF})\vert_{\mathbb{V}_{BF}} = -i\hbar \mathbb{Q}\vert_{\mathbb{V}_{BF}} = -i\hbar\int_{\Sigma} \left(d\rho + [\rho,A]\right)\frac{\delta}{\delta A} + \mu\frac{\delta}{\delta \chi} + [\rho, \tau] \frac{\delta}{\delta \tau} + \frac{1}{2} [\rho, \rho] \frac{\delta}{\delta \rho}, 
\end{equation}
and it satisfies $\mathbb{\Omega}_{BF}^2=[\mathbb{\Omega}_{BF},\mathbb{\Omega}_{BF}]=0$. 
\end{lemma}
\begin{proof}
    This is a direct consequence of the fact that $\{\mathbb{S}_{BF},\mathbb{S}_{BF}\}_{\bbomega_{BF}}=0$ and that quantisation is a Lie algebra morphism $\mathfrak{q}\colon (\mathcal{A}_{\mathbb{P}},\{\cdot,\cdot\}_{\bbomega_{BF}})$ and $(\mathrm{End}(\mathbb{V}_{BF}),[\cdot,\cdot])$, given that $\mathbb{S}_{BF}\in\mathcal{A}_{\mathbb{P}}$. To prove it {directly it is sufficient to show that, since $X_{\mathbb{S}_{BF}} = \mathbb{Q}$ and $\nabla_{X_f} s = L_{X_f}(s) +\frac{i}{\hbar}\iota_{X_f}\bbalpha_{BF} s + f s$ we get
    \[
    \mathfrak{q}(\mathbb{S}_{BF}) = \left(-i\hbar\mathbb{Q} - \iota_\mathbb{Q}\bbalpha_{BF}  + \mathbb{S}_{BF}\right)\vert_{\mathbb{V}_{BF}} = -i\hbar \mathbb{Q}\vert_{\mathbb{V}_{BF}}
    \]
    since 
    \[
    \iota_\mathbb{Q}\bbalpha_{BF} = \pm \trintl{\Sigma} Bd_A\rho \pm A^\dag \mu \pm B^\dag[\rho,\tau] \pm \frac12 \rho^\dag [\rho,\rho] = \mathbb{S}_{BF}
    \]
    up to boundary terms, which vanish because $\partial\Sigma=\emptyset$. Then $\mathbb{\Omega}_{BF}$ squares to zero because $\mathbb{Q}$ does.}

\end{proof}

Since $\mathbb{\Omega}_{BF}$ squares to zero, it makes sense to consider its cohomology. Let us define

\begin{definition}
    The double BFV quantum cochain complex is $\mathbb{V}_{BF}\doteq C^\infty(\mathbb{B}_{BF})$, where 
    \[
    \mathbb{B}_{BF}\doteq \mathcal{A}(\Sigma) \oplus \Omega^0[1]( \Sigma, \wedge^2\mathcal{V})\oplus \Omega^0[1]( \Sigma, \wedge^1\mathcal{V})\oplus\Omega^0[1]( \Sigma, \wedge^2\mathcal{V})\oplus \Omega^0[2]( \Sigma, \wedge^2\mathcal{V}) ,
    \]
    endowed with the differential given by the double quantum BFV operator $\mathbb{\Omega}_{BF}$ of Lemma \ref{l:Omega_R2}.
    
    \noindent The \emph{kinematical space of states of 3d general relativity} is the cohomology in degree zero of the double BFV quantum cochain complex $(\mathbb{V}_{BF},\mathbb{\Omega}_{BF})$:
\begin{align*}
    \mathbb{V}_{\mathrm{kin}} \doteq H^0(\mathbb{V}_{BF},\mathbb{\Omega}_{BF}).
\end{align*}
\end{definition}

On this space we want to define a BFV operator. In order to do so we resort to the construction of Section \ref{s:quantisation_BFVonBFV}: we define an operator $\check{\Omega}_{BF}\doteq \mathfrak{q}(\check{S}_{BF})$ on $\mathbb{V}_{BF}$ such that it satisfies \eqref{e:conditions_Omega_descends1} and \eqref{e:conditions_Omega_descends2}. In particular we want $\check{\Omega}_{BF}$ to commute with $\mathbb{\Omega}_{BF}$ and to be such that 
\begin{align*}
    [\check{\Omega}_{BF}^2]_{\mathbb{\Omega}_{BF}} = 0
\end{align*}
where the bracket denotes the $\mathbb{\Omega}_{BF}$-cohomology class.
 In virtue of Lemma \ref{l:operator_and_commutator} we can simply consider the quantisation of (cf. Remark \ref{r:S_m})
\[
\wt{S}_{\mathrm{res}} \doteq \check{S}_{BF} + \mathbb{Q}_{BF}\left( \chi\rho^\dag + \frac{1}{2}[\chi,\chi]\mu^\dag\right)= \trintl{\Sigma} \tau F_A
\]
defined in \eqref{e:S_rescheck}, and which is particularly convenient.

\begin{lemma}
The quantisation of $\wt{S}_{\mathrm{res}}$ yields the multiplication operator:
    \begin{align} \label{e:quantumBFVoperatorBF}
        \wt{\Omega}_{\mathrm{res}} \doteq \mathfrak{q}(\wt{S}_{\mathrm{res}}) = \int_\Sigma \tau F_A
    \end{align}
    on polarised sections.
    The following relations hold:
    \begin{align*}
        [\mathbb{\Omega}_{BF}, \wt{\Omega}_{\mathrm{res}}]&=0\\
        [\wt{\Omega}_{\mathrm{res}},\wt{\Omega}_{\mathrm{res}}]&=0.
    \end{align*}
\end{lemma}
\begin{proof}
The (pre-)quantization map sends $\wt{S}_{\mathrm{res}}$ to 
    \[
    \mathfrak{q}(\wt{S}_{\mathrm{res}}) = -i\hbar\nabla_{X_{\wt{S}_{\mathrm{res}}}} + \wt{S}_{\mathrm{res}} = -i\hbar X_{\wt{S}_{\mathrm{res}}} - \iota_{X_{\wt{S}_{\mathrm{res}}}}\bbalpha_{BF} + \wt{S}_{\mathrm{res}}
    \]
    where it is easy to check that
    \[
    X_{\wt{S}_{\mathrm{res}}} = \int_\Sigma F_A \frac{\delta}{\delta B^\dag} + d_A\tau \frac{\delta}{\delta B} \Longrightarrow \iota_{X_{\wt{S}_{\mathrm{res}}}}\bbalpha_{BF} =0
    \]
    The prequantization of $\wt{S}_{\mathrm{res}}$ then reads 
    \[
    \mathfrak{q}(\wt{S}_{\mathrm{res}}) = -\hbar\int_{\Sigma} \left(F_A \frac{\delta}{\delta B^\dag} + d_A\tau \frac{\delta}{\delta B}\right) + \int_\Sigma \tau F_A
    \]
    acting on all functions (sections of the trivial line bundle). However, on polarised sections $s\in \mathbb{V}_{BF}\doteq C^\infty(\mathbb{B}_{BF})$ one has
    \[
    \wt{\Omega}_{\mathrm{res}}\equiv\mathfrak{q}(\wt{S}_{\mathrm{res}})(s) = \left(\int_\Sigma \tau F_A\right)\cdot s
    \]
    since such sections do not depend on either $B^\dag$ or $B$.

    This operator clearly squares to zero because $\int_\Sigma \tau F_A$ is an odd functional. On the other hand, the bracket $[\mathbb{\Omega}_{BF}, \wt{\Omega}_{\mathrm{res}}]$ vanishes due to 
    \[
    \mathbb{\Omega}_{BF}\equiv \mathbb{Q} (\int_\Sigma \tau F_A) = 0.
    \]
\end{proof}

We can collect the results of this section in the following statement.

\begin{theorem}\label{thm:quantisationEH_doubleBFV}
    A quantisation of three-dimensional general relativity is given by the following data: the kinematical space of states  
    \begin{align*}
    \mathbb{V}_{\mathrm{kin}} = H^0(\mathbb{V}_{BF},\mathbb{\Omega}_{BF}),
    \end{align*}
    endowed with the quantum BFV operator $[\wt{\Omega}_{\mathrm{res}}]_{\mathbb{\Omega}_{BF}}$ and the \emph{space of physical states of 3d general relativity}, defined by
    \[
    \mathbb{V}_{\mathrm{phys}} \doteq H^0(\mathbb{V}_{\mathrm{kin}},[\wt{\Omega}_{\mathrm{res}}]_{\mathbb{\Omega}_{BF}})
    \]
    where
    \begin{align*} 
        \wt{\Omega}_{\mathrm{res}} = \int_{\Sigma}  {\tau F_A} 
    \end{align*}
    and 
    \begin{align*}
        {\mathbb{\Omega}_{BF}(f)= -i\hbar \mathbb{Q}(f) = -i\hbar\int_{\Sigma} \left\{\left(d\rho + [\rho,A]\right)\frac{\delta}{\delta A} + \mu\frac{\delta}{\delta \chi} + [\rho, \tau] \frac{\delta}{\delta \tau}  +\frac{1}{2} [\rho, \rho] \frac{\delta}{\delta \rho}\right\} (f).}
    \end{align*}
    
\end{theorem}

\subsection{A comment on nondegeneracy of triads}\label{sec:nondegeneracy}
{
Compare the geometric quantization of $\wt{M}=\mathbb{R}^{2n}$ and that of $M=\mathbb{R}^n\times \mathbb{R}^n\backslash\{0\}$, where we remove the zero section, with respect the vertical polarisation $\{\frac{\partial}{\partial p}\}$ associated to the projection to the first factor---the space of leaves of the polarisartion---which is $\mathbb{R}^n$ in both cases. If we look at the trivial complex line bundle over $\wt{M}$ and $M$, we see that polarised sections are $p$-constant functions. Since we are looking at $p$-constants, there is no possibility for polarised sections to become singular when removing the zero section, so that in both cases the space of polarised sections is isomorphic to $C^\infty(\mathbb{R}^n)$, functions on the space of leaves. Quantisable functions (i.e.\ polarisation preserving functions) are in both cases $p$-linear functions on $\wt{M}$ and $M$.

It would seem then that there is no retention of global information about the fibres of the polarisation in this approach to quantisation, but one has to keep in mind that the algebra that is represented on the space of quantum states (i.e.\ polarised sections) changes when going from $C^\infty(M,\mathbb{C})$ to $C^\infty(\wt{M},\mathbb{C})$, so strictly speaking the two cases produce two different representations on the same vector space. 

It is known that to extend the quantisable functions beyond those that are immediately compatible with the chosen polarisation, one has to take into account a more involved procedure that goes under the name of BKS kernel (after Blattner, Kostant and Souriau). We expect some differences to emerge already at this level.

In the field theoretic scenario of General Relativity, one works with an exact symplectic manifold (in degree zero) $\mathcal{F}=\Omega^1_{\mathrm{nd}}(\Sigma,\mathcal{V})\times \mathcal{A}(\Sigma,P)$, where $\mathrm{nd}$ stands for nondegenerate. (A one-form with values in $\mathcal{V}$ is nondegenerate if it defines an injective bundle map $T\Sigma\to \mathcal{V}$.) It is important to note that both factors are affine spaces, and their direct product should be compared to the densitised cotangent bundle 
\[
T^\vee \mathcal{A}(\Sigma,P)=\Omega^1(\Sigma,\mathcal{V})\times \mathcal{A}(\Sigma,P)
\]
where instead we allow degenerate maps $B\colon T\Sigma\to \mathcal{V}$. On $T^\vee\mathcal{A}$ there is a (weakly) nondegenerate local symplectic form, which is also exact.

Note that, importantly, $\Sigma$ is two dimensional and assumed compact and orientable, and we fix $\mathcal{V}$ to be a trivial vector bundle on $\Sigma$ with fibres given by $\mathbb{R}^3$ with Minkowski metric. Then one can see that the space of nondegenerate bundle maps $B\in \Omega^1_{\mathrm{nd}}(\Sigma,\mathcal{V})$ is not empty, and it is equivalent to asking that $T\Sigma$ is injected as a subbundle in $\mathcal{V}$.\footnote{This is a consequence of Whitney's embedding theorem applied to $n=2$, since $\Sigma$ embeds in $\mathbb{R}^3$ and this induces an embedding of $T\Sigma$ in a trivial vector bundle over $\Sigma$. We thank Jonas Schnitzer for this argument.}

Choosing the ``vertical'' polarisation associated to the surjective submersion $T^\vee \mathcal{A}(\Sigma,P)\to \mathcal{A}(\Sigma,P)$, which is generated by $\{\frac{\delta}{\delta B}\}$, we can construct polarised section given by $B$-constant functions, which are isomorphic to functions over $\mathcal{A}(\Sigma,P)$. These can be restricted to the \emph{nondegenerate sector}, i.e.\ to $B$-constant functions on $\mathcal{F}$, and once again these will be isomorphic to functions over the base of the polarisation $\mathcal{A}(\Sigma,P)$.

From this point of view, then, the space of states constructed via geometric quantisation of the Hamiltonian theory of GR viewed as a 3d BF theory is not sensitive to the request that the diads $B$ should be nondegenerate, and one gets the same answer. However, the algebras of functions that are represented on such space of states, i.e.\ observables, may differ.

Taking a different polarisation altogether, for example the one that is vertical for the projection $\mathcal{F} \to \Omega^1_{\mathrm{nd}}(\Sigma,\mathcal{V})$ will produce a potentially different quantisation. It remains to be established, in the field theoretic case, whether it is reasonable to show that these two quantisation are equivalent in some sense. In the next section we explore the quantisation one would obtain using such alternative polarisation.
}

\subsection{Alternative polarisation}
We conclude by showing what the output of our procedure may be, if we chose a different polarisation for both steps of the BFV geometric quantisation. Although there are obviously many possible options, the main relevant feature here is that we will choose to switch polarisation of the degree zero fields, with the net result of keeping functions of the $B$ field, instead of the connections $A$. This is akin to choosing to have the quantum space be some space of functions on metrics.

Recall that
\[
\mathbb{S}_{BF} = \trintl{\Sigma} \rho (d_A B + [\tau, B^{\dag}])+ \mu A^{\dag} + \frac{1}{2}[\rho,\rho] \rho^\dag
\]
while 
\[
\wt{S}_{\mathrm{res}} = \trintl{\Sigma}\tau F_A
\]
Moreover
\[
\bbomega_{BF}
    = \trintl{\Sigma} \delta B \delta A + \delta \chi \delta A^\dag + \delta \tau \delta B^\dag + \delta\mu^\dag \delta\mu + \delta \rho^\dag\delta\rho = \delta \bbalpha_{BF},
\]
and we will choose the polarisation generated by
\[
\mathbb{P}'\colon \left\{\frac{\delta}{\delta \mu^\dag},\frac{\delta}{\delta \rho^\dag}, \frac{\delta}{\delta B^\dag}, {\frac{\delta}{\delta\chi}}, \frac{\delta}{\delta A} 
\right\}.
\]

It is important to note that the function $\wt{S}_{\mathrm{res}}$ is not polarisation compatible, since it is quadratic in fibre variables. This is an issue shared by the standard geometric quantisation of quadratic Hamiltonians, and it is usually resolved by means of BKS kernels, which allow one to extend the class of quantisable functions. However, that relies on certain gometric structures on the space of leaves of the polarisation which are not readily available in infinite dimensions. We shall then use the \emph{Schr\"odinger quantisation ansatz} (with normal ordering) for this case, and simply think of the Polynomial functions $\wt{S}_{\mathrm{res}}, \mathbb{S}_{BF}$ as functions of the operators $\Phi, -i\hbar\frac{\delta}{\delta \Phi}$ for every variable $\Phi$ on the base of the polarisation
\[
\mathbb{B}'_{BF} \doteq \FF_{BF}/\mathbb{P}' \simeq  \underbrace{\Omega^1_{\mathrm{nd}}(\Sigma)}_{B} \oplus  \underbrace{\Omega^{\mathrm{top}}[1](\Sigma, \mathcal{V})}_{A^\dag} \oplus \underbrace{\Omega^0[-1]( \Sigma, \wedge^1\mathcal{V})}_{\tau}\oplus\underbrace{\Omega^0[1]( \Sigma, \wedge^2\mathcal{V})}_{\rho}\oplus \underbrace{\Omega^0[2]( \Sigma, \wedge^2\mathcal{V})}_{\mu} 
\]
which is parametrised by field variables $(B,A^\dag, \tau, \mu, \rho)$.

Define by 
\[
\mathbb{V}_{BF}' \doteq C^\infty(\mathbb{B}'_{BF}).
\]
Then, we obtain
\[
\wt{S}_{\mathrm{res}} \leadsto \wt{\Omega}'_{\mathrm{res}}= - i \hbar d \tau \frac{\delta}{\delta B} - \frac{\hbar^2}{2} \tau \left[\frac{\delta}{\delta B},\frac{\delta}{\delta B}\right],
\]
as well as
\begin{align*}
        \mathbb{S}_{BF}\leadsto \mathbb{\Omega}'_{BF} = \int_{\Sigma} \rho dB  -i \hbar [\rho, B] \frac{\delta}{\delta B}-i \hbar [\rho, \tau] \frac{\delta}{\delta \tau} + \mu A^\dag - i \hbar\frac{1}{2} [\rho, \rho] \frac{\delta}{\delta \rho}.
    \end{align*}
    These should be compared, respectively with \eqref{e:quantumBFVoperatorBF} and \eqref{e:OmegaBF_Apol}.

Since this ansatz is not given to us by a quantisation map with the property of being a Lie algebra morphism, we need to check that $(\mathbb{\Omega}'_{BF})^2=0$ \emph{directly}. Let $f \in \mathbb{V}_{BF}$. We have 
    \begin{align*}
        \mathbb{\Omega}'_{BF}(f) = \int_{\Sigma} \rho dBf -i \hbar [\rho, B] \frac{\delta f}{\delta B}-i \hbar [\rho, \tau] \frac{\delta f}{\delta \tau} + \mu A^\dag f-i \hbar\frac{1}{2} [\rho, \rho] \frac{\delta f}{\delta \rho}.
    \end{align*}
    and\footnote{We could prove this statement alternatively by observing that the Schroedinger quantisation of $\mathbb{S}_{BF}$ in the polarisation $\mathcal{P}'$ coincides with its geometric quantisation, since $\mathbb{S}_{BF}$ is a polarisable function, hence $\mathbb{\Omega}_{BF}'\equiv \mathfrak{q}_{\mathcal{P}'}(\mathbb{S}_{BF})$ squares to zero automatically. We leave it regardless, as it proves to be a useful warm-up calculation for what follows.} 
    \begin{align*}
        \mathbb{\Omega}'_{BF}(\mathbb{\Omega}'_{BF}(f)) &= \int_{\Sigma} \rho dB \int_{\Sigma} \rho dB f -\int_{\Sigma} \rho dB \int_{\Sigma}i \hbar [\rho, B] \frac{\delta f}{\delta B}-\int_{\Sigma} \rho dB \int_{\Sigma}i \hbar [\rho, \tau] \frac{\delta f}{\delta \tau}\\
        &\phantom{=} +\int_{\Sigma} \rho dB \int_{\Sigma} \mu A^\dag f-\int_{\Sigma} \rho dB \int_{\Sigma}i \hbar\frac{1}{2} [\rho, \rho] \frac{\delta f}{\delta \rho}  -i \hbar \int_{\Sigma}[\rho, B] \frac{\delta }{\delta B} \int_{\Sigma} \rho dB f \\
        &\phantom{=}+ \hbar^2 \int_{\Sigma}[\rho, B] \frac{\delta }{\delta B} \int_{\Sigma} [\rho, B] \frac{\delta f}{\delta B}+ \hbar^2 \int_{\Sigma}[\rho, B]\int_{\Sigma} [\rho, \tau] \frac{\delta^2 f}{\delta B\delta \tau} -i \hbar \int_{\Sigma}[\rho, B]  \int_{\Sigma} \mu A^\dag \frac{\delta f}{\delta B}\\
        &\phantom{=}+ \hbar^2 \int_{\Sigma}[\rho, B]  \int_{\Sigma}\frac{1}{2} [\rho, \rho] \frac{\delta^2 f}{\delta B\delta \rho}-i \hbar \int_{\Sigma} [\rho, \tau]  \int_{\Sigma} \rho dB \frac{\delta f}{\delta \tau}+ \hbar^2 \int_{\Sigma} [\rho, \tau]  \int_{\Sigma} [\rho, B] \frac{\delta^2 f}{\delta \tau\delta B}\\
        &\phantom{=}+ \hbar^2 \int_{\Sigma} [\rho, \tau] \frac{\delta }{\delta \tau} \int_{\Sigma} [\rho, \tau] \frac{\delta f}{\delta \tau} -i \hbar \int_{\Sigma} [\rho, \tau]  \int_{\Sigma} \mu A^\dag \frac{\delta }{\delta \tau}+ \hbar^2 \int_{\Sigma} [\rho, \tau]  \int_{\Sigma}\frac{1}{2} [\rho, \rho] \frac{\delta^2 f}{\delta \tau\delta \rho} \\
        &\phantom{=}+ \int_{\Sigma} \mu A^\dag \int_{\Sigma} \rho dB f-i \hbar \int_{\Sigma} \mu A^\dag \int_{\Sigma}[\rho, B] \frac{\delta f}{\delta B}-i \hbar \int_{\Sigma} \mu A^\dag \int_{\Sigma}[\rho, \tau] \frac{\delta f}{\delta \tau}\\
        &\phantom{=} + \int_{\Sigma} \mu A^\dag \int_{\Sigma}\mu A^\dag f-i \hbar\int_{\Sigma} \mu A^\dag \int_{\Sigma}\frac{1}{2} [\rho, \rho] \frac{\delta f}{\delta \rho} -i \hbar\int_{\Sigma}\frac{1}{2} [\rho, \rho] \frac{\delta }{\delta \rho} \int_{\Sigma} \rho dB f \\
        &\phantom{=}+ \hbar^2\int_{\Sigma}\frac{1}{2} [\rho, \rho] \frac{\delta }{\delta \rho} \int_{\Sigma} [\rho, B] \frac{\delta f}{\delta B}+ \hbar^2\int_{\Sigma}\frac{1}{2} [\rho, \rho] \frac{\delta }{\delta \rho} \int_{\Sigma} [\rho, \tau] \frac{\delta f}{\delta \tau}\\
        &\phantom{=} -i \hbar\int_{\Sigma}\frac{1}{2} [\rho, \rho]  \int_{\Sigma} \mu A^\dag\frac{\delta f}{\delta \rho} + \hbar^2\int_{\Sigma}\frac{1}{2} [\rho, \rho] \frac{\delta }{\delta \rho} \int_{\Sigma}-i \hbar\frac{1}{2} [\rho, \rho] \frac{\delta f}{\delta \rho}.
    \end{align*}
    
    All the terms without derivatives contain two degree-1 functions ($\int\rho dB$ and $\int \mu A^\dag$) and hence vanish. Instead of considering every single term let us now consider as an example the following terms:
    \begin{align*}
        -i \hbar \int_{\Sigma}[\rho, B] \frac{\delta }{\delta B} \int_{\Sigma} \rho dB f &= i \hbar \int_{\Sigma}[\rho, B]  \int_{\Sigma} \rho dB \frac{\delta f}{\delta B}- i \hbar \int_{\Sigma} d \rho [\rho, B] f \\
         -i \hbar\int_{\Sigma}\frac{1}{2} [\rho, \rho] \frac{\delta }{\delta \rho} \int_{\Sigma} \rho dB f  &=  i \hbar\int_{\Sigma}\frac{1}{2} [\rho, \rho]  \int_{\Sigma} \rho dB \frac{\delta f}{\delta \rho}-i \hbar\int_{\Sigma}\frac{1}{2} [\rho, \rho] dB f.
    \end{align*}
    The last terms in the two lines are equal and cancel out, while the first two are equal to analogous terms in the sum (namely, the second in the first row and the first in the second row respectively). Similarly one can find relations for all the other terms contained in $(\mathbb{\Omega}_{BF}')^2(f)$ and conclude that it vanishes.

    We want to show now that the two quantum operators $\wt{\Omega}'_{\mathrm{res}}$ and $\mathbb{\Omega}'_{BF}$ behave well, namely:
    \[
    [\mathbb{\Omega}'_{BF}, \wt{\Omega}'_{\mathrm{res}}]=0, \qquad [\wt{\Omega}'_{\mathrm{res}}, \wt{\Omega}'_{\mathrm{res}}]=0.
    \]
    
    We proceed as in the proof of Lemma \ref{l:Omega_R2}. 
    \begin{align*}
        [\mathbb{\Omega}'_{BF}, \wt{\Omega}'_{\mathrm{res}}](f)&= \left[\int_{\Sigma} \rho dB, \int_{\Sigma}  - i \hbar d \tau\frac{\delta}{\delta B} - \frac{\hbar^2}{2} \tau \left[\frac{\delta}{\delta B},\frac{\delta}{\delta B}\right]\right](f)\\
        &\phantom{=}+ \left[\int_{\Sigma}-i \hbar[\rho, B] \frac{\delta }{\delta B},\int_{\Sigma}  - i \hbar d \tau \frac{\delta}{\delta B} - \frac{\hbar^2}{2} \tau \left[\frac{\delta}{\delta B},\frac{\delta}{\delta B}\right]\right](f)\\
        &\phantom{=}+ \left[\int_{\Sigma}-i \hbar[\rho, \tau] \frac{\delta }{\delta \tau} , \int_{\Sigma}  - i \hbar d \tau \frac{\delta}{\delta B} - \frac{\hbar^2}{2} \tau \left[\frac{\delta}{\delta B},\frac{\delta}{\delta B}\right]\right](f)\\
        &\phantom{=}
        + \left[\int_{\Sigma} \mu A^\dag -i \hbar\frac{1}{2} [\rho, \rho] \frac{\delta }{\delta \rho}, \int_{\Sigma}  - i \hbar d \tau \frac{\delta}{\delta B} - \frac{\hbar^2}{2} \tau \left[\frac{\delta}{\delta B},\frac{\delta}{\delta B}\right]\right](f).
    \end{align*}
    Let us start from the fourth line. Since the operator on the right and the one on the left contain different derivatives and fields, the result of the commutator is zero. Now let us consider the other three lines.
    \begin{align*}
        \left[\int_{\Sigma} \rho dB, \int_{\Sigma}  - i \hbar d\tau \frac{\delta}{\delta B} \right](f)&= - i \hbar \left(\int_{\Sigma} \rho dB \int_{\Sigma}   d B \frac{\delta f}{\delta B} - \int_{\Sigma}   d \tau \frac{\delta }{\delta B}\int_{\Sigma} \rho dB f\right)\\
         & =- i \hbar \left(\int_{\Sigma} \rho dB \int_{\Sigma}   d \tau \frac{\delta f}{\delta B} - \int_{\Sigma}   d \tau \int_{\Sigma} \rho dB \frac{\delta }{\delta B}f- \int_{\Sigma}   d\tau d(\rho f)\right)=0
    \end{align*}
    where we used that $\int_{\Sigma}   d\tau d(\rho f)= \int_{\Sigma}   d^2 \tau \rho f=0$. Similarly
    \begin{align*}
        \hbar^2 \left[\int_{\Sigma}[\rho, B] \frac{\delta }{\delta B},\int_{\Sigma} d \tau \frac{\delta}{\delta B} \right](f) &  =
        \hbar^2 \left(\int_{\Sigma}[\rho, B] \frac{\delta }{\delta B}\int_{\Sigma} d \tau \frac{\delta f}{\delta B}- \int_{\Sigma} d \tau \frac{\delta }{\delta B}\int_{\Sigma}[\rho, B] \frac{\delta f}{\delta B} \right)\\
        & = \hbar^2 \left(\int_{\Sigma}[\rho, B]\int_{\Sigma} d \tau \frac{\delta^2 f}{\delta B \delta B} \right. \\
        & \phantom{=\hbar^2 (} \left.-\int_{\Sigma}d\tau\int_{\Sigma} [\rho, B] \frac{\delta^2 f}{\delta B \delta B}-\int_{\Sigma}[\rho,d\tau]\frac{\delta f}{\delta B}\right)\\
        & = \hbar^2 \int_{\Sigma} [\rho, d\tau]\frac{\delta f}{\delta B}.
    \end{align*}
    On the other hand we also have 
    \begin{align*}
         \frac{\hbar^2}{2} \left[\int_{\Sigma} d\rho B, \int_{\Sigma} \tau \left[\frac{\delta}{\delta B},\frac{\delta}{\delta B}\right]\right]f&=
         \frac{\hbar^2}{2}\left(\int_{\Sigma} d\rho B \int_{\Sigma} \tau \left[\frac{\delta}{\delta B},\frac{\delta}{\delta B}\right]f- \int_{\Sigma} \tau \left[\frac{\delta}{\delta B},\frac{\delta }{\delta B}\right]\int_{\Sigma} d\rho B f\right)\\
         &=\hbar^2 \int_{\Sigma}d\rho \left[\tau, \frac{\delta f}{\delta B}\right]=\hbar^2 \int_{\Sigma}\left[d\rho, \tau\right] \frac{\delta f}{\delta B}
    \end{align*}
    and 
    \begin{align*}
        \hbar^2  \left[\int_{\Sigma}[\rho, \tau] \frac{\delta }{\delta \tau} , \int_{\Sigma}   d \tau \frac{\delta}{\delta B} \right](f)&= 
         \hbar^2 \left( \int_{\Sigma}[\rho, \tau] \frac{\delta }{\delta \tau} \int_{\Sigma}   d \tau \frac{\delta}{\delta B}-  \int_{\Sigma}   d \tau \frac{\delta}{\delta B}\int_{\Sigma}[\rho, \tau] \frac{\delta }{\delta \tau}\right)f\\
         & = \hbar^2  \int_{\Sigma}[\rho, \tau] d \frac{\delta f }{\delta B}
    \end{align*}

    which cancels out with the result of the two previous computations. Then we have
    \begin{align*}
        i \frac{\hbar^3}{2}&\left[\int_{\Sigma}[\rho, B] \frac{\delta }{\delta B},\int_{\Sigma}    \tau \left[\frac{\delta}{\delta B},\frac{\delta}{\delta B}\right]\right](f) \\
        &= i \frac{\hbar^3}{2}\left(\int_{\Sigma}[\rho, B] \frac{\delta }{\delta B}\int_{\Sigma}    \tau \left[\frac{\delta}{\delta B},\frac{\delta}{\delta B}\right]f-\int_{\Sigma}    \tau \left[\frac{\delta}{\delta B},\frac{\delta}{\delta B}\right]\int_{\Sigma}[\rho, B] \frac{\delta f}{\delta B}\right)\\
        &= i \frac{\hbar^3}{2}\left(\int_{\Sigma}[\rho, B] \int_{\Sigma}    \tau \frac{\delta }{\delta B}\left[\frac{\delta}{\delta B},\frac{\delta}{\delta B}\right]f \right.\\
        & \phantom{= i \frac{\hbar^3}{2}(}\left.-\int_{\Sigma}    \tau \int_{\Sigma}[\rho, B] \left[\frac{\delta}{\delta B},\frac{\delta}{\delta B}\right]\frac{\delta f}{\delta B}-2\int_{\Sigma}  \left[\left[  \tau, \frac{\delta}{\delta B}\right],\rho \right] \frac{\delta f}{\delta B}\right)\\
         &= i \frac{\hbar^3}{2}\left(+\int_{\Sigma}\left[[\rho, \tau], \frac{\delta}{\delta B}\right]\frac{\delta}{\delta B}f\right)
    \end{align*}
    where in the last passage we have used once again the graded Jacobi identity. 
    Lastly we have 
    \begin{align*}
        i \frac{\hbar^3}{2}&\left[\int_{\Sigma}[\rho, \tau] \frac{\delta }{\delta \tau} , \int_{\Sigma}\tau \left[\frac{\delta}{\delta B},\frac{\delta}{\delta B}\right]\right](f) \\
        & = i \frac{\hbar^3}{2}\left(\int_{\Sigma}[\rho, \tau] \frac{\delta }{\delta \tau}\int_{\Sigma}\tau \left[\frac{\delta}{\delta B},\frac{\delta}{\delta B}\right] - \int_{\Sigma}\tau \left[\frac{\delta}{\delta B},\frac{\delta}{\delta B}\right]\int_{\Sigma}[\rho, \tau] \frac{\delta }{\delta \tau} \right)(f) \\
        & = i \frac{\hbar^3}{2}\left(\int_{\Sigma}[\rho, \tau] \left[\frac{\delta}{\delta B},\frac{\delta}{\delta B}\right]f \right)
    \end{align*}
    which using the Jacobi identity is the same as before.
    Hence we conclude that $[\mathbb{\Omega}'_{BF}, \wt{\Omega}'_{\mathrm{res}}](f)=0$ for all $f$. For the second equation we have:
    \begin{align*}
        [\wt{\Omega}'_{\mathrm{res}}, \wt{\Omega}'_{\mathrm{res}}](f)&= \left[\int_{\Sigma}  - i \hbar d \tau \frac{\delta}{\delta B} - \frac{\hbar^2}{2} \tau \left[\frac{\delta}{\delta B},\frac{\delta}{\delta B}\right],\int_{\Sigma}  - i \hbar d \tau \frac{\delta}{\delta B} - \frac{\hbar^2}{2} \tau\left[\frac{\delta}{\delta B},\frac{\delta}{\delta B}\right]\right](f)
    \end{align*}
    which vanishes since all the derivatives are with respect to $B$ and only $\tau$ appears.    Hence we conclude that $[\wt{\Omega}'_{\mathrm{res}}, \wt{\Omega}'_{\mathrm{res}}]=0$.

As a consequence, we get
\begin{theorem}\label{thm:quantisationEH_doubleBFV_BPOL}
    A quantisation of three-dimensional general relativity is given by the following data: the kinematical space of states  
    \begin{align*}
    \mathbb{V}'_{\mathrm{kin}} = H^0(\mathbb{V}'_{BF},\mathbb{\Omega}'_{BF}),
    \end{align*}
    endowed with the quantum BFV operator $[\wt{\Omega}'_{\mathrm{res}}]_{\mathbb{\Omega}'_{BF}}$ and the \emph{space of physical states of 3d general relativity}, defined by
    \[
    \mathbb{V}'_{\mathrm{phys}} \doteq H^0(\mathbb{V}_{\mathrm{kin}},[\wt{\Omega}_{\mathrm{res}}]_{\mathbb{\Omega}_{BF}})
    \]
    where
    \begin{align*} 
        \wt{\Omega}'_{\mathrm{res}}= - i \hbar d \tau \frac{\delta}{\delta B} - \frac{\hbar^2}{2} \tau \left[\frac{\delta}{\delta B},\frac{\delta}{\delta B}\right],
    \end{align*}
    and 
    \begin{align*}
        \mathbb{\Omega}'_{BF} = \int_{\Sigma} \rho dB  -i \hbar [\rho, B] \frac{\delta}{\delta B}-i \hbar [\rho, \tau] \frac{\delta}{\delta \tau} + \mu A^\dag - i \hbar\frac{1}{2} [\rho, \rho] \frac{\delta}{\delta \rho}.
    \end{align*}
    
\end{theorem}

\appendix
\section{Graded coisotropic reduction of \texorpdfstring{$BF$}{BF} theory} \label{a:coisotropic_BFtoEH}

    In this appendix, we give a precise description of the coisotropic reduction of the BFV formulation of $BF$ theory with respect to the submanifold defined by \eqref{e:ConstraintBF1} and \eqref{e:ConstraintBF2}.

        Consider the restriction of the symplectic form \eqref{e:BF_sympl_form} to the submanifold defined by \eqref{e:ConstraintBF1} and \eqref{e:ConstraintBF2}. The characteristic distribution is defined by the kernel of such restriction, which is also generated by the Hamiltonian vector fields $X_{\mathsf{J}^\rho}$ and $X_{\mathsf{M}^\mu}$ of $\mathsf{J}^\rho$ and $\mathsf{M}^\mu$, which have already been computed in the proof of Lemma \ref{l:brackets_constraints} and read (we drop $\rho$ and $\mu$ for ease of notation)
    \begin{align*}
        {X}_{\mathsf{J}}(B)&= [\rho, B] & {X}_{\mathsf{J}}(A)&= d_A \rho & {X}_{\mathsf{J}}(\tau)&= -[\rho, \tau]\\
        {X}_{\mathsf{J}}({B^\dag})&= -[\rho, B^{\dag}] & {X}_{\mathsf{J}}(\chi)&= 0& {X}_{\mathsf{J}}({A^\dag})&= 0\\
        {X}_{\mathsf{M}}(B)&= 0 & {X}_{\mathsf{M}}(A)&= 0 & {X}_{\mathsf{M}}(\tau)&= 0\\
        {X}_{\mathsf{M}}({B^\dag})&= 0 & {X}_{\mathsf{M}}(\chi)&= \mu & {X}_{\mathsf{M}}({A^\dag})&= 0.
    \end{align*}
    These vector fields are generated by the parameters $\mu, \rho \in \Omega^{0,2}$. Let us start by flowing along the vector field generated by $\mu$. We get the following equation:
    \begin{align*}
        \dot{\chi}(t)= \mu.
    \end{align*}
    Solving it we get $\chi(t)=\mu t + \chi(0)$ and we can set $\chi(1)=0$ by choosing $\mu= -\chi(0)$. Hence in the reduction we will have $\chi'=0$ and all the other fields remain unchanged. 

    Let us now switch to the vector field generated by $\rho$. The corresponding equations read
    \begin{align*}
        \dot{B}(t)&= [\rho, B(t)] & \dot{A}(t)&= d_{A(t)} \rho & \dot{\tau}(t)&= -[\rho, \tau(t)]\\
        \dot{B^\dag}(t)&= -[\rho, B^{\dag}(t)] & \dot{\chi}(t)&= 0& \dot{A^\dag}(t)&= 0
    \end{align*}
    We now fix a basis $\{v_0,v_1,v_2\}$ of $\mathcal{V}_{\Sigma}$ such that in components $B_{\rho}^i$ ($i=1,2$) is an invertible diad (in components, $B$ is represented by a $3\times 2$ matrix with at least an invertible $2\times 2$ minor, by the non-degeneracy property of $B$).

    Using this basis we can rewrite the differential equations in components as 
    \begin{align*}
        \dot{B}_{\nu}^i(t)&= \rho^{0i}B^{0}_{\nu}(t) \eta_{00} + \rho^{ij}B^{k}_{\nu}(t) \eta_{jk} & \dot{B}_{\nu}^0(t)&= \rho^{0i}B^{j}_{\nu}(t) \eta_{ij} \\
        \dot{A}^{0i}_{\nu}(t)&= d_{\nu} \rho^{0i} - A_{\nu}^{ij}(t)\rho^{0k} \eta_{jk} + A_{\nu}^{0j}(t)\rho^{ki} \eta_{jk}  & \dot{A}^{ij}_{\nu}(t)&= d_{\nu} \rho^{ij} - A_{\nu}^{0i}(t)\rho^{0j}\eta_{00}\\
        \dot{\tau}^0(t)&=  \tau^{i}(t)\rho^{0j}\eta^{ij} & \dot{\tau}^i(t)&=  \tau^{j}(t)\rho^{0j}\eta^{ij} + \tau^{0}(t)\rho^{0i}\eta^{00}\\
        \dot{B^\dag}^{0i}_{\nu\mu}(t)&= - {B^\dag}_{\nu\mu}^{ij}(t)\rho^{0k} \eta_{jk} + {B^\dag}_{\nu\mu}^{0j}(t)\rho^{ki} \eta_{jk}  & \dot{B^\dag}^{ij}_{\nu\mu}(t)&=  - {B^\dag}_{\nu\mu}^{0i}(t)\rho^{0j}\eta_{00}\\
        \dot{\chi}(t)&= 0& \dot{A^\dag}(t)&= 0
    \end{align*}

    This system has three free parameters, $\rho^{12}$, $\rho^{01}$, $\rho^{02}$ and we can use them to set to zero three components of $B$. Note that the reduction with respect to these parameters corresponds  to the reduction with respect to the three components of the constraint $d_A B +[\tau, B^\dag]=0$, respectively components 0,1 and 2. 
    
    We start with the reduction with respect to $\rho^{01}$.
    The first equations of the previous system become
    \begin{align*}
        \dot{B}_{\nu}^1(t)&= \rho^{01}B^{0}_{\nu}(t) \\
        \dot{B}_{\nu}^0(t)&= -\rho^{01}B^{1}_{\nu}(t) 
    \end{align*}
    where we used that $\eta_{11}=1$. 
    The solution to this system is 
    \begin{align*}
        {B}_{\nu}^0(t) &= \cos{(\rho^{01}t)}{B}_{\nu}^0(0) + \sin{(\rho^{01}t)}{B}_{\nu}^1(0)\\
        {B}_{\nu}^1(t) &= -\sin{(\rho^{01}t)}{B}_{\nu}^0(0) + \cos{(\rho^{01}t)}{B}_{\nu}^1(0)
    \end{align*}
    The reduced variables will be set at the parameter $t=1$, while for $t=0$ we recover the old ones.    
    Hence we can set $\rho^{01}= \arctan{\left(-\frac{{B}_{0}^1(0)}{{B}_{1}^1(0)}\right)}$ and get for $t=1$,
    \begin{align*}
        (B_1^0)'&=0 & (B_1^1)'&= \alpha_{01} \\
        (B_2^0)'&= \frac{B_1^1 B_2^0-B_1^0 B_2^1}{\alpha_{01}} & (B_2^1)'&= \frac{B_1^0 B_2^0+B_2^1 B_1^1}{\alpha_{01}} \\
        (B_2^1)'&= B_2^1 & (B_2^2)'&=B_2^2
    \end{align*}
    where $\alpha_{01}=\sqrt{(B_1^0)^2+ (B_1^1)^2}$. Proceeding similarly, for the other fields one obtains:
       \begin{align*}
        (A_{\nu}^{02})'&= \frac{A_{\nu}^{02} B_1^1-B_1^0 A_{\nu}^{12}}{\alpha_{01}} & 
        ({B^\dag}_{\nu\mu}^{02})'&= \frac{{B^\dag}_{\nu\mu}^{02} B_1^1-B_1^0 {B^\dag}_{\nu\mu}^{12}}{\alpha_{01}} \\ 
        (A_{\nu}^{12})'&= \frac{A_{\nu}^{12} B_1^1+B_1^0A_{\nu}^{02}}{\alpha_{01}} & 
        ({B^\dag}_{\nu\mu}^{12})'&= \frac{{B^\dag}_{\nu\mu}^{12} B_1^1+B_2^0 {B^\dag}_{\nu\mu}^{02}}{\alpha_{01}} \\ 
        (A_{\nu}^{01})'&= A_{\nu}^{01} - \frac{ B_1^1 d_{\nu} B_1^0 - B_1^0 d_{\nu} B_1^1 }{\alpha^2_{01}} 
          &({B^\dag}_{\nu\mu}^{01})'&= {B^\dag}_{\nu\mu}^{01}  \\
          (\tau^{0})'&= \frac{\tau^{0} B_1^1-B_1^0 \tau^{1}}{\alpha_{01}}& (\tau^{2})'&= \tau^{2} \\
          (\tau^{1})'&= \frac{\tau^{1} B_1^1+B_1^0 \tau^{0}}{\alpha_{01}}
    \end{align*}

    Similarly we can flow along $\rho^{02}$ and obtain, by fixing $\rho^{02}= \arctan{\left(-\frac{({B}_{2}^0)'}{({B}_{2}^2)'}\right)}$ and setting $\alpha_{02}=\sqrt{(B_2^0)'^2+ (B_2^2)'^2}$, 
    \begin{align*}
        (B_2^2)''&= \alpha_{02} & (B_2^0)''&=0 \\
        (B_1^0)''&= \frac{(B_1^0)' (B_2^2)'-(B_2^0)' (B_1^2)'}{\alpha_{02}} & (B_1^2)''&= \frac{(B_1^0)' (B_2^0)'+(B_2^2)' (B_1^2)'}{\alpha_{02}} \\
        (B_2^1)''&=(B_2^1)' &  (B_1^1)''&= (B_1^1)'
    \end{align*}    
    and 
        \begin{align*}
        (A_{\nu}^{01})''&= \frac{(A_{\nu}^{01})' (B_2^2)'-(B_2^0)' (A_{\nu}^{02})'}{\alpha_{02}} & ({B^\dag}_{\nu\mu}^{01})''&= \frac{({B^\dag}_{\nu\mu}^{01})' (B_2^2)'-(B_2^0)' ({B^\dag}_{\nu\mu}^{12})'}{\alpha_{02}} \\ 
        (A_{\nu}^{12})''&= \frac{(A_{\nu}^{12})' (B_2^2)'+(B_2^0)'(A_{\nu}^{01})'}{\alpha_{02}} & ({B^\dag}_{\nu\mu}^{12})''&= \frac{({B^\dag}_{\nu\mu}^{12})' (B_2^2)'+(B_2^0)' ({B^\dag}_{\nu\mu}^{01})'}{\alpha_{02}} \\ 
        (A_{\nu}^{02})''&= (A_{\nu}^{02})' - \frac{ (B_2^2)' d_{\nu} (B_2^0)' - (B_2^0)' d_{\nu} (B_2^2)' }{\alpha^2_{02}} 
          &({B^\dag}_{\nu\mu}^{02})''&= ({B^\dag}_{\nu\mu}^{02})'  \\
          (\tau^{0})''&= \frac{(\tau^{0})' (B_2^2)'-(B_2^0)' (\tau^{2})'}{\alpha_{02}}& (\tau^{1})''&= (\tau^{1})' \\
          (\tau^{2})''&= \frac{(\tau^{2})' (B_2^2)'+(B_2^0)' (\tau^{0})'}{\alpha_{02}}
    \end{align*}
    \begin{remark}\label{r:intermediate_theory}
        Before proceeding further with the reduction with respect to $\rho^{12}$ let us note that we can consider an intermediate BFV theory defined on the space of \emph{double primed} fields given as follows:
        \begin{align*}
            \varpi_{\mathrm{int}} &= \int_{\Sigma} \delta (A_{\mu}^{0i})'' \delta (B_{\nu}^{j})''\epsilon_{ij}\epsilon^{\mu\nu} + \delta (\tau)'' \delta (B^\dag)''\\
            S_{\mathrm{int}} &= \int_{\Sigma} (\tau)'' F_{(A)''} \\
        \end{align*}
        where we still have to impose the constraint
        \begin{align*}
            (A_{[\mu}^{0i})'' (B_{\nu]}^j)'' \eta_{ij} +(\tau^i)'' (B^{\dag {0j}}_{\mu\nu})''\eta_{ij}=0. 
        \end{align*}
        and make the corresponding coisotropic reduction with respect to $\rho^{12}$. This intermediate theory will be the starting point for the symplectomorphism to EH theory, which will be invariant under the reduction with respect to $\rho^{12}$ and easier to write in this context than the one after the complete reduction, in which one is forced to work in components. We will anyway give a symplectomorphism also in the latter case.
    \end{remark}

    Let us conclude now with the reduction with respect to $\rho^{12}$. The first equations of the previous system become
    \begin{align*}
        (\dot{B}_{\nu}^1)''(t)&= \rho^{12}(B^{2}_{\nu})''(t) \\
        (\dot{B}_{\nu}^2)''(t)&= -\rho^{12}(B^{1}_{\nu})''(t) 
    \end{align*}
    where we used that $\eta_{11}=\eta_{22}=1$.The solution to this system is 
    \begin{align*}
        (B^{1}_{\nu})''(t) &= \cos{(\rho^{12}t)}(B^{1}_{\nu})''(0) + \sin{(\rho^{12}t)}(B^{2}_{\nu})''(0)\\
        (B^{2}_{\nu})''(t) &= -\sin{(\rho^{12}t)}(B^{1}_{\nu})''(0) + \cos{(\rho^{12}t)}(B^{2}_{\nu})''(0)
    \end{align*}
    The reduced variables will be set at the parameter $t=1$, while for $t=0$ we recover the old ones.
    Hence we can set $\rho^{12}= \arctan{\left(-\frac{(B^{1}_{2})''(0)}{(B^{2}_{2})''(0)}\right)}$ and get, for $t=1$,
    \begin{subequations}\label{e:gradedcoisored1}
    \begin{align}
        (B_2^1)'''&=0 & (B_2^2)'''&= \alpha_{12} \\
        (B_1^1)'''&= \frac{(B_1^1)'' (B_2^2)''-(B_2^1)'' (B_1^2)''}{\alpha_{12}} & (B_1^2)'''&= \frac{(B_1^1)'' (B_2^1)''+(B_2^2)'' (B_1^2)''}{\alpha_{12}}
    \end{align}
    where $\alpha_{12}=\sqrt{((B_2^1)'')^2+ ((B_2^2)'')^2}$. Proceeding similarly, for the other fields one obtains:
    \begin{align}
        (A_{\nu}^{01})'''&= \frac{(A_{\nu}^{01})'' (B_2^2)''-(B_2^1)'' (A_{\nu}^{02})''}{\alpha_{12}} & 
        ({B^\dag}_{\nu\mu}^{01})'''&= \frac{({B^\dag}_{\nu\mu}^{01})'' (B_2^2)''-(B_2^1)'' ({B^\dag}_{\nu\mu}^{02})''}{\alpha_{12}} \\ 
        (A_{\nu}^{02})'''&= \frac{(A_{\nu}^{02})'' (B_2^2)''+(B_2^1)'' (A_{\nu}^{01})''}{\alpha_{12}} & 
        ({B^\dag}_{\nu\mu}^{02})'''&= \frac{({B^\dag}_{\nu\mu}^{02})'' (B_2^2)''+(B_2^1)'' ({B^\dag}_{\nu\mu}^{01})''}{\alpha_{12}} \\ 
        (A_{\nu}^{12})'''&= (A_{\nu}^{12})'' - \frac{ (B_2^2)'' d_{\nu} (B_2^1)'' - (B_2^1)'' d_{\nu} (B_2^2)'' }{\alpha^2_{12}} 
          &({B^\dag}_{\nu\mu}^{12})'''&= ({B^\dag}_{\nu\mu}^{12})'' \\ 
          (\tau^{1})'''&= \frac{(\tau^{1})'' (B_2^2)''-(B_2^1)'' (\tau^{2})''}{\alpha_{12}} & (\tau^{2})'''&= \frac{(\tau^{2})'' (B_2^2)''+(B_2^1)'' (\tau^{1})''}{\alpha_{12}}\\
          (\tau^{0})'''&= (\tau^{0})'' 
    \end{align}\end{subequations}

    Hence, the fields on the reduced space $\C_{\mathrm{Gauss}}$ are $(\tau)'''$, $(B^{\dag})'''$ (all components), $(A)'''$ (only the components $(A_{1}^{01})'''$, $(A_{2}^{01}))'''$ and $(A_{2}^{02})'''$, the components $A_{1}^{02}$ and $A_{\nu}^{12}$ are fixed respectively by \eqref{e:ConstraintBF_antisym} and \eqref{e:ConstraintBF_LeviCivita}) and $(B))'''$ (only the components $(B_1^1)'''$, $(B_1^2)'''$ and $(B_2^2)'''$). On this space the symplectic form reads
    \begin{align}\label{e:sympl_form_reduced}
        \varpi_{\mathrm{res}} = \int_{\Sigma} \left(\delta (A_1^{01})''' \delta (B_2^2)'''+\delta (A_2^{01})''' \delta (B_1^2)'''+\delta (A_2^{02})''' \delta (B^1_1)'''\right)dx^1 dx^2 + \delta (\tau)''' \delta (B^\dag)'''
    \end{align}
    and the action is 
    \begin{align}\label{e:action_reduced}
        S_{\mathrm{res}} = \int_{\Sigma} (\tau)''' F_{(A)'''}.
    \end{align}

\section{Properties of \texorpdfstring{$\sigma_{\BFV}$}{sBFV}}
\subsection{Invariance of \texorpdfstring{$\sigma_{\BFV}$}{sBFV}} \label{a:independence}

We have to prove that the map 
        \begin{equation*}
        \sigma_{\BFV}\colon\begin{cases}
        g_{\mu\nu}= B_{\mu}^i B_{\nu}^j \eta_{ij}\\
        \Pi^{\mu\nu}= \sqrt{\mathsf{g}} (K_{\mu\nu}- g_{\mu\nu}g^{\rho\sigma}K_{\rho\sigma})-  \frac{1}{2} B_{\mu}^i B_{\nu}^j \left(\epsilon_{ik}\eta_{jl}+\epsilon_{jk}\eta_{il}\right)\tau^l  {B^{\dag}}^{0k}_{\rho\sigma}\epsilon^{\rho\sigma}\\
        \txi{n} =   \tau^0\\
        \txi{\mu} =  (B^{-1})^{\mu}_i\tau^i\\
        \phin = {B^{\dag}}^{ij}_{\rho\sigma}\epsilon^{\rho\sigma}\epsilon_{ij}\\
        \varphi_{\mu} = 2 {B^{\dag}}^{0i}_{\rho\sigma}\epsilon^{\rho\sigma}B_{\mu}^j\epsilon_{ij}.
        \end{cases}
    \end{equation*}
is invariant with respect to the projection built in Appendix \ref{a:coisotropic_BFtoEH} above (cf. Equations \eqref{e:gradedcoisored1}). Let us start with $g:$
\begin{align*}
    (g_{11})'''&= (B_{1}^i)''' (B_{1}^j)''' \eta_{ij} = (B_{1}^1)''' (B_{1}^1)''' + (B_{1}^2)''' (B_{1}^2)''' \\
    &= \left(\frac{(B_1^1)'' (B_2^2)''-(B_2^1)'' (B_1^2)''}{\alpha_{12}}\right)^2 + \left(\frac{(B_1^1)'' (B_2^1)''+(B_2^2)'' (B_1^2)''}{\alpha_{12}}\right)^2\\
    &= \frac{(B_1^1)''^2 (B_2^2)''^2 
        + (B_2^1)''^2 (B_1^2)''^2 
        - 2 (B_1^1)'' (B_2^2)''(B_2^1)'' (B_1^2)''}{\alpha_{12}^2}\\
    &+ \frac{(B_1^1)''^2 (B_2^1)''^2 
        +(B_2^2)''^2 (B_1^2)''^2
        +2(B_1^1)'' (B_2^1)''(B_2^2)'' (B_1^2)''}{\alpha_{12}^2}\\
    &= \frac{(B_1^1)''^2 ((B_2^2)''^2+(B_2^1)''^2) +  (B_1^2)''^2((B_2^1)''^2+ (B_2^2)''^2)}{\alpha_{12}^2}\\
    &= (B_1^1)''^2 + (B_1^2)''^2 = (g_{11})''\\
    (g_{12})'''&= (B_{1}^i)''' (B_{2}^j)''' \eta_{ij} = (B_{1}^1)''' (B_{2}^1)''' + (B_{1}^2)''' (B_{2}^2)'''\\
    &= 0 + \frac{(B_1^1)'' (B_2^1)''+(B_2^2)'' (B_1^2)''}{\alpha_{12}}\alpha_{12} = (B_1^1)'' (B_2^1)''+(B_2^2)'' (B_1^2)'' = (g_{12})''\\
    (g_{22})'''&= (B_{1}^i)''' (B_{2}^j)''' \eta_{ij} = (B_{2}^1)''' (B_{2}^1)''' + (B_{2}^2)''' (B_{2}^2)'''\\
    &= 0 +  \alpha_{12}^2 =(B_{2}^1)''^2+ (B_{2}^2)''^2 = (g_{22})''.
\end{align*}
In order to prove a similar statement for $\txi{}$ we need the following identies:
\begin{align*}
    ((B^{-1})_1^1)'''&= \frac{1}{(B_1^1)'''} & ((B^{-1})_2^2)'''&=  \frac{1}{(B_2^2)'''} & ((B^{-1})_1^2)'''&= -\frac{(B_1^2)'''}{(B_2^2)'''(B_1^1)'''}
\end{align*}
and 
\begin{align*}
    ((B^{-1})_1^1)'' & = \frac{(B_2^2)''}{\det B''} & ((B^{-1})_1^2)'' & = - \frac{(B_1^2)''}{\det B''} \\
    ((B^{-1})_2^1)'' & = -\frac{(B_2^1)''}{\det B''} & ((B^{-1})_2^2)'' & =  \frac{(B_1^1)''}{\det B''} 
\end{align*}
and note that $(B_1^1)'''= \frac{\det B''}{\alpha_{12}}$.

Consider now $\txi{\mu}$. We have
\begin{align*}
    (\txi{1})'''&= ((B^{-1})_1^1)''' (\tau^1)''' +  ((B^{-1})_2^1)''' (\tau^2)''' = \frac{\alpha_{12}}{\det B''}\frac{(\tau^{1})'' (B_2^2)''-(B_2^1)'' (\tau^{2})''}{\alpha_{12}} =  (\txi{1})''\\
    (\txi{2})'''&= ((B^{-1})_1^2)''' (\tau^1)''' +  ((B^{-1})_2^2)''' (\tau^2)'''\\
    &=  -\frac{(B_1^2)'''}{(B_2^2)'''(B_1^1)'''} (\tau^1)''' + \frac{1}{(B_2^2)'''}(\tau^2)'''\\
    &= -\frac{(B_1^1)'' (B_2^1)''+(B_2^2)'' (B_1^2)''}{\alpha_{12}\det B''}\frac{(\tau^{1})'' (B_2^2)''-(B_2^1)'' (\tau^{2})''}{\alpha_{12}} + \frac{(\tau^{2})'' (B_2^2)''+(B_2^1)'' (\tau^{1})''}{\alpha_{12}^2}\\
    &= \frac{-(\tau^{1})''(B_1^2)''+ (\tau^{2})''(B_1^1)''}{\det B''} = (\txi{2})''
\end{align*}

Similarly we can prove that the other components of the symplectomorphism are invariant. 

\subsection{\texorpdfstring{$\underline{\sigma_{\BFV}}$}{usBFV} is a symplectomorphism}\label{a:symplectomorphism}
The proof of this and the following section make a repeated use of the following relations between the Levi-Civita symbol $\epsilon$ and $\eta$ valid for indexes $i,j,k,l=1,2$ (and equivalent formulations):
    \begin{align}\label{e:relations_epsilon_eta}
        \epsilon^{ij}\epsilon_{ij} &= 2, & 
        \epsilon^{i j}\epsilon_{i k} &= \delta^{j}_{k}, & 
        \epsilon_{i j}\epsilon_{k l} & = \eta_{ik}\eta_{jl} - \eta_{il}\eta_{jk}.
    \end{align}
Before proceeding, we need to introduce some notation: we define
    \begin{align}\label{e:Def_Bline}
        B^{\dag {lj}}_{-} &= B^{\dag {lj}}_{\rho\sigma} \epsilon^{\rho\sigma}\\ \label{e:Def_Dx}
        \mathrm{vol} &= \frac{1}{2}\epsilon_{\rho\sigma}dx^{\rho}dx^{\sigma}
    \end{align}
    and note that 
    \begin{align}\label{e:Prop_Bline}
        B^{\dag {lj}}_{\mu\nu} &= \frac{1}{2} B^{\dag {lj}}_{-} \epsilon_{\mu\nu}\\ \label{e:Prop_Dx}
        dx^\mu dx^\nu & = \frac{1}{2}\epsilon^{\mu\nu}\epsilon_{\rho\sigma}dx^{\rho}dx^{\sigma} = \epsilon^{\mu\nu} \mathrm{vol}
    \end{align}

We want to prove that $\underline{\sigma_{\BFV}}^{*}(\varpi_{EH})= \varpi_{\mathrm{int}}$.  Since these symplectic form are exact, we equivalently prove the claim for the corresponding 1-forms, which read 
    \begin{align*}
        \alpha_{\mathrm{int}} &= \int_{\Sigma} \left(2  A_{\mu}^{0i} \delta B_{\nu}^{j}  +2 \delta \tau^i  B^{\dag 0j}_{\mu\nu} + \delta \tau^0  B^{\dag ij}_{\mu\nu}  \right) \epsilon_{ij}dx^{\mu}dx^{\nu}\\
        &= \int_{\Sigma} \left(2  A_{\mu}^{0i} \delta B_{\nu}^{j} \epsilon^{\mu\nu} +2 \delta \tau^i  B^{\dag 0j}_{-} + \delta \tau^0  B^{\dag ij}_{-}  \right) \epsilon_{ij}\mathrm{vol}\\
        \alpha_{EH} &=\int_{\Sigma} \left( \delta {g}^{\mu\nu} {\Pi}_{\mu\nu} + \delta\txi{\rho} {\varphi}_{\rho}+ \delta\txi{n}{\varphi}_{n}\right) \mathrm{vol}
\end{align*}
where in the second step we used \eqref{e:Prop_Bline} and \eqref{e:Prop_Dx}.

We first prove the following preliminary lemma.
\begin{lemma}
    Define 
    \begin{align} \label{e:defPi_cl}
        \Pi^{\mathrm{cl}}_{\mu\nu} &= \sqrt{\mathsf{g}} (K_{\mu\nu}- g_{\mu\nu}g^{\rho\sigma}K_{\rho\sigma})\\ \label{e:defPi_BV}
        \Pi^{\mathrm{BV}}_{\mu\nu} &= -  \frac{1}{2} B_{\mu}^i B_{\nu}^j \left(\epsilon_{ik}\eta_{jl}+\epsilon_{jk}\eta_{il}\right)\tau^l  {B^{\dag}}^{0k}_{-}
    \end{align}
    such that $\underline{\sigma_{\BFV}}^{*}(\Pi_{\mu\nu})=\Pi^{\mathrm{cl}}_{\mu\nu}+\Pi^{\mathrm{BV}}_{\mu\nu}$. Then
    \begin{align}\label{e:identity_for_Pi_cl1}
    \Pi^{\mathrm{cl}}_{\mu\nu} &= \frac{1}{2}g_{\theta\mu}\epsilon^{\sigma\theta}A^{0c}_{\sigma}\epsilon_{bc}B_{\nu}^b + \frac{1}{2}g_{\theta\nu}\epsilon^{\sigma\theta}A^{0c}_{\sigma}\epsilon_{bc}B_{\mu}^b\\\label{e:identity_for_Pi_cl2}
    2\Pi^{\mathrm{cl}}_{\mu\nu}(B^{-1})^{\mu}_i \delta (B^{-1})^{\nu}_j \eta^{ij} 
    &= 2 \epsilon^{\sigma\nu}A_{\sigma}^{0c}\epsilon_{bc} \delta B_{\nu}^b- \tau^k {B^{\dag}}^{0c}_{\sigma\theta}\eta_{ck}\epsilon^{\sigma\theta} B_{\nu}^b \delta (B^{-1})^{\nu}_j \epsilon^j_{\;b}\\\label{e:identity_for_Pi_BV}
    2\Pi^{\mathrm{BV}}_{\mu\nu}(B^{-1})^{\mu}_i \delta (B^{-1})^{\nu}_j \eta^{ij} 
    &=  - 2  (B^{-1})^{\rho}_i  \tau^i  {B^{\dag}}^{0k}_{-}\delta B_{\rho}^j\epsilon_{kj} +\tau^k {B^{\dag}}^{0c}_{-}\eta_{ck} B_{\nu}^b \delta (B^{-1})^{\nu}_j \epsilon^j_{\;b}.
    \end{align}
\end{lemma}

\begin{proof}
    In order to prove \eqref{e:identity_for_Pi_cl1}, we substitute the definition of $K_{\mu\nu}$ (given in \eqref{e:def_K_EH}) and obtain
    \begin{align*}
        \Pi^{\mathrm{cl}}_{\mu\nu} &= \sqrt{\mathsf{g}}\left(\frac{1}{2} B_{\mu}^a A_{\nu}^{0b} \eta_{ab} + \frac{1}{2} B_{\nu}^a A_{\mu}^{0b} \eta_{ab} - B_{\mu}^a B_{\nu}^b \eta_{ab} A_{\rho}^{0c} (B^{-1})^{\rho}_c \right)\\
        &=\sqrt{\mathsf{g}}\left(\frac{1}{2} B_{\mu}^a A_{\nu}^{0b} \eta_{ab} + \frac{1}{2} B_{\nu}^a A_{\mu}^{0b} \eta_{ab} - B_{\mu}^a B_{\nu}^b \eta_{ab} A_{\rho}^{0c} (B^{-1})^{\rho}_d \delta_c^d \right).
    \end{align*}
    where we used that $K= g^{\mu\nu}K_{\mu\nu}= A_{\rho}^{0c} (B^{-1})^{\rho}_c$.
    We now use on the last term both the following identities (derived from \eqref{e:relations_epsilon_eta}):
    \begin{align*}
        \frac{1}{2}\eta_{ab}\delta_c^d &=  \frac{1}{2}\eta_{bc}\delta_a^d + \frac{1}{2}\epsilon_{ac}\epsilon_b^{\;d}, & 
        \frac{1}{2}\eta_{ab}\delta_c^d &=  \frac{1}{2}\eta_{ac}\delta_b^d + \frac{1}{2}\epsilon_{bc}\epsilon_a^{\;d}
    \end{align*}
    and with some simple algebra we get
        \begin{align*}
        \Pi^{\mathrm{cl}}_{\mu\nu} &= \sqrt{\mathsf{g}}\left( -\frac{1}{2} B_{\mu}^a B_{\nu}^b \epsilon_{ac}\epsilon_b^{\;d} A_{\rho}^{0c} (B^{-1})^{\rho}_d -\frac{1}{2} B_{\mu}^a B_{\nu}^b \epsilon_{bc}\epsilon_a^{\;d} A_{\rho}^{0c} (B^{-1})^{\rho}_d \right)\\
        &= -\frac{1}{4} B_{\sigma}^i B_{\theta}^j \epsilon_{ij}\epsilon^{\sigma \theta}\left( B_{\mu}^a B_{\nu}^b \epsilon_{ac}\epsilon_b^{\;d} A_{\rho}^{0c} (B^{-1})^{\rho}_d + B_{\mu}^a B_{\nu}^b \epsilon_{bc}\epsilon_a^{\;d} A_{\rho}^{0c} (B^{-1})^{\rho}_d \right).
    \end{align*}
    Subsequently we use again some identities derived from \eqref{e:relations_epsilon_eta}:
    \begin{align*}
        \epsilon_{ij}\epsilon_b^{\;d} &= \eta_{ib} \delta_j^d - \eta_{jb} \delta_i^d, & \epsilon_{ij}\epsilon_a^{\;d} &= \eta_{ia} \delta_j^d - \eta_{ja} \delta_i^d,
    \end{align*}
    and with some straightforward manipulations we obtain \eqref{e:identity_for_Pi_cl1}.
    Using this expression we subsequently obtain
    \begin{align*}
        2\Pi^{\mathrm{cl}}_{\mu\nu}(B^{-1})^{\mu}_i \delta (B^{-1})^{\nu}_j \eta^{ij} &= B_{\theta}^j \epsilon^{\sigma\theta} A_{\sigma}^{0c} \epsilon_{bc} B_{\nu}^b \delta (B^{-1})^{\nu}_j + B_{\theta}^k B_{\nu}^l \eta_{kl} \epsilon^{\sigma\theta} A_{\sigma}^{0c} \epsilon_{ic} \delta (B^{-1})^{\nu}_j \eta^{ij}\\
        &= B_{\theta}^k \epsilon^{\sigma\theta} A_{\sigma}^{0c}  B_{\nu}^b \delta (B^{-1})^{\nu}_j \left( \epsilon_{bc} \delta_k^j + \epsilon_{\;c}^j \eta_{kb}\right)
    \end{align*}
    and using again \eqref{e:relations_epsilon_eta} in the form 
    $\epsilon_{\;c}^j \eta_{kb}=  \epsilon_{bc} \delta_k^j + \epsilon_{\;b}^j \eta_{ck} $ we obtain
    \begin{align*}
    2\Pi^{\mathrm{cl}}_{\mu\nu}(B^{-1})^{\mu}_i \delta (B^{-1})^{\nu}_j \eta^{ij} &= 2 \epsilon^{\sigma\nu}A_{\sigma}^{0c}\epsilon_{bc} \delta B_{\nu}^b + B_{\theta}^k \epsilon^{\sigma\theta} \eta_{kc}A_{\sigma}^{0c} B_{\nu}^b \delta (B^{-1})^{\nu}_j \epsilon^j_{\;b}\\
    &= 2 \epsilon^{\sigma\nu}A_{\sigma}^{0c}\epsilon_{bc} \delta B_{\nu}^b- \tau^k {B^{\dag}}^{0c}_{\sigma\theta}\eta_{ck}\epsilon^{\sigma\theta} B_{\nu}^b \delta (B^{-1})^{\nu}_j \epsilon^j_{\;b}
\end{align*}
where in the last step we used the residual constraint \eqref{e:ConstraintBF_antisym}. Similarly, we have 
\begin{align*}
    2\Pi^{\mathrm{BV}}_{\mu\nu}(B^{-1})^{\mu}_i \delta (B^{-1})^{\nu}_j \eta^{ij} 
    &=  -  (B^{-1})^{\rho}_i  \tau^l  {B^{\dag}}^{0k}_{-}\delta B_{\rho}^j\left(\epsilon_{kj}\delta_l^i + \epsilon_{k}^{\;i}\eta_{lj}\right)
\end{align*}
and using $\epsilon_{k}^{\;i}\eta_{lj} = \epsilon_{kj}\delta_l^i - \epsilon_{\;j}^{i}\eta_{lk}$ we obtain \eqref{e:identity_for_Pi_BV}.
\end{proof}

We can now pull back $\alpha_{EH}$. We have
\begin{align*}
    \underline{\sigma_{\BFV}}^{*}(\delta\txi{n}{\varphi}_{n}) &= \delta\tau^0 {B^{\dag}}^{ij}_{-}\epsilon_{ij} \\ 
    \underline{\sigma_{\BFV}}^{*}(\delta\txi{\rho} {\varphi}_{\rho}) & = 2\delta ( (B^{-1})^{\rho}_i\tau^i)  {B^{\dag}}^{0k}_{-}B_{\rho}^j\epsilon_{kj}\\
    & = 2 (B^{-1})^{\rho}_i \delta \tau^i  {B^{\dag}}^{0k}_{-}B_{\rho}^j\epsilon_{kj} + 2\delta  (B^{-1})^{\rho}_i  \tau^i  {B^{\dag}}^{0k}_{-}B_{\rho}^j\epsilon_{kj}\\
    & = 2 \delta \tau^i  {B^{\dag}}^{0k}_{-}\epsilon_{ki} + 2  (B^{-1})^{\rho}_i  \tau^i  {B^{\dag}}^{0k}_{-}\delta B_{\rho}^j\epsilon_{kj}
\end{align*}
and using \eqref{e:identity_for_Pi_cl2} and \eqref{e:identity_for_Pi_BV}, 
\begin{align*}
    \underline{\sigma_{\BFV}}^{*}(\Pi_{\mu\nu}\delta g^{\mu\nu}) 
    &= 2 \epsilon^{\sigma\nu}A_{\sigma}^{0c}\epsilon_{bc} \delta B_{\nu}^b  - 2  (B^{-1})^{\rho}_i  \tau^i  {B^{\dag}}^{0k}_{-}\delta B_{\rho}^j\epsilon_{kj}.
\end{align*}

Hence 
\begin{align}\label{e:Part_Of_pushforward_Symplform}
    \underline{\sigma_{\BFV}}^{*}(\delta\txi{\rho} {\varphi}_{\rho} + \Pi_{\mu\nu}\delta g^{\mu\nu}) = 2 \delta \tau^i  {B^{\dag}}^{0k}_{-}\epsilon_{ki} + 2 \epsilon^{\sigma\nu}A_{\sigma}^{0c}\epsilon_{bc} \delta B_{\nu}^b
\end{align}
and collecting all the terms we get the desired claim.

\subsection{\texorpdfstring{$\underline{\sigma_{\BFV}}$}{usBFV} preserves the actions}\label{a:actions}
The proof of this statement is extraordinarily long, but composed of mostly straightforward calculations. Since the hard part is is about matching the various pieces of the expressions coming from $\underline{\sigma_{\BFV}}^*S_{EH}$ with $S_{\mathrm{int}}$ we will slit the proof in multiple lemmas, showing exactly what pieces need to match. Let us begin with a preliminary result:

\begin{lemma}
The action $S_{\mathrm{int}}$ can be expressed as
    \begin{align} \label{e:S_BF_reduced_final}
    S_{\mathrm{int}} &= \int_{\Sigma} \left( d_{\mu} \tau^0  \Gamma_{\nu}^{ij} \epsilon_{ij}\epsilon^{\mu\nu} +  d_{\rho} \tau^0 \left( \tau^m B^{\dag {lj}}_{-}\eta_{jm}+2\tau^0 B^{\dag {l0}}_{-}\right)(B^{-1})^\rho_k \epsilon_{l}^{\;k}\right)\mathrm{vol} \\
    &\phantom{= \int_{\Sigma}} +\left(\tau^0 A_{\mu}^{0i}A_{\nu}^{0j} + 2 \tau^{i} d_{\mu}A_{\nu}^{0j} +2 \tau^{i} \Gamma_{\mu}^{jk}A_{\nu}^{0l}\eta_{kl} \epsilon_{ij} + 4 \tau^{i}\tau^0 B^{\dag {l0}}_{\mu\rho}(B^{-1})^\rho_m \epsilon_{l}^{\;m}A_{\nu}^{0n}\eta_{in}\right) \epsilon^{\mu\nu}\mathrm{vol}. \nonumber
\end{align}
\end{lemma}

\begin{proof}
    The result is obtained by unpacking the action $S_{\mathrm{int}}= \trintl{\Sigma} \tau F_A$:
\begin{align}
    S_{\mathrm{int}} & = \int_{\Sigma} \left( \tau^0 (F_A)_{\mu\nu}^{ij} + 2 \tau^{i} (F_A)_{\mu\nu}^{j0}\right)\epsilon_{ij} dx^{\mu}dx^{\nu} \nonumber\\
    &= \int_{\Sigma} \left( \tau^0 d_{\mu} A_{\nu}^{ij} + \tau^0 A_{\mu}^{ik}A_{\nu}^{jl}\eta_{kl} + \tau^0 A_{\mu}^{0i}A_{\nu}^{0j} + 2 \tau^{i} d_{\mu}A_{\nu}^{0j} + 2 \tau^{i} A_{\mu}^{jk}A_{\nu}^{0l}\eta_{kl}\right)\epsilon_{ij} \epsilon^{\mu\nu}\mathrm{vol} \label{e:S_BF_reduced}
\end{align}
where $A_{\mu}^{jk}$ is not an independent field but it is defined as 
    \begin{align*}
        A_{\mu}^{ij} = \Gamma_{\mu}^{ij} + 2 X_{\mu\nu}^l(B^{-1})^\nu_k \epsilon^{ij}\epsilon_{l}^{\;k}
    \end{align*}
    where $\Gamma_{\mu}^{ij}$ is the spin connection associated to the diad $B_\mu^i$ 
    \begin{align*}
        \Gamma_{\mu }^{ij}= (B^{-1})^{\nu}_a d_{\lambda}B_{\theta}^{k}\epsilon^{\lambda\theta}\epsilon_{\mu\nu}\left(\eta^{ai}\delta_k^j- \eta^{aj} \delta_k^i\right)
    \end{align*}
    and
    \begin{align*}
        X_{\mu\nu}^i = \tau^k B^{\dag {ij}}_{\mu\nu}\eta_{jk}+\tau^0 B^{\dag {i0}}_{\mu\nu}.
    \end{align*}

    The first term of \eqref{e:S_BF_reduced} can be rewritten as follows:
    \begin{align*}
        \int_{\Sigma}\tau^0 d_{\mu} A_{\nu}^{ij} \epsilon_{ij} \epsilon^{\mu\nu}\mathrm{vol}&= \int_{\Sigma}\left(\tau^0 d_{\mu} \Gamma_{\nu}^{ij} \epsilon_{ij} + 2 \tau^0 d_{\mu} \left(  X_{\nu\rho}^l(B^{-1})^\rho_k \epsilon^{ij}\epsilon_{l}^{\;k} \epsilon_{ij} \right)\right) \epsilon^{\mu\nu}\mathrm{vol}\\
        &= \int_{\Sigma} \left(d_{\mu} \tau^0  \Gamma_{\nu}^{ij} \epsilon_{ij} + 4 d_{\mu} \tau^0   X_{\nu\rho}^l(B^{-1})^\rho_k \epsilon_{l}^{\;k} \right) \epsilon^{\mu\nu}\mathrm{vol}\\
        &= \int_{\Sigma} \left( d_{\mu} \tau^0  \Gamma_{\nu}^{ij} \epsilon_{ij} + 2 d_{\mu} \tau^0 \left( \tau^m B^{\dag {lj}}_{\nu\rho}\eta_{jm}+2\tau^0 B^{\dag {l0}}_{\nu\rho}\right)(B^{-1})^\rho_k \epsilon_{l}^{\;k}\right) \epsilon^{\mu\nu}\mathrm{vol}\\
        & = \int_{\Sigma} d_{\mu} \tau^0  \Gamma_{\nu}^{ij} \epsilon_{ij} \epsilon^{\mu\nu}\mathrm{vol} +  d_{\rho} \tau^0 \left( \tau^m B^{\dag {lj}}_{-}\eta_{jm}+2\tau^0 B^{\dag {l0}}_{-}\right)(B^{-1})^\rho_k \epsilon_{l}^{\;k}\mathrm{vol}\\
    \end{align*}
    where we integrated by parts, and used \eqref{e:relations_epsilon_eta} and \eqref{e:Prop_Bline}. 
    Analogously, we can substitute the solution of $A$ in the last term of \eqref{e:S_BF_reduced}:
    \begin{align*}
        \int_{\Sigma} 2 \tau^{i} A_{\mu}^{jk}A_{\nu}^{0l}\eta_{kl} \epsilon_{ij} \epsilon^{\mu\nu}\mathrm{vol} &=  \int_{\Sigma} \left( 2 \tau^{i} \Gamma_{\mu}^{jk}A_{\nu}^{0l}\eta_{kl} \epsilon_{ij}  +  4 \tau^{i} X_{\mu\rho}^l(B^{-1})^\rho_m \epsilon^{jk}\epsilon_{l}^{\;m}A_{\nu}^{0n}\eta_{kn} \epsilon_{ij} \right)\epsilon^{\mu\nu}\mathrm{vol}\\
        &=  \int_{\Sigma}\left( 2 \tau^{i} \Gamma_{\mu}^{jk}A_{\nu}^{0l}\eta_{kl} \epsilon_{ij} + 4 \tau^{i} X_{\mu\rho}^l(B^{-1})^\rho_m \epsilon_{l}^{\;m}A_{\nu}^{0n}\eta_{in}  \right)\epsilon^{\mu\nu}\mathrm{vol}\\
        &=  \int_{\Sigma} \left(2 \tau^{i} \Gamma_{\mu}^{jk}A_{\nu}^{0l}\eta_{kl} \epsilon_{ij} + 4 \tau^{i} \tau^k B^{\dag {lj}}_{\mu\rho}\eta_{jk}(B^{-1})^\rho_m \epsilon_{l}^{\;m}A_{\nu}^{0n}\eta_{in}\right)  \epsilon^{\mu\nu}\mathrm{vol}\\
        &\phantom{=}
        + \int_{\Sigma} 4 \tau^{i}\tau^0 B^{\dag {l0}}_{\mu\rho}(B^{-1})^\rho_m \epsilon_{l}^{\;m}A_{\nu}^{0n}\eta_{in}  \epsilon^{\mu\nu}\mathrm{vol}.
    \end{align*}
    The second term of this expression vanishes as it is proportional to $\tau^i\tau^j\tau^k=0$, after using the residual constraint \eqref{e:ConstraintBF_antisym}.
    Note as well that the second term of \eqref{e:S_BF_reduced} vanishes due to antisymmetry\footnote{ This is specific to the dimension of the boundary $\dim \Sigma=2$: $\epsilon^{ik}\epsilon^{jl}\eta_{kl}\epsilon_{ij}=\epsilon^{kl}\eta_{kl}=0 $.}. Collecting all the computations we get the desired claim.
\end{proof}

Let us move to $\underline{\sigma_{\BFV}}^* (S_{EH})$ and compute separately some of the terms.

\begin{remark}
    The proof of the following calculations is long but straightforward. This should be taken as a guide for the reader interested in checking the statement, which will indicate what terms should match.
\end{remark}

\begin{calculation}
    \begin{align*}
    \underline{\sigma_{\BFV}}^* \left(\int_{\Sigma} \sqrt{\mathsf{g}}R^\Sigma  \txi{n} \mathrm{vol}\right) 
    = \int_{\Sigma}  2 \partial_{\mu}\tau^0 (B^{-1})^{\nu}_i \eta^{ij} \partial_{\lambda} B_{\theta}^k \epsilon^{\lambda \theta} \epsilon_{jk} \mathrm{vol}
\end{align*}
\end{calculation}

\begin{calculation}
    \begin{align*}
    \underline{\sigma_{\BFV}}^{*}\left(\int_{\Sigma}2 \txi{\mu} \partial_{\rho}\left( \Pi_{\mu\nu}g^{\rho\nu}\right)\right) \mathrm{vol} &= \int_{\Sigma} 2 \partial_{\rho} \tau^k \left((B^{-1})^{\rho}_j \tau^j \epsilon_{kl} {B^{\dag}}^{0l}_{-} + \epsilon^{\sigma \rho} A_{\sigma}^{0i} \epsilon_{ki}\right)\mathrm{vol} \\
    & \phantom{=} + \int_{\Sigma} 2  \tau^i \partial_{\rho}(B^{-1})^{\mu}_i B_{\mu}^k \left((B^{-1})^{\rho}_j \tau^j \epsilon_{kl} {B^{\dag}}^{0l}_{-} + \epsilon^{\sigma \rho} A_{\sigma}^{0i} \epsilon_{ki}\right)\mathrm{vol}
\end{align*}
\end{calculation}

\begin{calculation}
\begin{align*}
    \underline{\sigma_{\BFV}}^{*}\left(\int_{\Sigma} \frac{1}{\sqrt{\mathsf{g}}} \left(\Pi^{\mu\nu}\Pi_{\mu\nu}- \Pi^2\right) \txi{n}\right) \mathrm{vol} &= \int_{\Sigma} \epsilon^{\mu\nu}\epsilon_{ij} A_{\mu}^{0i}A_{\nu}^{0j} \tau^0 \mathrm{vol} \\
    & \phantom{=} + \int_{\Sigma} 2 \tau^i \eta_{ij} \tau^0 {B^{\dag}}^{0m}_{-} (B^{-1})^{\mu}_l \epsilon^l_{\; m} A_{\mu}^{0j} \mathrm{vol}.
\end{align*}
\end{calculation}

Using the explicit expresssion of $\sigma_{\BFV}$ we also get
\begin{align*}
    \underline{\sigma_{\BFV}}^* \left(\int_{\Sigma}  \partial_\mu\left(\txi{\mu}{\phin}\right) \txi{n} \mathrm{vol}\right) 
    &= \int_{\Sigma} \partial_{\nu}\tau^0 (B^{-1})^{\nu}_j \tau^j {B^{\dag}}^{ik}_{-} \epsilon_{ik} \mathrm{vol}
\end{align*}
\begin{align*}
    \underline{\sigma_{\BFV}}^* \left(\int_{\Sigma}  {g}^{\mu\nu}{\varphi}_\mu\partial_\nu\txi{n} \txi{n}\mathrm{vol}\right)
    &= \int_{\Sigma} 2 \partial_{\nu}\tau^0 (B^{-1})^{\nu}_j \eta^{lj} \epsilon_{kl} {B^{\dag}}^{0k}_{-} \tau^0 \mathrm{vol}
\end{align*}
Finally, using \eqref{e:Part_Of_pushforward_Symplform} and formally substituting $\delta$ with $\partial_{\rho}$, we also get
\begin{align*}
    \underline{\sigma_{\BFV}}^{*}\left(\int_{\Sigma}\txi{\rho} \partial_{\rho}\txi{\mu} {\varphi}_{\mu} + \txi{\rho} \Pi_{\mu\nu}\partial_{\rho} g^{\mu\nu}\right) \mathrm{vol}
    & = \int_{\Sigma} 2 (B^{-1})^{\rho}_j \tau^j \partial_{\rho} \tau^i  {B^{\dag}}^{0k}_{-}\epsilon_{ki}\mathrm{vol}\\
    & \phantom{= \int_{\Sigma}}+ 2  (B^{-1})^{\rho}_j \tau^j \epsilon^{\sigma\nu}A_{\sigma}^{0c}\epsilon_{bc} \partial_{\rho} B_{\nu}^b \mathrm{vol}.
\end{align*}

It is now easy to compare the expression for $\underline{\sigma_{\BFV}}^* (S_{EH})$ (collecting all the previous contributions) and $S_{\mathrm{int}}$ given by \eqref{e:S_BF_reduced_final}. The proof is then concluded using the following lemma:

\begin{calculation}
    \begin{align*}
    \int_{\Sigma} 2 \tau^{i} \Gamma_{\mu}^{jk}A_{\nu}^{0l}\eta_{kl} \epsilon_{ij} \epsilon^{\mu\nu}\mathrm{vol} &= \int_{\Sigma} 2  \tau^i \partial_{\rho}(B^{-1})^{\mu}_i B_{\mu}^k \left((B^{-1})^{\rho}_j \tau^j \epsilon_{kl} {B^{\dag}}^{0l}_{-} + \epsilon^{\sigma \rho} A_{\sigma}^{0i} \epsilon_{ki}\right)\mathrm{vol}\\
    & \phantom{=}+ \int_{\Sigma} 2 (B^{-1})^{\rho}_j \tau^j \partial_{\rho} \tau^i  {B^{\dag}}^{0k}_{-}\epsilon_{ki}\mathrm{vol}.
\end{align*}
\end{calculation}

\section{Tables of symbols}

For each symbol we provide a (short) description and the corresponding place of definition.

 \begin{table}[H]
     \centering
     \renewcommand{\arraystretch}{1.2}
     \begin{tabular}{|c|c|c|}
     \hline
        $(F,\omega)$  & Graded symplectic manifold  & Section \ref{s:coiso_by_BFV}\\
         $C$ & Coisotropic submanifold of $(F,\omega)$ & Section \ref{s:coiso_by_BFV}\\
         $\underline{C}$ & Coisotropic reduction of $(F,\omega)$ w.r.t. $C$ & Section \ref{s:coiso_by_BFV}\\
         $\pi_C$ & projection $C \rightarrow \underline{C}$ & Section \ref{s:coiso_by_BFV}\\
         $\mathcal{I}_C$ & vanishing ideal of $C$ & Section \ref{s:coiso_by_BFV}\\
         $\{\phi_i\}_i$ & Set of first class Hamiltonian functions generating $\mathcal{I}_C$  & Section \ref{s:coiso_by_BFV}\\
        \hline
     \end{tabular}
     \vspace{3mm}
     \caption{Coisotropic submanifolds. }
     \label{tab:coisotropic}
 \end{table}
The same notation is used throughout the article with subscript denoting the specific coisotropic submanifold considered. If the (graded) symplectic manifold is infinite dimensional we use calligraphic letters, $(\mathcal{F}, \varpi)$, $\mathcal{C}$.

 \begin{table}[H]
     \centering
    \renewcommand{\arraystretch}{1.2}
     \begin{tabular}{|c|c|c|}
     \hline
            $\BFV(F,C)$ & BFV data associated to $C \subset F$ & Definition \ref{d:BFVtheory_of_C}\\
         $(\mathbb{F}(F,C), \omega_C)$ & BFV space of fields of $\BFV(F,C)$ & Definition \ref{d:BFVtheory_of_C}\\
         $\varpi_C$ & Alternative notation for $\omega_C$ in applications to gravity &\\
         $S_C$ & Action of $\BFV(F,C)$ & Definition \ref{d:BFVtheory_of_C}\\
         $Q_C$ & Cohomological vector field of  $\BFV(F,C)$ & Definition \ref{d:BFVtheory_of_C}\\
         $\mathfrak{BFV}^{\bullet}(F,C)$ & BFV complex resolving $\underline{C}$ & \\
         $W_{\phi}$ & Space of ghost fields associated to $\phi_i$ &\\
         \hline
     \end{tabular}
     \vspace{3mm}
     \caption{BFV data associated to the resolution of $C \subset F$}
     \label{tab:general_BFV}
 \end{table}

\begin{table}[H]
    \centering
    \renewcommand{\arraystretch}{1.2}
     \begin{tabular}{|c|c|c|}
     \hline
        $\Co \equiv C_{\phi}$ & Coisotropic submanifold defined by $\phi_i=0$ &\\
        $C_{\psi}$ & Coisotropic submanifold defined by $\psi_j=0$ &\\
        $C$ & $C_{\psi} \cap C_{\phi}$ &\\
        $\Cres$ & Coisotropic submanifold of $\underline{C_{\phi}}$ such that $\underline{\Cres}= \underline{C}$ &\\
        $f,h,m,g$ & Structure functions of Poisson algebra of constraints $\phi_i$ and $\psi_j$ & \eqref{e:Coisostructure}\\
        \hline
    \end{tabular}
    \vspace{3mm}
    \caption{Notation for nested coisotropic embeddings, general theory}
    \label{tab:nested_general}
\end{table}

\begin{table}[H]
    \centering
    \renewcommand{\arraystretch}{1.2}
     \begin{tabular}{|c|c|c|}
     \hline
        $\chi, \chi^\dag$ & Ghost and antighost relative to constraint $\phi_i$ & Assumption \ref{ass:BFV_action}\\
        $\lambda, \lambda^\dag$ & Ghost and antighost relative to constraint $\psi_i$ & Assumption \ref{ass:BFV_action}\\
        $\mathsf{M}^{\mu}, \mathsf{L}^{\rho}, \mathsf{J}^{\rho}$ & Constraints on $\BFV(F,C)$  & Theorem \ref{thm:BFVcoisotropic}\\
        $\mu, \mu^\dag, \rho, \rho^\dag$ & Ghosts and antighosts relative to constraints $\mathsf{M}^{\mu}, \mathsf{L}^{\rho}, \mathsf{J}^{\rho}$ & Theorem \ref{thm:BFVcoisotropic}\\
        $\mathbf{\Phi}, \mathbf{\Phi}^{\dag}$ & Auxiliary notation & Theorem \ref{thm:BFVcoisotropic} \\
        $\llbracket \cdot, \cdot \rrbracket$ & brackets generated by $f$ & Theorem \ref{thm:BFVcoisotropic}\\
        $m(\cdot, \cdot), g(\cdot, \cdot), h(\cdot, \cdot) $ & brackets generated by $m,g,h$ &Theorem \ref{thm:BFVcoisotropic} \\
        $\CCo$ & Coisotropic submanifold generated by $\mathsf{M}^{\mu}, \mathsf{L}^{\rho}, \mathsf{J}^{\rho}$ & Theorem \ref{thm:BFVcoisotropic}\\
        $\FF, \bbomega, \mathbb{S}, \mathbb{Q}$ & BFV theory resolving $\CCo$ & Definition \ref{d:doubleBFVdata}\\
        $\check{S}_C$ & Extension of $S_C$ to $\FF$ &  \eqref{e:generalSC}\\
        \hline
    \end{tabular}
    \vspace{3mm}
    \caption{Double BFV}
    \label{tab:doubleBFV}
\end{table}

\begin{table}[H]
    \centering
    \renewcommand{\arraystretch}{1.2}
     \begin{tabular}{|c|c|c|}
     \hline
        $\mathfrak{q}$ & Quantisation map & \eqref{e:quantummorphism} \\
         $\mathbb{V}$ & Space of operators quantising $\FF$ &  \eqref{e:quantummorphism}\\
         $\mathbb{\Omega}$ & Quantisation of $\mathbb{S}$ & \eqref{e:conditions_Omega_descends} \\
         $\check{\Omega}_C$ & Quantisation of $\check{S}_C$ & \eqref{e:conditions_Omega_descends}\\
         \hline
    \end{tabular}
    \vspace{3mm}
    \caption{Quantum double BFV}
    \label{tab:QuantumdoubleBFV}
\end{table}

\begin{table}[H]
    \centering
    \renewcommand{\arraystretch}{1.2}
     \begin{tabular}{|c|c|c|}
     \hline
        $g$ & element of $S_{nd}^2(T\Sigma)$ (co)-metric on $\Sigma$ & \eqref{e:boundaryfields_EH}  \\
        $\txi{}$ & element of $\mathfrak{X}[1](\Sigma)$, ghost parametrizing diffeomeorphisms & \eqref{e:boundaryfields_EH} \\
        $\txi{n}$ & element of $C^\infty[1](\Sigma)$, ghost  & \eqref{e:boundaryfields_EH} \\
        $\Pi$ & element of $S^2(T^*\Sigma)$ momentum of $g$ & \eqref{e:fiberfields_EH}\\
        $\phip$ &  element of $\Omega^2[-1](\Sigma)$ momentum of $\txi{}$ & \eqref{e:fiberfields_EH}\\
        $\phin$ & element of $C^\infty[-1](\Sigma)$ momentum of $\txi{n}$ & \eqref{e:fiberfields_EH}\\
        $R^\Sigma$ & Ricci scalar of $g^{-1}$ & \eqref{Hamiltonianconstraint}\\
        $L_{\txi{}}$ & Lie Derivative w.r.t. $\txi{}$ & \\
        \hline
    \end{tabular}
    \vspace{3mm}
    \caption{Einstein Hilbert theory, BFV formulation}
    \label{tab:EH_BFV}
\end{table}

\begin{table}[H]
    \centering
    \renewcommand{\arraystretch}{1.2}
     \begin{tabular}{|c|c|c|}
     \hline
        $A$ & element of $\Omega^1(\Sigma,\wedge^2\mathcal{V})$ (connection) & \eqref{e:boundaryfields_EH}\\
        $B$ & element of $\Omega^1_{nd}(\Sigma, \mathcal{V})$ (coframe) &\eqref{e:boundaryfields_EH}\\
        $\chi$ & element of $\Omega^0[1]( \Sigma, \wedge^2\mathcal{V})$ (ghost) &\eqref{e:boundaryfields_EH}\\
        $\tau$ & element of $\Omega^0[1]( \Sigma, \wedge^1\mathcal{V})$ (ghost) &\eqref{e:boundaryfields_EH}\\
        $A^\dag$ & element of  $\Omega^2[-1]( \Sigma, \wedge^1\mathcal{V})$ (antifield)& \eqref{e:boundaryfields_EH}\\
        $B^\dag$ & element of $\Omega^2[-1]( \Sigma, \wedge^2\mathcal{V})$ (antifield) &\eqref{e:boundaryfields_EH}\\
        $F_A$ & curvature of $A$ & \\
        $\mathcal{C}_{BF}$ & Coisotropic submanifold of $\mathcal{F}_{BF}$ &\eqref{e:coisotropic_submanifold_BF}\\
        $\mathcal{C}_{\mathrm{Gauss}}$ & Coisotropic submanifold of $\mathcal{F}_{BF}$  & \eqref{e:gaussian_coiso_BF}\\
        \hline
    \end{tabular}
    \vspace{3mm}
    \caption{BF theory, BFV formulation}
    \label{tab:BF_BFV}
\end{table}

\begin{table}[H]
    \centering
    \renewcommand{\arraystretch}{1.2}
     \begin{tabular}{|c|c|c|}
     \hline
$\te$ & element of $\Omega^1_{nd}(\Sigma, \mathcal{V})$ (coframe) & Proposition \ref{prop:BFVdata}\\
$\tom$ & element of $\Omega^1(\Sigma,\wedge^2\mathcal{V})$ (connection) & Proposition \ref{prop:BFVdata}\\
 $\tc$ & element of $\Omega^0[1](\Sigma,\wedge^2\mathcal{V})$ (ghost)  & Proposition \ref{prop:BFVdata}\\ 
 $\xi$ & element of $\mathfrak{X}[1](\Sigma)$ (ghost) & Proposition \ref{prop:BFVdata}\\$\xi^n$ & element of  $C^\infty[1](\Sigma)$ (ghost) & Proposition \ref{prop:BFVdata}\\ 
$\tom^\dag$ & element of $\Omega^3[-1](\Sigma,\mathcal{V})$ (antifield) & Proposition \ref{prop:BFVdata}\\
$\te^\dag$ & element of $\Omega^{2}[-1](\Sigma,\wedge^2 \mathcal{V})$ (antifield) & Proposition \ref{prop:BFVdata}\\
$F_{\tom}$ & curvature of $\tom$ & \\
$\mathcal{C}_{PC}$ & Coisotropic submanifold of $\mathcal{F}_{PC}$ &Proposition \ref{prop:BFVdata}\\
$\mathcal{C}_{\mathrm{Lor}}$ & Coisotropic submanifold of $\mathcal{F}_{PC}$ & Theorem \ref{t:RPS_PCandEH}\\
$\psi$ & symplectomorphism between BF and PC theory & \eqref{e:PCtoBF}\\
\hline
    \end{tabular}
    \vspace{3mm}
    \caption{Palatini--Cartan theory, BFV formulation}
    \label{tab:PC_BFV}
\end{table}

\section*{Declarations}
\subsubsection*{Conflict of interest statement} The authors have no competing interests to declare that are relevant to the content of this article.
\subsubsection*{Data availability statement} Data sharing not applicable to this article as no datasets were generated or analysed during the current study.

\begin{refcontext}[sorting=nyt]
  \printbibliography[] 
\end{refcontext}

\end{document}